\documentclass{article}

\usepackage{tabto}

\usepackage[a4paper,margin=3.5cm]{geometry}

\usepackage{thmtools}
\usepackage{hyperref}
\usepackage{mathtools}
\usepackage{amssymb,amsmath,amsthm,amsfonts}
\usepackage{hyperref}
\usepackage{cleveref}
\usepackage{comment}

\usepackage{xspace,xpunctuate}

\usepackage{calrsfs}
\DeclareMathAlphabet{\pazocal}{OMS}{zplm}{m}{n}

\newcommand{\awcomplete}[1]{\ensuremath{\mbox{AW$[#1]$-complete}}}

\newcommand{\graphs}{\mathfrak G}
\newcommand{\labelgraphs}{\mathfrak G_{\it lb}}

\def\FOMCG{$p$-$\rm MC(FO,\graphs)$\xspace}
\def\FOMCGL{$p$-$\rm MC(FO,\labelgraphs)$\xspace}
\def\FOMCX{$p$-$\rm MC(FO,{\cal G})$\xspace}

\newcommand{\N}{\ensuremath{\mathbf{N}}\xspace} 
\newcommand{\Q}{\ensuremath{\mathbf{Q}}\xspace} 
\renewcommand{\cal}{\pazocal}

\def\problem#1Input:#2Parameter:#3Problem:#4\par{%
\smallskip
\noindent
\framebox{\parindent=0pt\vbox{\hsize=0.7\hsize
#1\par\smallskip
\halign{\hfil\quad\it##&\quad##\hfil\cr
Input:&#2\cr
Parameter:&#3\cr
Problem:&#4\cr
}}}\par\smallbreak}

\newcommand{\fummler}{tie\xspace}
\newcommand{\fummlers}{ties\xspace}
\newcommand{\Fummler}{Tie\xspace}
\newcommand{\Fummlers}{Ties\xspace}
\newcommand{\excess}{edge-excess\xspace}
\def\mom{\tilde d}
\def\scale{\alpha}
\def\momn{\mom_\scale(n)}
\def\hatdegree{\hat d_\scale(n)}
\def\t{{b}}
\def\y{{\mu}}
\newcommand{\partition}{$\t$-$r$-$\y$-partition\xspace}

\newcommand{\partitionable}{$\t$-$r$-$\y$-partitionable\xspace}
\newcommand{\Partitionable}{$\boldsymbol \t$-$\boldsymbol r$-$\boldsymbol \y$-Partitionable\xspace}
\newcommand{\nhoodpartition}{$\t$-$r$-$\y$-local-protrusion-partition\xspace}
\newcommand{\nhoodpartitions}{$\t$-$r$-$\y$-local-protrusion-partitions\xspace}
\newcommand{\nhoodpartitionable}{$\t$-$r$-$\y$-locally-protrusion-partitionable\xspace}
\newcommand{\Nhoodpartitionable}{$\boldsymbol \t$-$\boldsymbol r$-$\boldsymbol \y$-Locally-Protrusion-Partitionable\xspace}
\newcommand{\OOnhoodpartitionable}{$O(\y^{17}r^3\t)$-$r$-$O(\y)$-locally-protrusion-partitionable\xspace}
\newcommand{\OOnhoodpartition}{$O(\y^{17}r^3\t)$-$r$-$O(\y)$-local-protrusion-partition\xspace}

\newcommand{\dominated}{power-law-bounded\xspace}
\newcommand{\Dominated}{Power-Law-Bounded\xspace}
\newcommand{\domination}{power-law-boundedness\xspace}
\newcommand{\Domination}{Power-Law-Boundedness\xspace}

\newcommand{\Gn}{\cal G_n}
\newcommand{\Gnn}{(\cal G_n)_{n \in \N}}

\def\Nesetril{Ne\v{s}et\v{r}il\xspace}

\def\Erdos{Erd\H{o}s\xspace}
\def\Renyi{R\'{e}nyi\xspace}
\def\ErRe{\Erdos--\Renyi}
\def\ErReGrs{\ErRe graphs\xspace}

\def\eg{e.g\xperiod}

\def\aas{a.a.s\xperiod}

\newcommand{\topnab}{\mathop{\widetilde \triangledown}}
\def\topgrad_#1{\widetilde \nabla\!_{#1}}

\renewcommand{\P}{\Pr}
\newcommand{\E}{\operatorname{E}}

\newtheorem{counter}{counter}[section]

\newtheorem{innercustomcol}{Corollary}
\newenvironment{customcol}[1]
  {\renewcommand\theinnercustomcol{#1}\innercustomcol}
  {\endinnercustomcol}

\newtheorem{innercustomthm}{Theorem}
\newenvironment{customthm}[1]
  {\renewcommand\theinnercustomthm{#1}\innercustomthm}
  {\endinnercustomthm}

\newtheorem{proposition}[counter]{Proposition}

\newtheorem{lemma}[counter]{Lemma}
\newtheorem{theorem}[counter]{Theorem}
\newtheorem{corollary}[counter]{Corollary}

\theoremstyle{definition}
\newtheorem{definition}[counter]{Definition}

\newtheorem{innercustomdef}{Definition}
\newenvironment{customdef}[1]
  {\renewcommand\theinnercustomdef{#1}\innercustomdef}
  {\endinnercustomdef}

\title{First-Order Model-Checking in\\Random Graphs and Complex Networks%
\thanks{A short version of this paper appeared in the Proceedings of
the 28th Annual European Symposium on Algorithms (ESA 2020).}}

\author{Jan Dreier, Philipp Kuinke, Peter Rossmanith \\
    \small Theoretical Computer Science, RWTH Aachen University \\
    \small \texttt{\string{dreier,kuinke,rossmani\string}@cs.rwth-aachen.de}
    }

\date{}

\begin{document}

\maketitle

\kern-10pt
\begin{abstract}
Complex networks are everywhere. 
They appear for example in the form of
biological networks, social networks, or computer networks
and have been studied extensively.
Efficient algorithms to solve problems on complex networks
play a central role in today's society. 
Algorithmic meta-theorems show that many problems can be solved efficiently.
Since logic is a powerful tool to model problems,
it has been used to obtain very general meta-theorems.
In this work, we consider all problems definable in first-order logic
and analyze which properties of complex networks allow them to be solved efficiently.

The mathematical tool to describe complex networks are random graph models.
We define a property of random graph models called $\scale$-\domination.
Roughly speaking, 
a random graph is $\scale$-\dominated if it does not
admit strong clustering and its degree sequence is
bounded by a power-law distribution with exponent
at least $\scale$ (i.e.\ the fraction of vertices with degree $k$ is
roughly $O(k^{-\scale})$).

We solve the first-order model-checking problem (parameterized by the length of the formula) in almost linear FPT time
on random graph models satisfying this property with $\scale \ge 3$.
This means in particular that one can solve 
every problem expressible in first-order logic
in almost linear expected time on these random graph models.
This includes for example preferential attachment graphs,
Chung--Lu graphs, configuration graphs, and sparse \ErRe graphs.
Our results match known hardness results
and generalize previous tractability results on this topic.

\end{abstract}


\section{Introduction}


Complex networks, as they occur in society, biology and technology, play a central role in our everyday lives.
Even though these networks occur in vastly different contexts, 
they are structured and evolve according to a common set of underlying principles.
Over the last two decades, with the emergence of the field of network science,
there has been an explosion in research to understand these fundamental laws.
%
One well observed property is the \emph{small-world phenomenon},
which means that distances between vertices are very small.
This has been verified for the internet
and many other networks~\cite{albert1999internet,milgram1967small}.
Furthermore, many real networks tend to be \emph{clustered}. 
They contain groups of vertices that are densely connected~\cite{schaeffer2007graph}.
If two vertices share a common neighbor, then there is a high chance
that there is also an edge between them.
A network can be considered clustered
if the ratio between the number of triangles and the number of paths with three
vertices is non-vanishing.
This is formalized by the clustering coefficient,
which is high for many networks~\cite{watts1998collective}.
A third important property is a \emph{heavy tailed degree distribution}.
While most vertices have a low number of connections, there are a few
hubs with a high degree.
Experiments show that the degrees follow for example a
power-law or log-normal distribution.
In a power-law distribution, the fraction of vertices with degree $k$ is proportional to
$k^{-\scale}$ (usually with $\scale$ between 2 and 3).
This behavior makes complex networks highly
inhomogeneous~\cite{prvzulj2007biological,mislove2007measurement,broido2018scale,clauset2009power}.\looseness-1

One important goal of theoretical computer science has always been to explore what kinds of inputs allow or forbid us to construct efficient algorithms.
In this context, \emph{algorithmic meta-theorems}~\cite{kreutzer2008algorithmic} are of particular interest.
They are usually theorems stating that
problems definable in a certain logic can be solved efficiently on 
graph classes that satisfy certain properties.
Logic is a powerful tool to model problems and therefore has been used to obtain
very general meta-theorems.
A well-known example is Courcelle's theorem~\cite{Cou90}, which
states that every problem expressible in counting monadic second-order logic
can be solved in linear time on graph classes with bounded treewidth.
It has been further generalized to graph classes with bounded cliquewidth~\cite{CourcelleMR2000}.
To obtain results for larger graph classes one has to consider weaker logics.
The languages of relational database systems are based on first-order logic.
In this logic, one is allowed to quantify over vertices
and to test equality and adjacency of vertices.
With $k$ existential quantifiers, one may ask for the existence of a fixed graph with $k$ vertices ($k$-subgraph isomorphism),
a problem relevant to motif-counting~\cite{milo2002network,countmotif}.
On the other hand,
connectivity properties cannot be expressed in first-order logic.
We define for every graph class $\cal G$
the parameterized first-order model-checking problem \FOMCX~\cite{grohe2008logic}.
\begin{center}
    \problem \FOMCX
Input: A graph $G\in\cal G$ and a first-order sentence $\varphi$
Parameter: The number of symbols in $\varphi$, denoted by $|\varphi|$
Problem: Does $\varphi$ hold on $G$ (i.e.\ $G \models \varphi$)?

\end{center}

The aim is to show for a given graph class $\cal G$ that \FOMCX is
\emph{fixed parameter tractable} (FPT),
i.e., can be decided in time $f(|\varphi|)n^{O(1)}$ for some function~$f$
(see for example \cite{cygan2015parameterized} for an introduction to fixed parameter tractability).
Since input graphs may be large, a linear dependence on $n$ is desirable.
If one is successful, then every problem expressible in first-order logic can be solved on $\cal G$ in linear time.

For the class of all graphs $\graphs$, \FOMCG is \awcomplete{*}~\cite{downey1996parameterized}
and therefore most likely not fpt.
Over time, tractability of \FOMCX has been shown for more and more sparse graph classes $\cal G$:
bounded vertex degree~\cite{seese1996linear}, forbidden
minors~\cite{flum2001fixed}, bounded local
treewidth~\cite{flum2002query}, and further generalizations~\cite{
dawar2007locally,dvorak2010deciding,DBLP:conf/pods/SchweikardtSV18}. 
Grohe, Kreutzer and Siebertz prove that \FOMCX 
can be solved in almost linear FPT time $f(|\varphi|,\varepsilon) n^{1 + \varepsilon}$ for all $\varepsilon > 0$
if $\cal G$ is a nowhere dense graph class~\cite{grohe2017deciding}.
On the other hand if $\cal G$ is a monotone somewhere dense graph class,
\FOMCX is $\rm AW[\ast]$-hard~\cite{grohe2017deciding}.
Nowhere dense graph classes were introduced by \Nesetril
and Ossona de Mendez as those graph classes where for every $r \in \N$
the size of all $r$-shallow clique minors of all graphs in the graph class is bounded by a function of $r$
(\Cref{sec:prelim:sparsity}).
A graph class is somewhere dense if it is not nowhere dense.
The tractability of the model-checking problem on monotone graph classes
is completely characterized with a dichotomy between nowhere dense and somewhere dense graph classes.
These very general results come at a cost:
Frick and Grohe showed that the dependence of the run time on $\varphi$ is
non-elementary~\cite{frick2004complexity}.
We want to transfer this rich algorithmic theory to complex networks.
But what is the right abstraction to describe complex networks?




Network scientists observed that
the chaotic and unordered structure of real networks can by captured using \emph{randomness}.
There is a vast body of research using random processes
to create graphs that mimic the fundamental properties of complex networks.
The most prominent ones are
the preferential attachment model~\cite{barabasi1999emergence,price1976general}, 
Chung--Lu model~\cite{chung2002average,chung2002connected}, 
configuration model~\cite{molloy1995critical,MR98},
Kleinberg model~\cite{kleinberg2000small,kleinberg2000navigation},
hyperbolic graph model~\cite{krioukov2010hyperbolic},
and random intersection graph model~\cite{karonski1999random,rybarczyk2011diameter}.
All these are random models.
It has been thoroughly analyzed how well they predict
various properties of complex networks~\cite{goldenberg2010survey}.

When it comes to algorithmic meta-theorems on random graph models
``even the most basic questions are wide open,'' as Grohe puts it~\cite{grohe2008logic}.
By analyzing which models of complex networks and which values of the model-parameters
allow for efficient algorithms, we aim to develop an
understanding how the different properties of complex networks control their
algorithmic tractability.

In this work we show for a wide range of models, including the well known preferential attachment model,
that one can solve the parameterized first-order model-checking problem in almost linear FPT time.
This means in particular that one can solve 
every problem expressible in first-order logic
efficiently on these models.
%
Our original goal was to obtain efficient algorithms only
for preferential attachment graphs,
but we found an abstraction that transfers these results to
many other random graph models.
Roughly speaking,
the following two criteria are sufficient for 
efficiently solving first-order definable problems
on a random graph model:
\begin{itemize}
    \item
    The model needs to be unclustered.
    In particular the expected number of triangles needs to be subpolynomial.
    \item
    For every $k$, the fraction of vertices with degree $k$ is roughly
    $O(k^{-3})$.
    In other words, the degree sequence needs to be bounded by a power-law distribution with exponent~$3$ or higher.
\end{itemize}
Models satisfying these properties include
sparse \ErRe graphs, preferential attachment graphs as well as certain Chung--Lu and configuration graphs.
On the other hand, the Kleinberg model, the hyperbolic random graph model,
or the random intersection graph model
do not satisfy these properties.
Our results generalize previous results~\cite{grohe2001generalized,StrucSpars}
and match known hardness results:
The model-checking problem has been proven to be hard on
power-law distributions with exponent smaller than
$3$~\cite{averagehardness}.
We therefore identify the threshold for tractability to be a power-law coefficient of~$3$.
It is also a big open question whether the model-checking problem can
also be solved on clustered random graph~models, especially since
real networks tend to be clustered.
Furthermore, significant engineering challenges need to be overcome
to make our algorithms applicable in practice.

\subsection{Average Case Complexity}
Average-case complexity analyzes the typical run time of algorithms on random
instances (see~\cite{bogdanov2006average} for a survey), based
on the idea that a worst-case analysis often is too pessimistic
as for many problems hard instances occur rarely in the real world. 
Since models of complex networks are probability distributions over graphs,
we analyze the run time of algorithms under average-case complexity.
However, there are multiple notions
and one needs to be careful which one to choose.


Assume a random graph model is asymptotically almost surely (\aas) nowhere dense, i.e.,
a random graph from the model with $n$ vertices
belongs with probability $1-\delta(n)$ to a nowhere dense graph class,
where $\lim_{n \to \infty}\delta(n) = 0$ (\Cref{sec:prelim:sparsity}).
Then the first-order model-checking problem can be efficiently solved with a probability converging to one~\cite{grohe2017deciding}.
However, with probability $\delta(n)$ the run time can be arbitrarily high
and the rate of convergence of $\delta(n)$ to zero can be arbitrarily slow.
These two missing bounds are undesirable from an algorithmic standpoint
and the field of average-case complexity has established a theory
on how the run time needs to be bounded with respect to the
fraction of inputs that lead to this run time.

This is formalized by the well-established notion of 
\emph{average polynomial run time}, introduced by Levin~\cite{levin1986average}. 
An algorithm has
average polynomial run time with respect to a random graph model if there is
an $\varepsilon > 0$ and a polynomial $p$ such that for every $n,t$ the
probability that the algorithm runs longer than $t$ steps on an input of size
$n$ is at most $p(n)/t^\varepsilon$.
This means there is a
polynomial trade-off between run time and fraction of inputs. 
This notion has been widely studied~\cite{bogdanov2006average,arora2009computational}
and is considered from a complexity theoretic standpoint the right notion
of polynomial run time on random inputs. 
It is closed under invoking polynomial subroutines.

In our work, however, we wish to explicitly distinguish linear time.
While Levin's complexity class is a good analogy to the class~P,
it is not suited to capture algorithms with average linear run time.
For this reason, we turn to the expected value of the run time,
a stronger notion than average polynomial time.
In fact, using Markov's inequality we see that if an algorithm has expected linear run time,
all previous measures of average tractability are also bounded.
Their relationship is as follows.
$$
\text{expected linear}
~~\Rightarrow~~
\text{expected polynomial}
~~\Rightarrow~~
\text{average polynomial}
~~\Rightarrow~~
\text{\aas polynomial}
$$



With this in mind we can present
our notion of algorithmic tractability.
A labeled graph is a graph where every vertex can have
(multiple) labels.
First-order formulas can have unary predicates for each type of label.
These predicates test whether a vertex has a label of a certain type.
We define $\graphs$ to be the class of all graphs,
and $\labelgraphs$ to be the class of all vertex-labeled graphs.
A function $L \colon \graphs \to \labelgraphs$ is an
\emph{$l$-labeling function} for $l \in \N$ if for every $G \in \graphs$,
$L(G)$ is a labeling of $G$ with up to $l$ classes of labels
(see \Cref{sec:prelim} for details).
Furthermore, a \emph{random graph model} is a sequence
${\cal G} = (\cal G_n)_{n \in \N}$,
where $\Gn$ is a probability distribution over
unlabeled simple graphs with $n$ vertices.


\begin{customdef}{\ref{def:expectedFPT}}
    We say \FOMCGL can be decided
    on a random graph model $\Gnn$ in \emph{expected time} $f(|\varphi|,n)$ if
    there exists a deterministic algorithm $\cal A$ which decides \FOMCGL on input $G$,
    $\varphi$ in time $t_{\cal A}(G,\varphi)$ and
    if for all $n \in \N$, all first-order sentences $\varphi$ and all $|\varphi|$-labeling functions $L$,
    $
    \E_{G \sim \Gn}\bigl[ t_{\cal A} (L(G),\varphi) \bigr] \le f(|\varphi|,n).
    $
    We say \FOMCGL on a random graph model can be decided in \emph{expected FPT time} 
    if it can be decided in expected time $g(|\varphi|)n^{O(1)}$ for some function~$g$.
\end{customdef}

In particular, this definition implies efficient average run time according to Levin's
notion (which is closed under polynomial subroutines).
We choose to include labels into our notion of
average-case hardness for two reasons:
First, it makes our algorithmic results stronger,
as the expected run time is small, even 
in the presence of an adversary that labels the vertices
of the graph.
Secondly, it matches known hardness results that require adversary labeling.

\subsection{Previous Work}

There have been efforts to transfer the results for classical 
graph classes to random graph models by showing that a graph sampled
from some random graph model belongs with high probability to a certain
algorithmically tractable graph class.
For most random graph models the treewidth is polynomial in the size of the graph~\cite{gao2012treewidth,blasius2016hyperbolic}.
Therefore, people have considered more permissive graph measures than treewidth,
such as low degree~\cite{grohe2001generalized}, or bounded expansion \cite{StrucSpars,farrell2015hyperbolicity}.
Demaine et al.\ showed that some Chung--Lu and configuration graphs have bounded expansion and 
provided empirical evidence that some real-world networks do too~\cite{StrucSpars}.
However, this technique is still limited,
as many random graph models (such as the preferential attachment model~\cite{StrucSpars,cliqueminors})
are not known to be contained in any of the well-known tractable graph classes.

The previous tractability results presented in this section all use the following technique:
Assume we have a formula $\varphi$ and sample a graph of size $n$ from a random graph model.
If the sampled graph belongs to the tractable graph class,
an efficient model-checking algorithm for the graph class can
solve the instance in FPT time.
If the graph does not belong the graph class,
the naive model-checking algorithm can still solve the instance in time $O(n^{|\varphi|})$.
Assume we can show that the second case only happens with probability $\delta(n)$
converging to zero faster than any polynomial.
Then $\delta(n)O(n^{|\varphi|})$ converges to zero and
the expected run time remains bounded by an FPT function.

Let $p(n)$ be a function with $p(n) = O(n^{\varepsilon}/n)$ for all $\varepsilon > 0$.
Grohe showed that one can solve \FOMCGL on \ErRe graphs $G(n,p(n))$
in expected time $f(|\varphi|,\varepsilon)n^{1+\varepsilon}$ for every
$\varepsilon > 0$~\cite{grohe2001generalized}.
This result was obtained by showing that with high probability
the maximum degree of the random graph model is $O(n^\varepsilon)$ for every $\varepsilon > 0$
and then using a model-checking algorithm for low degree graphs.
Later Demaine et al.\ 
and Farrell et al.\ showed that
certain Chung--Lu and configuration graphs
whose degrees follow a power-law distribution with exponent $\alpha > 3$~\cite{StrucSpars} 
as well as certain random intersection graphs~\cite{farrell2015hyperbolicity} belong with high
probability to a graph class with bounded expansion.
While they do not mention it explicitly, the previous argument implies
that one can solve \FOMCGL in expected
time $f(|\varphi|)n$ on these random graph models.

There further exist some average-case hardness results for the model-checking problem.
It has been shown that one cannot decide \FOMCGL
on \ErRe graphs $G(n,1/2)$ or $G(n,p(n))$ with $p(n) = n^\varepsilon/n$ for some $0 < \varepsilon < 1$, $\varepsilon \in \Q$,
in expected FPT time (unless $\rm AW[*] \subseteq FPT/poly$)~\cite{averagehardness}.
The same holds for Chung--Lu graphs with exponent $2.5 < \alpha < 3$, $\alpha \in \Q$.
These hardness results fundamentally require the adversary labeling of \Cref{def:expectedFPT}.
It is a big open question whether they can be transferred
to model-checking without labels.




Another thing to keep in mind when considering logic and random graphs~\cite{spencer2013strange}
are \emph{zero-one laws}.
They state that in many \ErRe graphs
every first-order formula holds in the limit either with probability zero or one~\cite{spencer2013strange,glebskii1969range,fagin1976probabilities}.
Not all random graph models satisfy a zero-one law for first-order logic
(e.g.\ the limit probability of the existence of a $K_4$ in 
a Chung--Lu graph with weights $w_i = \sqrt{n/i}$ is neither zero nor one).

\section{Our Results}

We define a property called \emph{$\scale$-\domination}.
This property depends on a parameter $\scale$ and captures many unclustered random graph models
for which the fraction of vertices with expected degree $d \in \N$ is roughly $O(d^{-\scale})$.
Our main contribution is solving the model-checking problem efficiently
on all $\scale$-\dominated random graph models with $\scale \ge 3$.
This includes preferential attachment graphs, Chung--Lu graphs, \ErRe graphs
and other random graph models.
Note that graphs do not need to have a power-law degree distribution to be $\scale$-\dominated.
Our results hold for arbitrary labelings of the random graph and
are based on a novel decomposition technique for local regions of random graphs.
While all previous algorithms work by placing the
random graph model with high probability in a sparse graph class,
our technique also works for some \aas somewhere dense random graphs
(e.g.\ preferential attachment graphs~\cite{cliqueminors}).

\subsection{\Domination}

We start by formalizing our property.
Since it generalizes the Chung--Lu model, we define this model first.
A Chung--Lu graph with exponent $\scale$ and vertices $v_1,\dots,v_n$ is defined such that two vertices $v_i$ and $v_j$
are adjacent with probability $\Theta(w_{i} w_{j} /n)$ where $w_i = (n/i)^{1/(\scale-1)}$~\cite{chung2002average}.
Furthermore all edges are independent,
which means that the probability
that a set of edges occurs equals the product over the probabilities
of each individual edge.
In our model the probability of a set of edges can be a certain
factor larger than the product of the individual probabilities,
which allows edges to be moderately dependent.

\begin{definition}\label{def:wellBehavedPowerLaw}
    Let $\scale > 2$.
    We say a random graph model $\Gnn$ is $\scale$-\dominated
    if for every $n \in \N$ there exists
    an ordering $v_1,\dots,v_n$ of $V(\Gn)$
    such that for all $E \subseteq {\{v_1,\dots,v_n \} \choose 2}$
    $$
    \P\bigl[E \subseteq E(\Gn) \bigr] \le \\
    \prod_{v_iv_j \in E} 
    \frac{(n/i)^{1/(\scale-1)}(n/j)^{1/(\scale-1)}}{n}
    \cdot
    \begin{cases}
        2^{O(|E|^2)} &\text{if }\scale > 3 \\
        \log(n)^{O(|E|^2)} &\text{if }\scale = 3 \\
        O(n^\varepsilon)^{|E|^2} \text{ for every $\varepsilon > 0$ }
	&\text{if }\scale < 3.
    \end{cases}
    $$
\end{definition}

The probability that a set of edges $E$ occurs may be up to
a factor $2^{O(|E|^2)}$ or $\log(n)^{O(|E|^2)}$ or $O(n^\varepsilon)^{|E|^2}$ (depending on $\scale$)
larger than the probability in the corresponding Chung--Lu graph.
For conditional probabilities this means the following:
The probability bound for an edge under
the condition that some set of $l$ edges is already present may be up 
to a factor $2^{O(l)}$ or $\log(n)^{O(l)}$ or $O(n^\varepsilon)^{l}$ larger than the unconditional probability.
This lets \dominated random graphs capture moderate dependence between edges.
The factor undergoes a phase transition at $\scale = 3$.
The smaller factor $2^{O(|E|^2)}$ for $\scale > 3$ was chosen to guarantee
linear FPT run time of our model-checking algorithm (\Cref{thm:MCpowerlaw2}) if $\scale > 3$.
The slightly larger factor of $\log(n)^{O(|E|^2)}$ for $\scale = 3$
was chosen to capture preferential attachment graphs
while still maintaining a quasilinear FPT run time of our algorithm.

The parameter $\scale$ of an $\scale$-\dominated random graph model controls the degree distribution.
Note that if a graph class is $\scale$-\dominated it is also
$\scale'$-\dominated for all $2 < \scale'<\scale$.
It can be easily seen that a vertex $v_i$ has expected degree at most $O(n^\varepsilon) (n/i)^{1/(\scale-1)}$ for every $\varepsilon > 0$. 
This means the expected degree sequence of an $\scale$-\dominated random graph model
is not power-law distributed with exponent smaller than $\scale$.
The gap is often tight:
For example, Chung--Lu graphs with a power-law degree distribution exponent $\scale$
are $\scale$-\dominated
and preferential attachment graphs have a power-law degree distribution with exponent $3$ and are $3$-\dominated.
For the interesting case $\scale=3$, the inequality in Definition~\ref{def:wellBehavedPowerLaw}
simplifies to
$$
\P\bigl[E \subseteq E(\Gn)\bigr] \le \log(n)^{O(|E|^2)} \prod_{v_iv_j \in E} \frac{1}{\sqrt{ij}}.
$$

\subsection{Model Checking}

We now present our model-checking algorithm for $\scale$-\dominated graphs.
We express its run time relative to the term
$$
    \momn =
    \begin{cases}
        O(1) &\quad \scale > 3 \\
        \log(n)^{O(1)} &\quad \scale = 3 \\
        O(n^{3-\scale}) &\quad \scale < 3.
    \end{cases}
$$
This term is related to an established property of degree distributions,
namely the \emph{second order average degree}~\cite{chung2002average}.
If a random graph with $n$ vertices has expected degrees $w_1,\dots,w_n$
then the second order average degree is defined as $\sum_{i=1}^n w_i^2/\sum_{k=1}^n w_k$.
In graphs with a power-law degree distribution $\scale$ we have $w_i = \Theta((n/i)^{1/(\scale-1)})$.
The second order average degree then equals $\Theta\bigl(\sum_{i=1}^n (n/i)^{2/(\scale-1)}/\sum_{k=1}^n (n/k)^{1/(\scale-1)}\bigr)$.
For $\scale > 3$, this term is constant,
for $\scale = 3$ it is logarithmic,
and for $\scale < 3$ it is polynomial in $n$~\cite{chung2002average}.
Thus, we can interpret $\momn$ as an estimate of the second order average degree.
We prove that the model-checking problem can be solved efficiently if $\momn$ is small.
\begin{customthm}{\ref{thm:MCpowerlaw1}}
    There exists a function $f$ such that one can solve \FOMCGL 
    on every $\scale$-\dominated random graph model 
    in expected time $\momn^{f(|\varphi|)} n$.
\end{customthm}

The term $\momn$ naturally arises in our proofs
and is not a consequence of how we defined the multiplicative factor
(i.e., $2^{O(|E|^2)}$, $\log(n)^{O(|E|^2)}$, $O(n^\varepsilon)^{|E|^2}$)
in \Cref{def:wellBehavedPowerLaw}.
In fact the dependence goes the other way: We defined the factor for each $\scale$
as large as possible such that it does not dominate the run time of the algorithm.
Next we specify exactly those values of $\scale$ where the previous theorem leads to FPT run times.
(In the third case $\varepsilon > 0$ can be chosen arbitrarily small
since we require $\scale$ to be arbitrarily close to $3$.)

\begin{customthm}{\ref{thm:MCpowerlaw2}}
        Let $\Gnn$ be a random graph model.
        There exists a function $f$ such that one can solve \FOMCGL in expected time
        \begin{itemize}
        \item $f(|\varphi|)n$                                                              \tabto{4.5cm}  if $\Gnn$ is $\scale$-\dominated for some $\scale > 3$,
        \item $\log(n)^{f(|\varphi|)}n$                                                    \tabto{4.5cm}  if $\Gnn$ is $\scale$-\dominated for $\scale = 3$,
        \item $f(|\varphi|,\varepsilon) n^{1 + \varepsilon}$ for all $\varepsilon > 0$     \tabto{4.5cm}  if $\Gnn$ is $\scale$-\dominated for every $2 < \scale < 3$.
        \end{itemize}
\end{customthm}

This solves the model-checking problem efficiently on a wide range of random graph models.
These tractability results are matched by previous intractability results.
(Note that the third case of \Cref{thm:MCpowerlaw2} requires \domination for
\emph{every} $2 < \scale < 3$ and thus does not contradict \Cref{prop:blub}.)

\begin{proposition}[\cite{averagehardness} and Lemma~\ref{lem:chungludominated}]\label{prop:blub}
    For every $2 < \scale < 3$
    there exists an $\scale$-\dominated random graph model $\Gnn$
    such that one cannot solve \FOMCGL on $\Gnn$
    in expected {\rm FPT} time unless $\rm AW[*] \subseteq FPT/poly$.
\end{proposition}

We observe a phase transition in tractability at power-law exponent $\alpha=3$.
Also the run time of our algorithm cannot be linear in $n$ for $\scale \le 3$
as a $3$-\dominated random graph can have for example $n\log(n)$ edges in expectation.
We discuss the algorithmic implications of our result for some well-known random graph models in \Cref{sec:domination}.

\subsection{Structure}

Many algorithmic results are based on 
\emph{structural decompositions.}
For example, bidimensionality theory introduced by Demaine et
al.~\cite{Demaine:2005:SPA:1101821.1101823,DH08} is based on the grid minor theorem, which is itself based on a
structural decomposition into a clique-sum of almost-embeddable graphs
developed by Robertson and Seymour~\cite{RS03}.
The model-checking algorithm for graph classes with bounded expansion by 
Dvo\v{r}ak, Kr{\'a}l, and Thomas~\cite{dvorak2010deciding}
relies on a structural decomposition of bounded expansion graph classes
by \Nesetril and Ossona de Mendez called low tree-depth colorings~\cite{NOdM08}.
Our algorithms are based on a structural decomposition of $\scale$-\dominated random graph models.

All algorithms prior to this work rely on showing that a certain graph model
is with high probability contained in a certain well-known tractable graph class (for example bounded expansion)
and then use the structural decompositions~\cite{NOdM08} of said graph class.
However, these decompositions were not originally designed with random
graphs in mind and therefore may not provide the optimal level of abstraction
for random graphs.
Our algorithms are based on a specially defined structural decomposition.
This direct approach helps us capture
random graph models that could otherwise not be captured
such as the \aas somewhere dense preferential attachment model.
By focusing on $\scale$-\dominated random graph models,
we obtain structural decompositions for a wide range of models.

We observe that $\scale$-\dominated random graphs have mostly an extremely sparse structure
with the exception of a part whose size is bounded by the second order average degree.
However, this denser part can be separated well from the remaining graph.
We show that local regions consist of a 
core part, bounded in size by the second order average degree,
to which trees and graphs of constant size are attached by a constant number of edges.
This decomposition is similar to so called \emph{protrusion
decompositions}, which have been used by Bodlaender
et al.\ to obtain meta-theorems on kernelization~\cite{bodlaender2016meta}.
Our structural decomposition is valid for all graphs that fit into the framework of $\scale$-\domination,
such as preferential attachment graphs or Chung--Lu graphs.
We define an approximation of the second order average degree of the degree distribution
as
$\hatdegree = 2$ for $\scale > 3$,
$\hatdegree = \log(n)$ for $\scale = 3$
and
$\hatdegree = n^{3 - \scale}$ for $\scale < 3$
(similarly to $\momn$ without $O$-notation).

\begin{customthm}{\ref{thm:aasprotrusion}}
    Let $\Gnn$ be an $\scale$-\dominated random graph model.
    There exist constants $c,r_0$ such that for every $r \ge r_0$
    \aas for every $r$-neighborhood $H$ of $\Gn$ one can 
    partition $V(H)$ into three (possibly empty) sets $X$, $Y$, $Z$ with the following properties.
    \begin{itemize}
        \item $|X| \le \hatdegree^{cr^2}$.
        \item Every connected component of $H[Y]$ has size at most $cr$
            and at most $c$ neighbors in~$X$.
        \item Every connected component of $H[Z]$ is a tree with at most one edge to ${H[X \cup Y]}$.
    \end{itemize}
\end{customthm}
Removing a few vertices makes the local neighborhoods even sparser:
\begin{customcol}{\ref{aastreewidth}}
    Let $\Gnn$ be an $\scale$-\dominated random graph model.
    There exist constants $c,r_0$ such that 
    for every $r \ge r_0$ 
    \aas one can remove $\hatdegree^{cr^2}$ vertices from $\Gn$
    such that every $r$-neighborhood has treewidth at most $26$.
\end{customcol}
\Cref{aastreewidth} is a consequence of \Cref{thm:aasabc} from \Cref{sec:structure}.
Further structural results that 
may be interesting beyond the purpose of model-checking
can be found in \Cref{sec:structure}.
We now discuss how we use the decomposition of Theorem~\ref{thm:aasprotrusion}
for our algorithms and why decompositions similar to Corollary~\ref{aastreewidth}
are not sufficient for our purposes.

\section{Techniques}


A first building block of our algorithm is Gaifman's locality theorem~\cite{gaifman1982local}.
It implies that in order to solve the first-order model-checking problem on a graph,
it is sufficient to solve the problem on all $r$-neighborhoods of the graph for some small~$r$.
We can therefore restrict ourselves to the model-checking problem on the neighborhoods
of random graphs.
With this in mind, we want to obtain structural decompositions of
these neighborhoods.

One important thing to note is that a decomposition according to Corollary~\ref{aastreewidth}
is not sufficient.
Let us focus on the interesting case $\scale=3$ where efficient model-checking is still possible.
Corollary~\ref{aastreewidth} then states that the removal of polylogarithmically
many vertices yields neighborhoods with treewidth at most~26.
While we could easily solve the model-checking problem on graphs with treewidth at most 26
via Courcelle's theorem~\cite{Cou90},
we cannot solve it on graphs where we need to remove a set $X$ of $\log(n)$ vertices
to obtain a treewidth of at most~26.
Every vertex not in $X$ may have an arbitrary subset of $X$ as neighborhood.
Since there are $2^{|X|} = n$ possible neighborhoods,
we can encode a large complicated structure into this graph by
stating that two vertices $i,j \in \N$ are adjacent if and only if
there is a vertex whose neighborhood in $X$ represents a binary encoding of the edge $ij$ (omitting some details).
Because of this, the model-checking problem on this graph class is as hard as
on general graphs.
We need the additional requirement that $X$ is only loosely
connected to the remaining graph.
The decomposition in Theorem~\ref{thm:aasprotrusion} fulfills this requirement.
Every component of $H \setminus X$ has at
most a constant number of neighbors in~$X$.

Let us assume we have decompositions of the neighborhoods of a graph according to Theorem~\ref{thm:aasprotrusion}
where the sets $X$ are chosen as small as possible.
We can now use a variant of the Feferman--Vaught theorem~\cite{FV59} for each $r$-neighborhood
to prune the protrusions and thereby construct a smaller graph that satisfies the same (short) first-order formulas as the original graph,
We call this smaller graph the \emph{kernel}.
The size of this kernel will be some function of~$|X|$.
We then use the brute-force model-checking algorithm on the kernel.



For the first steps of the algorithm (decomposition into neighborhoods, kernelization using Feferman--Vaught)
one can easily show that they always take FPT time.
However, the run time of the last step requires a careful analysis.
One can check a formula $\varphi$ on a graph of size $x$ in time
$O(x^{|\varphi|})$ by brute force.
Thus, checking the formula on the kernel of all $n$ many $r$-neighborhoods of a random graph
takes expected time at most
$
n \sum_{x=1}^n p_x O(x^{|\varphi|}),
$
where $p_x$
is the probability that the kernelization procedure on an $r$-neighborhood of a random graph 
yields a kernel of size $x$.
In order to guarantee a run time of the form $\log(n)^{f(|\varphi|)}n$ for some function $f$,
$p_x$ should be of order $\log(n)^{f(|\varphi|)}x^{-|\varphi|}$.

Earlier, we discussed that the size of the kernel will be some function of $|X|$
and that we choose $X$ as small as possible.
It is therefore sufficient to bound the probability that the set $X$ of the decomposition 
of a neighborhood exceeds a certain size.
Parameterizing the decomposition by two values (denoted by $\t$ and $\y$ later on)
gives us enough control to guarantee such a bound on $p_x$.
A large part of this work is devoted to proving a good trade-off between
the size of the set $X$ of the decomposition
and the probability that $X$ is of minimal size.
Furthermore, computing the set $X$ is computationally hard, so the whole procedure
has to work without knowing the set $X$, but only its existence.



Our proofs are structured as follows.
First, we show in \Cref{sec:partitionprob} that $\scale$-\dominated random graph models
have the following structure with high probability:
They can be partitioned into sets $A$, $B$, $C$, where $A \cup B$ is small,
$B \cup C$ is sparse and $A$ and $C$ locally share only few edges.
This is done by characterizing this structure by a collection of small forbidden edge-sets
and then excluding these edge-sets using the union bound and \Cref{def:wellBehavedPowerLaw}.
Then in \Cref{sec:protrusiondecomp} we show that the partition into $A$, $B$, $C$ implies
the protrusion decomposition of \Cref{thm:aasprotrusion}.
In \Cref{sec:kernel}, we partially recover the protrusion decomposition from a given input,
and use it to kernelize each $r$-neighborhood into an equivalent smaller graph.
At last, in \Cref{sec:modelchecking}, we combine Gaifman's locality theorem
with the previous algorithms and probability bounds to obtain our algorithm and bound its run time.
Some proofs are quite tedious, but the nature of this problem seems to stop us from using simpler methods.
Furthermore, in \Cref{sec:structure} give a simpler presentation of our structural results and 
in \Cref{sec:domination}, we discuss the algorithmic implications of our results for various random graph models.

\section{Notations and Definitions}
\label{sec:prelim}

\subsection{Graph Notation}
We use common graph theory notation~\cite{diestel}.
The \emph{length} of a path equals its number of edges.
The \emph{distance} between to vertices $u$ and $v$ ($\text{dist}(u,v)$)
equals the length of a shortest path between $u$ and $v$.
For a vertex $v$ let $N^G_r(v)$ be the set of vertices which have
in $G$ distance at most $r$ to $v$.
The \emph{radius} of a graph is the minimum among all maximum distances from one vertex to all other vertices.
An \emph{$r$-neighborhood} in $G$ is an induced subgraph of $G$ with radius at most $r$.
The \emph{order} of a graph is $|G| = |V(G)|$.
The \emph{size} of a graph is $\|G\| = |V(G) + E(G)|$.
The \emph{\excess} of a graph $G$ is $|E(G)|-|V(G)|$.

In this work we obtain results for \emph{labeled graphs}~\cite{grohe2008logic}.
A labeled graph is a tuple
$G=(V(G),E(G),P_1(G),\dots,P_l(G))$ with $P_i(G) \subseteq V(G)$.
We call $P_1(G),\dots,P_l(G)$ the \emph{labels} of $G$.
We say a vertex $v$ is labeled with label $P_i(G)$ if $v \in P_i(G)$.
A vertex may have multiple labels.
We say the unlabeled simple graph $G'=(V(G),E(G))$ is the
\emph{underlying graph} of $G$
and $G$ is a \emph{labeling} of $G'$.
All notion for graphs extends to labeled graphs as expected.
The union of two labeled graphs $G$ and $H$,
($G\cup H$), is obtained by setting $V(G \cup H)=V(G)\cup V(H)$,
$E(G \cup H)=E(G)\cup E(H)$ and for each label 
$P_i(G \cup H) = P_i(G) \cup P_i(H)$.

For a graph class $\cal G$, we define $\cal G_{\it lb}$ to be the class of all
labelings of $\cal G$.
We define $\graphs$ to be the class of all simple graphs and $\labelgraphs$
to be the class of all labeled simple graphs.

\subsection{Probabilities and Random Graph Models}
We denote probabilities by $\P[*]$ and expectation by $\E[*]$.
We consider a random graph model to be a sequence of probability
distributions. For every $n \in \bf N$ a random graph model describes
a probability distribution on unlabeled simple graphs with $n$ vertices. 
In order to speak of probability distributions over graphs we
fix a sequence of vertices $(v_i)_{i \ge 1}$
and require that a graph with $n$ vertices has the vertex set
$\{v_1,\dots,v_n\}$.  
A random graph model is a sequence ${\cal G} = (\cal G_n)_{n \in \N}$,
where $\Gn$ is a probability distribution over
all unlabeled simple graphs $G$ with $V(G) = \{v_1,\dots,v_n\}$.
Even though some random processes naturally lead to graphs with multi-edges or self-loops,
we interpret them as simple graphs by removing all self-loops and replacing multiple edges
with one single edge.
In slight abuse of notation, we also write $\Gn$ for the random variable
which is distributed according to $\Gn$.
This way, we can lift graph notation to notation for random variables of graphs:
For example edge sets and neighborhoods of a random graph $\Gn$ are represented 
by random variables $E(\Gn)$ and $N^{\Gn}_r(v)$.

\subsection{Sparsity}
\label{sec:prelim:sparsity}

At first, we define nowhere and somewhere density
as a property of \emph{graph classes}
and then lift the notation to \emph{random graph models}.
There are various equivalent definitions and we use the most common definition based on
shallow topological minors.

\begin{definition}[Shallow topological minor~\cite{NOdM08}]
\label{def:shallowtopminor+}
    A graph $H$ is an $r$-shallow topological minor of $G$ if a graph obtained
    from $H$ by subdividing every edge up to $2r$ times is isomorphic to a
    subgraph of $G$.
    The set of all $r$-shallow topological minors of a graph $G$ is denoted by
    $G \topnab r$.
    We define the maximum clique size over all shallow topological minors of $G$ as
    $$\omega(G \topnab r) = \max_{H\in G \topnab r }
    \omega(H).$$
\end{definition}

\begin{definition}[Nowhere dense~\cite{nevsetvril2012sparsity}]
  A graph class~$\cal G$ is nowhere dense if there exists a function $f$,
  such that for all $r \in \N$ and all $G\in\cal G $, $\omega(G \topnab r ) \leq
  f(r)$.
\end{definition}

\begin{definition}[Somewhere dense~\cite{nevsetvril2012sparsity}]
  A graph class~$\cal G$ is somewhere dense if for all functions $f$ there
  exists an $r \in \N$ and a $G\in\cal G $, such that $\omega( G
  \topnab r ) > f(r)$.
\end{definition}

Observe that a graph class is somewhere dense if and only if it is
not nowhere dense. 
We lift these notions to random graph models using the following two definitions.

\begin{definition}[\aas nowhere dense]
    A random graph model~$\cal G$ is \aas nowhere dense if there exists a
    function $f$ such that for all $r \in \N$
   $$
   \lim_{n\to\infty} \P[\omega( {\cal G}_n \topnab r ) \leq f(r)] = 1.
   $$
\end{definition}
\begin{definition}[\aas somewhere dense]\label{def:aas-somewhere-dense}
    A random graph model~$\cal G $ is \aas somewhere dense if for all functions
    $f$ there is an $r \in N$ such that
    $$
    \lim_{n\to\infty} \P[\omega( {\cal G}_n \topnab r ) > f(r)] = 1.
    $$
\end{definition}
While for graph classes the concepts are complementary,
a random graph model can both be \emph{neither} \aas somewhere dense
\emph{nor} \aas nowhere dense (e.g., if the random graph model is 
either the empty or the complete graph, both with a probability of $1/2$).

\subsection{First-Order Logic}
We consider only first-order logic over labeled graphs.
We interpret a labeled graph $G = (V,E,P_1,\dots,P_l)$, 
as a structure with universe $V$ and signature
$(E, P_1,\dots,P_l)$. The binary relation $E$ expresses adjacency between vertices
and the unary relations $P_1,\dots,P_l$ indicate the labels of the vertices.
Other structures can be easily converted into labeled graphs.
We write $\varphi(x_1,\dots,x_k)$ to indicate that a formula $\varphi$ has \emph{free variables} $x_1,\dots,x_k$.
The \emph{quantifier rank} of a formula is the maximum nesting depth of quantifiers in the formula.
Two labeled graphs $G_1, G_2$ with the same signature are $q$-\emph{equivalent} ($G_1 \equiv_q G_2$) if for every first-order sentence $\varphi$ 
with quantifier rank at most $q$ and matching signature holds $G_1 \models \varphi$ if and only if $G_2 \models \varphi$.
Furthermore, $|\varphi|$ is the number of symbols in $\varphi$.
There exists a simple algorithm which decides whether $G \models \varphi$ in time $O(|G|^{|\varphi|})$.

\subsection{Model-Checking}
With all definitions in place, we can now properly
restate the model-checking problem and what it means
to solve it efficiently on a random graph model.
The model-checking problem on labeled graphs is defined as follows.

\begin{center}
    \problem \FOMCGL
Input: A graph $G\in\cal \labelgraphs$ and a first-order sentence $\varphi$
Parameter: $|\varphi|$
Problem: $G\models \varphi$?

\end{center}

Under worst-case complexity, \FOMCGL is \awcomplete{*}~\cite{downey1996parameterized} (and PSPACE-complete when unparameterized~\cite{stockmeyer1976polynomial}).
We want average case algorithms for \FOMCGL to be efficient for
all possible labelings of a random graph model.
A function $L : \graphs \to \labelgraphs$ is called a
\emph{$l$-labeling function} for $l \in \N$ if for every $G \in \graphs$,
$L(G)$ is a labeling of $G$ with up to $l$ labels.

\begin{definition}\label{def:expectedFPT}
    We say \FOMCGL can be decided
    on a random graph model $\Gnn$ in \emph{expected time} $f(|\varphi|,n)$ if
    there exists a deterministic algorithm $\cal A$ which decides \FOMCGL on input $G$,
    $\varphi$ in time $t_{\cal A}(G,\varphi)$ and
    if for all $n \in \N$, all first-order sentences $\varphi$ and all $|\varphi|$-labeling functions $L$,
    $
    \E_{G \sim \Gn}\bigl[ t_{\cal A} (L(G),\varphi) \bigr] \le f(|\varphi|,n).
    $
    We say \FOMCGL on a random graph model can be decided in \emph{expected FPT time} 
    if it can be decided in expected time $g(|\varphi|)n^{O(1)}$ for some function~$g$.
\end{definition}


%

\section{Structure Theorem for \Dominated Random Graph Models}\label{sec:partitionprob}

\newcommand{\PartitionRadiusExtraEdges}{40 \y r}

\newcommand{\PartitionRadiusToA}{20 \y r}

\newcommand{\embeds}{\sqsubseteq}
\def\HH{\cal H_n(\t,r,\y)}
\def\HHminus{\cal H_n(\t-1,r,\y)}
\def\LL{\cal L_n(\t_1,\t_2,l,k)}

The goal of this section is to partition $\scale$-\dominated random graph models.
We show in Theorem~\ref{thm:partitionProbability}
that their vertices can with high probability be partitioned into sets $A,B,C$
with the following properties: The sets $A$ and $B$ are small, the
graph $G[B \cup C]$ is locally almost a tree, i.e., has locally only a
small \excess,
and the set $B$ almost separates $A$ from $C$,
i.e., every neighborhood in $G[C]$
has only a small number of edges to $A$.
We call $(A,B,C)$ an \partition.
We state the formal definition.

\begin{definition}[\partition]\label{def:partition}
    Let $\t,r,\y \in \N^+$. Let $G$ be a graph.
    A tuple $(A, B, C)$ is called an \partition of $G$ if
    \begin{enumerate}
        \item\label{prop:ABCPartition}
            the sets $A,B,C$ are pairwise disjoint and their union is $V(G)$,
        \item\label{prop:ABCSize}
            $|A| \le \t$ and $|B| \le \t^\y$,
        \item\label{prop:ABCExcess}
            every $\PartitionRadiusExtraEdges$-neighborhood in $G[B
            \cup C]$ has an \excess of at most $\y^2$, and
        \item\label{prop:ABCEdgesToA}
            for every $\PartitionRadiusToA$-neighborhood in $G[C]$ there
            are at most $\y$ edges incident to both the neighborhood and to $A$.
    \end{enumerate}
    A graph for which an \partition exists is called \partitionable.
\end{definition}
%

In summary, $B$ and $C$ are well behaved
and the large set $C$ is almost separated from $A$.
Note that the properties of an \partition depend on three parameters $\t$, $r$, $\y$.
The results of this section imply that our random graphs are asymptotically
almost surely \partitionable for $\t = \momn^{\Omega(1)}$ and constant $r,\y$.
It therefore helps to assume that $\t$ is a slowly
growing function in $n$, such as $\log(n)$ and $r,\y$ are constants.
Higher values of $\y$ boost the probability of a random graph being \partitionable.
The parameter $\y$ is therefore crucial for the design of efficient algorithms.

For an $\scale$-\dominated random graph model $\Gnn$, we always assume
the vertices of $\Gn$ to be 
$v_1,\dots,v_n$, ordered as in Definition~\ref{def:wellBehavedPowerLaw}.
We will choose
$A=\{v_1,\dots,v_\t\}$,
$B=\{v_{\t+1},\dots,v_{\t^\y}\}$,
$C=\{v_{\t^\y+1},\dots,v_n\}$
and show that the probability is low that $(A,B,C)$ does not form a \partition.
We do this in two steps:
In \Cref{sec:forbiddenSubgraph}, 
we define $\HH$ to be a set of graphs over the vertex set $\{v_1,\dots,v_n\}$.
We show that if $(A,B,C)$ is not a \partition then
the complete edge-set of some graph in $\HH$ is present in the graph.
In \Cref{sec:subgraphProb} we bound the probability 
of the edge-set of any graph from $\HH$ being present in the random graph model.

At last, in Lemma~\ref{lem:condNhood}, \Cref{sec:nhoodsizes}, 
we bound the sum of the expected
sizes of all $r$-neighborhoods in an $\scale$-\dominated graph class.
This is needed to bound the expected run time of an algorithm
that iterates over all $r$-neighborhoods of a graph.



\subsection{Forbidden Edge-Sets Characterization}\label{sec:forbiddenSubgraph}

Here we show that if a graph is not \partitionable then it contains
some forbidden edge-set.

\begin{definition}
    Let $G$ be a graph and $\cal H$ be a set of graphs over $V(G)$.
    We say $\cal H \embeds G$ if for some $H \in \cal H$, $E(H) \subseteq E(G)$.
\end{definition}

\begin{definition}\label{def:HH}
    Let $\t,r,\y,n \in \N^+$.
    We define $\HH$ to be the set of
    \begin{itemize}
        \item 
            all graphs with vertex set $V \subseteq \{v_{\t+1},\dots,v_{n}\}$ such that
            $|V| \le 200r\y^3$,
            all vertices have degree at least two,
            and the graph has an \excess of $\y^2$, and
        \item
            all graphs $(V_1 \cup V_2, E)$ such that
            $V_1 \subseteq \{v_1,\dots,v_{\t}\}$,
            $V_2 \subseteq \{v_{\t^\y+1},\dots,v_{n}\}$,
            $|V_1 \cup V_2| \le 25r\y^2$,
            $|E| \le 25r\y^2$,
            $|V_1| \le \y$,
            all vertices in $V_2$ have degree at least two,
            and the summed degree of $V_2$ is $2|V_2|-2+\y$.
    \end{itemize}
\end{definition}
\begin{lemma}
  \label{lem:ifnotpartthenembed}
    Let $\t,r,\y,n \in \N^+$.
    If a graph $G$ with vertex set $\{v_1,\dots,v_n\}$ is not \partitionable, 
    then $\HH \embeds G$.
\end{lemma}
\begin{proof}
    Assume a graph $G$ is not \partitionable.
    Then the tuple $(A,B,C)$ with
    $A= \{v_1,\dots,v_\t\}$,
    $B=\{v_{\t+1},\dots,v_{\t^\y}\}$,
    $C=\{v_{\t^\y+1},\dots,v_n\}$
    is not a \partition of $G$.
    This means $(A,B,C)$ either does not satisfy
    Property~\ref{prop:ABCExcess} or~\ref{prop:ABCEdgesToA} of
    Definition~\ref{def:partition}.

    Assume now $(A,B,C)$ does not satisfy Property~\ref{prop:ABCExcess}.
    Then there is a $\PartitionRadiusExtraEdges$-neighborhood in $G[B \cup C]$
    with an \excess of at least $\y^2$. Let $T$ be a breadth-first-search
    tree of this neighborhood of depth $\PartitionRadiusExtraEdges$ with a root $v$. There
    are $\y^2+1$ extra edges in this neighborhood which are not contained in $T$. 
    Let $H$ be the graph constructed by the following procedure:
    We induce $G$ on all vertices
    which are either an endpoint of the $\y^2+1$ extra edges
    or lie on the unique path in $T$ from such an endpoint to the root $v$.
    Then we iteratively remove all vertices with degree one.
    Every vertex in $H$ has degree at least two.
    In $T$, each path starting at $v$ has length at most $\PartitionRadiusExtraEdges$,
    and there are at most $(\y^2+1)$ extra edges.
    Therefore, $H$ consists of at most
    $2(\y^2+1)(\PartitionRadiusExtraEdges+1)$ vertices.
    Furthermore, $H$ contains $\y^2$ more edges than vertices. 
    This means $G[B \cup C]$ contains a subgraph with an \excess of $\y^2$
    and $2(\y^2+1)(\PartitionRadiusExtraEdges+1) \le 200r\y^3$ 
    vertices that all have degree at least two.
    Such a graph is contained in $\HH$.

    Assume now $(A,B,C)$ does not satisfy Property~\ref{prop:ABCEdgesToA}.
    Then $G[C]$ contains a $\PartitionRadiusToA$-neighbor\-hood
    such that there are $\y$ edges  going from this $\PartitionRadiusToA$-neighborhood
    to $A$. Let these edges be $u_1w_1,\dots,u_\y w_\y$ with $u_i \in A$ and $w_i \in C$.
    Let $T$ be a breadth-first-search tree of depth $\PartitionRadiusToA$ of this
    $\PartitionRadiusToA$-neighborhood with root $v$. Let $V_1 = \{u_1,\dots,u_\y\}$
    and let $V_2$ be the set of vertices that lie for each $w_i$ on
    the unique path in $T$ of length at most $\PartitionRadiusToA$ from $w_i$ to the
    root $v$, including $w_i$ and $v$. 
    Let $H$ be the graph with vertex set $V_1 \cup V_2$ and 
    all edges from $T[V_2]$, as well as all edges between $V_1$ and $V_2$ in $G$.
    Notice that $|V_1| \le \y$ and 
    $|V_2| \le (\PartitionRadiusToA+1) \y$.
    Also $H[V_2]$ forms a tree with $\y$ outgoing edges to
    $V_1$. Therefore, the vertices in $V_2$ have in $H$ a summed
    degree of $2|V_2|-2 + \y$.
    They also have degree at least two in $H$.
    This means $G$ contains a subgraph $(V_1 \cup V_2, E)$ such that
    $V_1 \subseteq A$, $V_2 \subseteq C$, $|V_1+V_2| \le 25r\y^2$, 
    $|V_1| \le \y$, $|E| \le 25r\y^2$,
    and the vertices in $V_2$ have degree at least two and 
    a summed degree of $2|V_2|-2 + \y$.
    Such a graph is contained in $\HH$.
\end{proof}

\subsection{Bounding Probabilities of Edge-Sets}\label{sec:subgraphProb}

In this section we bound for an $\scale$-\dominated
random graph model $\Gnn$ the probability that $\HH \embeds \Gn$,
thereby bounding the probability that $\Gn$ is not \partitionable.

\begin{definition}\label{def:pbeta}
    Let $\cal H$ be a set of graphs over the vertex set $\{v_1,v_2,\dots\}$,
    and let $E$ be a set of edges over the same vertex set.
    We define
    \begin{align*}
        p_\scale(E,n) =&
        \momn^{|E|^2}
        \prod_{v_iv_j \in E}
        \frac{(n/i)^{1/(\scale-1)}}{n^{1/2}}
        \frac{(n/j)^{1/(\scale-1)}}{n^{1/2}} \\
        p_\scale(\cal H,n) =& \sum_{H \in \cal H} p_\scale(E(H),n).
    \end{align*}
\end{definition}

\begin{lemma}\label{lem:boundEmbeddingProb}
    Let $\Gnn$ be an $\scale$-\dominated random graph model.
    Let $(\cal H_n)_{n \in \N}$ be a sequence of sets of graphs
    over the vertex set $\{v_1,v_2,\dots\}$.
    Then $\P[\cal H_n \embeds \Gn] \le p_\scale(\cal H_n,n)$.
\end{lemma}
\begin{proof}
    Using the union bound and Definition~\ref{def:wellBehavedPowerLaw} we see
    \begin{align*}
        \P[\cal H_n \embeds \Gn] \le \sum_{H \in \cal H_n} 
        \P[ E(H) \subseteq E(\Gn)]
        \le \sum_{H \in \cal H_n} p_\scale(H,n)
        = p_\scale(\cal H_n,n).
    \end{align*}
\end{proof}
It is therefore sufficient to bound $p_\scale(\HH)$.
The following two Lemmas prove some technicalities we need to do so.

\begin{lemma}\label{lem:sumbound}
    Let $\scale \ge 2$.
    For $n \in \N^+$
    $$
    \sum_{i=1}^{n} \frac{(n/i)^{1/(\scale-1)}}{n^{1/2}} \le \momn\sqrt{n}, \quad
    \sum_{i=1}^{n} \frac{(n/i)^{2/(\scale-1)}}{n} \le \momn^2.
    $$
    For $\t,n \in \N^+, \delta \in \N$
    $$
    \sum_{i=1}^{\t} \frac{(n/i)^{\delta/(\scale-1)}}{n^{\delta/2}}
    \le \momn^\delta \t.
    $$
    For $\t,n,\delta \in \N^+$, $\delta \ge 2$
    $$
    \sum_{i=\t+1}^{n} \frac{(n/i)^{\delta/(\scale-1)}}{n^{\delta/2}}
    \le \momn^\delta \t^{1-\delta/2}.
    $$
\end{lemma}
\begin{proof}
    We define $\gamma = 1/(\scale-1)$ and
    $$
    \rho_\gamma(n) =
    \begin{cases}
        O(1) &\quad \gamma < 1/2 \\
        \log(n)^{O(1)} &\quad \gamma = 1/2 \\
        O(n^{\gamma - 1/2}) &\quad \gamma > 1/2.
    \end{cases}
    $$
    For $\scale > 3$ we have $\gamma < 1/2$ and thus $\rho_\gamma(n) = \momn = O(1)$.
    Similarly, for $\scale = 3$ holds $\gamma = 1/2$ and $\rho_\gamma(n) = \momn = \log(n)^{O(1)}$.
    For $2 \le \scale < 3$ we have $\gamma - 1/2 \le 3 - \scale$.
    Therefore 
    $\rho_\gamma(n) \le \momn$ for all values of $\scale \ge 2$.
    It is now sufficient to show
    $$
    \sum_{i=1}^{n} \frac{(n/i)^\gamma}{n^{1/2}} = \rho_\gamma \sqrt{n}, \quad
    \sum_{i=1}^{n} \frac{(n/i)^{2\gamma}}{n} = \rho_\gamma(n)^2,
    $$
    $$
    \sum_{i=1}^{\t} \frac{(n/i)^{\delta\gamma}}{n^{\delta/2}}
    \le \rho_\gamma(n)^\delta \t,
    $$
    and for $\delta \ge 2$
    $$
    \sum_{i=\t+1}^{n} \frac{(n/i)^{\delta\gamma}}{n^{\delta/2}}
    \le 
    \rho_\gamma(n)^\delta
    \t^{1-\delta/2}.
    $$
    We bound with $\gamma < 1$
    $$
    \sum_{i=1}^{n} \frac{(n/i)^{\gamma}}{n^{1/2}}
    \le
    n^{\gamma-1/2}
    \int_{0}^{n} \frac{1}{t^{\gamma}} dt 
    = n^{\gamma-1/2}
    n^{1-\gamma}/(1-\gamma)
    \le \rho_\gamma(n)\sqrt n
    $$
    and
    \begin{equation}\label{eq:asdfasdf}
    \sum_{i=1}^{n} \frac{(n/i)^{2\gamma}}{n} 
    =
    n^{2\gamma-1} +
    n^{2\gamma-1}
    \sum_{2}^{n} \frac{1}{i^{2\gamma}}
    \le
    O(n^{2\gamma-1})
    \int_{1}^{n} \frac{1}{t^{2\gamma}} dt 
    \le \rho_\gamma(n)^2.
    \end{equation}
    We further bound
    $$
    \sum_{i=1}^{\t} \frac{(n/i)^{\delta\gamma}}{n^{\delta/2}}
    \le
    n^{\delta(\gamma-1/2)}\t
    \le \rho_\gamma(n)^\delta \t.
    $$
    To prove the last bound, we make a case distinction over $\delta$ and
    $\gamma$.  
    At first, assume $\delta=2$. Then
    $$
    \sum_{i=\t+1}^{n} \frac{(n/i)^{\delta\gamma}}{n^{\delta/2}}
    \stackrel{(\ref{eq:asdfasdf})}{\leq}
    \rho_\gamma(n)^\delta
    = \rho_\gamma(n)^\delta \t^{1-\delta/2}.
    $$
    Assume now that $\delta \ge 3$.
    We have for $\gamma \le 1/2$
    $$
    \sum_{i=\t+1}^{n} \frac{(n/i)^{\delta\gamma}}{n^{\delta/2}}
    \le \sum_{i=\t+1}^{n} \frac{(n/i)^{\delta/2}}{n^{\delta/2}}
    \le
    \int_{\t}^{n} \frac{1}{t^{\delta/2}} dt
    \le O(1) \t^{1-\delta/2}
    \le \rho_\gamma(n)^\delta \t^{1-\delta/2}
    $$
    and for $\gamma > 1/2$
    $$
    \sum_{i=\t+1}^{n} \frac{(n/i)^{\delta\gamma}}{n^{\delta/2}}
    \le
    n^{\delta(\gamma-1/2)} \int_{\t}^{n} \frac{1}{t^{\delta\gamma}} dt
    \le O(n^{\delta(\gamma-1/2)}) \t^{1-\delta\gamma} 
    \le \rho_\gamma(n)^\delta \t^{1-\delta/2}.
    $$
\end{proof}

\begin{lemma}\label{lem:boundLL}
    Let $\t_1,\t_2,l,k,n \in \N^+$ with $\t_1 \le \t_2$.
    Let $\LL$ be the set of all graphs $(V_1 \cup V_2, E)$ such that
    $V_1 \subseteq \{v_1,\dots,v_{\t_1}\}$,
    $V_2 \subseteq \{v_{\t_2+1},\dots,v_{n}\}$,
    $|V_1 \cup V_2| \le l$,
    $|E| \le l$,
    $|V_1| \le k$,
    all vertices in $V_2$ have degree at least two,
    and the summed degree of $V_2$ is $2|V_2|-2+k$.
    Then $p_\scale(\LL,n) \le \momn^{O(l^2)} \t_1^k \t_2^{1-k/2}$.
\end{lemma}
\begin{proof}
    We can partition the set $\LL$ into at most $2^{l^2}$ many isomorphism classes.
    Let $\cal L' \subseteq \LL$ be the isomorphism class which
    maximizes $p_\scale(\cal L',n)$.
    We have that $p_\scale(\LL,n) \le 2^{l^2} p_\scale(\cal L',n)$.
    We fix a representative $H = (V_1 \cup V_2, E) \in \cal L'$.

    Let now $\gamma = |V_1|$ and $\Gamma = |V_2|$.
    We order the sets $V_1$ and $V_2$ such that we can speak of the
    first, second, etc.\ vertex in each set.
    Let $F$ be the set of all sequences of integers 
    $(x_1,\dots,x_\gamma,y_1,\dots,y_\Gamma)$
    without duplicates and with $1 \le x_i \le \t_1$ and $\t_2+1 \le y_i \le n$.
    For a sequence $f \in F$ let $f(H)$ be the homomorphism of $H$ 
    where the $i$th vertex from $V_1$ is assigned to $v_{x_i}$ (for $1 \le i \le \gamma$)
    and the $i$th vertex from $V_2$ is assigned to $v_{y_i}$ (for $1 \le i \le \Gamma$).
    Then $\cal L' = \bigcup_{f \in F}f(H)$.

    The vertices in $V_1$ and $V_2$ have a degree sequence 
    $\delta_1,\dots,\delta_{\gamma}, \Delta_1,\dots,\Delta_{\Gamma}$.
    We fix a sequence $f = (x_1,\dots,x_\gamma,y_1,\dots,y_\Gamma) \in F$.
    Then by Definition~\ref{def:pbeta}
    \begin{equation} \label{eq:sum_bound5}
        p_\scale(E(f(H)), n) =
        \momn^{(\delta_1+\dots+\delta_\gamma+\Delta_1+\dots+\Delta_\Gamma)^2/4}
        \prod_{i=1}^{\gamma} \frac{(n/x_i)^{\delta_i/(\scale-1)}}{n^{\delta_i/2}}
        \prod_{i=1}^{\Gamma} \frac{(n/y_i)^{\Delta_i/(\scale-1)}}{n^{\Delta_i/2}}.
    \end{equation}
    Observe that
    \begin{equation} \label{eq:sum_bound2}
        \Delta_i \ge 2 \textnormal{ for $1 \le i \le \Gamma$},
    \end{equation}
    \begin{equation} \label{eq:sum_bound3}
        \delta_1+\dots+\delta_\gamma+\Delta_1+\dots+\Delta_\Gamma \le 2l ,
    \end{equation}
    \begin{equation} \label{eq:sum_bound4}
        \sum_{i=1}^{\Gamma} (1-\Delta_i/2) = \Gamma - 
        \frac12 \sum_{i=1}^{\Gamma}\Delta_i = \Gamma - (\Gamma-1+k/2) = 1-k/2.
    \end{equation}

        We enumerate all sequences in $F$, and use
        (\ref{eq:sum_bound5}),
        (\ref{eq:sum_bound2}),
        (\ref{eq:sum_bound3}),
        (\ref{eq:sum_bound4}),
        and Lemma~\ref{lem:sumbound} to bound
        \begin{align*}
            & p_\scale(\LL,n) \le 2^{l^2} p_\scale(\cal L',n) = 2^{l^2} \sum_{f \in F} p_\scale(E(f(H)),n) \\
        \stackrel{\textnormal{(\ref{eq:sum_bound3})(\ref{eq:sum_bound5})}}{\leq}&
            2^{l^2}
            \sum_{x_1=1}^{\t_1} \dots \sum_{x_{\gamma}=1}^{\t_1}
            \sum_{y_1=\t_2+1}^n \dots \sum_{y_{\Gamma}=\t_2+1}^n
            \momn^{O(l^2)} \\&
            \prod_{i=1}^{\gamma} \frac{(n/x_i)^{\delta_i/(\scale-1)}}{n^{\delta_i/2}}
            \prod_{i=1}^{\Gamma} \frac{(n/y_i)^{\Delta_i/(\scale-1)}}{n^{\Delta_i/2}} \\
        \le&
            \momn^{O(l^2)}
            \sum_{x_1=1}^{\t_1} 
            \frac{(n/x_1)^{\delta_1/(\scale-1)}}{n^{\delta_1/2}}
            \dots 
            \sum_{x_{\gamma}=1}^{\t_1}
            \frac{(n/x_\gamma)^{\delta_\gamma/(\scale-1)}}{n^{\delta_\gamma/2}} \\
        &
            \sum_{y_1=\t_2+1}^n 
            \frac{(n/y_1)^{\Delta_1/(\scale-1)}}{n^{\Delta_1/2}}
            \dots 
            \sum_{y_{\Gamma}=\t_2+1}^n
            \frac{(n/y_\Gamma)^{\Delta_\Gamma/(\scale-1)}}{n^{\Delta_\Gamma/2}} \\
        \stackrel{\textnormal{(\ref{eq:sum_bound2}), Lemma \ref{lem:sumbound}}}{\leq}&
            \momn^{O(l^2)}
            \prod_{i=1}^{\gamma} 
            \momn^{\delta_i}
            \t_1
            \prod_{i=1}^{\Gamma} 
            \momn^{\Delta_i}
            \t_2^{1-\Delta_i/2} \\
        \stackrel{\textnormal{(\ref{eq:sum_bound3})}}{\leq}&
            \momn^{O(l^2)}
            \prod_{i=1}^{\gamma} 
            \t_1
            \prod_{i=1}^{\Gamma} 
            \t_2^{1-\Delta_i/2} 
        \stackrel{(\ref{eq:sum_bound4})}{\leq}
            \momn^{O(l^2)}
            \t_1^{k}
            \t_2^{1-k/2}.
    \end{align*}

\end{proof}

\begin{lemma}
\label{lem:embedprob}
    Let $\t,r,\y,n \in \N^+$ with $\y \ge 5$. Then
    $$
    p_\scale(\HH,n) \le 
    \momn^{O(\y^6r^2)}\t^{-\y^2/10}.
    $$
\end{lemma}
\begin{proof}
    We compare the definition of $\HH$ and $\LL$ and see that
    $$
    \HH \subseteq \cal L_n(1,\t,200r\y^3+\y^2,2\y^2+2) \cup \cal L_n(\t,\t^\y,25r\y^2,\y).
    $$
    Using Lemma~\ref{lem:boundLL} and the union bound we compute
    \begin{align*}
    &   p_\scale(\HH,n) \\
    &\le p_\scale(L_n(1,\t,200r\y^3+\y^2,2\y^2+2),n) + p_\scale(\cal L_n(\t,\t^\y,25r\y^2,\y),n) \\
    &\le \momn^{O(r \y^3)^2}\t^{-\y^2} + 
         \momn^{O(r\y^2)^2}\t^{\y}(\t^\y)^{-\y/2 +1} \\
    &\le \momn^{O(r^2 \y^6)}(\t^{-\y^2} + 
         \t^{\y}(\t^\y)^{-\y/2 +1}) \\
    &\le \momn^{O(r^2 \y^6)}(\t^{-\y^2} + 
         \t^{-\y^2/2 +2\y}) \\
    &\le \momn^{O(r^2 \y^6)}\t^{-\y^2/2 +2\y} \\
    &\stackrel{\y \ge 5}{\le} \momn^{O(r^2 \y^6)}\t^{-\y^2/10}.
    \end{align*}
\end{proof}

\begin{theorem}\label{thm:partitionProbability}
    Let $\Gnn$ be an $\scale$-\dominated random graph model
    and let $\t,r,\y,n \in \N^+$ with $\y \ge 5$.
    The probability that $\Gn$ is not \partitionable is at
    most 
    $\momn^{O(\y^6r^2)}\t^{-\y^2/10}$.
\end{theorem}
\begin{proof}
    Combining Lemma~\ref{lem:ifnotpartthenembed},~\ref{lem:boundEmbeddingProb}, and \ref{lem:embedprob}.  
\end{proof}

\subsection{Expected Neighborhood Sizes}
\label{sec:nhoodsizes}

In this section we bound the sum of the expected sizes of all $r$-neighborhoods
of a graph from an $\scale$-\dominated random graph model
under the condition that
$\t \in \N$ is the minimal value such that a graph is \partitionable.
We start with the simpler condition that $\cal H_n \embeds \Gn$
for some set of graphs $\cal H_n$
and then lift this result using Lemma~\ref{lem:ifnotpartthenembed} from the previous section.

\begin{lemma}\label{lem:condNhood1}
    Let $\Gnn$ be an $\scale$-\dominated random graph model.
    Let $r \in \N^+$ and let $(\cal H_n)_{n \in \N}$ be a sequence of sets of graphs
    where every graph has size at most $h \ge 1$.
    Then
    $$
        \E\bigl[\sum_{v \in V(\Gn)} \|\Gn[N^{\Gn}_r(v)]\| \bigm| \cal H_n \embeds \Gn \bigr]  
        \P[\cal H_n \embeds \Gn]
        \le h^{O(r)} \momn^{O(r^2+h^2)} n p_\scale(\cal H_n,n).
    $$
\end{lemma}
\begin{proof}
    We fix an $n$.
    Let $\cal Q$ be the set of all paths of length at most $r+1$
    over $V(\Gn)$.
    Then by linearity of expectation
    $$
        \E\bigl[\sum_{v \in V(\Gn)} \|\Gn[N^{\Gn}_r(v)]\| \bigr] \le 
        2 \sum_{Q \in \cal Q} \P[E(Q) \subseteq E(\Gn)]. 
    $$
    We use this observation and the union bound to compute
    \begin{align}\label{eq:expsize}
        & \E\bigl[\sum_{v \in V(\Gn)} \|\Gn[N^{\Gn}_r(v)]\| \bigm| \cal H_n \embeds \Gn \bigr]  \P[\cal H_n \embeds \Gn] \nonumber\\
        &\le \sum_{H \in \cal H_n} 
        \E\bigl[\sum_{v \in V(\Gn)} \|\Gn[N^{\Gn}_r(v)]\| \bigm| E(H) \subseteq E(\Gn) \bigr]  \P[E(H) \subseteq E(\Gn)] \nonumber\\
        &\le 2 \sum_{H \in \cal H_n} \sum_{Q \in \cal Q}
        \P[E(Q) \subseteq E(\Gn) \mid E(H) \subseteq E(\Gn)]  \P[E(H) \subseteq E(\Gn)] \nonumber\\
        &\le 2 \sum_{H \in \cal H_n} \sum_{Q \in \cal Q}
        \P[E(Q) \subseteq E(\Gn), E(H) \subseteq E(\Gn)] \nonumber\\
        &\le 2 \sum_{H \in \cal H_n} \sum_{Q \in \cal Q}
        p_\scale(E(H) \cup (E(Q) \setminus E(H)),n) \nonumber\\
        &\le \momn^{O(r^2+h^2)} \sum_{H \in \cal H_n} p_\scale(E(H),n) \sum_{Q \in \cal Q}
        p_\scale(E(Q)\setminus E(H),n).
    \end{align}
    We fix a graph $H \in \cal H_n$.
    We want to find a good bound for $p_\scale(E(Q)\setminus E(H),n)$ for every $Q \in \cal Q$.
    Let $Q \in \cal Q$ be a path.
    We assume the vertices $V(Q) = \{w_1,\dots,w_q \}$ with $q \le r+2$ to be ordered such that 
    edges are only between consecutive vertices.
    Let $s = (s_1,\dots,s_{q-1}) \in \{0,1\}^{q-1}$ be the unique bit-string
    with
    $$
    \{ (w_i,w_{i+1}) \mid s_i=1, 1 \le i < q \} = E(Q) \setminus E(H).
    $$
    This means $s$ describes which edges of $Q$ are not present in $H$.
    Let $Q' = (V(Q), E(Q) \setminus E(H))$.
    The degree sequence of $Q'$ is 
    $\delta^s_1,\dots,\delta^s_{q}$ with $\delta_i^s = s_{i-1} + s_{i}$ (we assume $s_0=s_q=0$).
    We define 
    $$
    W(\delta)=
    \begin{cases}
        V(\Gn) &\quad \delta = 2 \\
        V(H) &\quad   \delta < 2
    \end{cases}
    \quad \text{and} \quad
    X(\delta)=
    \begin{cases}
        \{1,\dots,n\} &\quad \delta = 2 \\
        \{i \mid v_i \in V(H)\} &\quad \delta < 2.
    \end{cases}
    $$
    If $\delta_1^s = 0$ then $w_1 \in V(H) = W(\delta_1^s$+1). 
    If $\delta_q^s = 0$ then $w_q \in V(H) = W(\delta_q^s$+1). 
    If $\delta_i^s \in \{0,1\}$ then $w_i \in V(H) = W(\delta_i^s)$ for $2 \le i \le q-1$.
    Therefore
    \begin{align}\label{eq:fufufu}
        &   \sum_{Q \in \cal Q} p_\scale(E(Q) \setminus E(H)) \nonumber\\
        \le&
        \sum_{q=1}^{r+2}
            \sum_{s \in \{0,1\}^{q-1}}
            \sum_{w_1 \in W(\delta^s_1+1)} 
            \sum_{w_2 \in W(\delta^s_2)} 
            \dots
            \sum_{w_{q-1} \in W(\delta^s_{q-1})} 
            \sum_{w_q \in W(\delta^s_q+1)} \nonumber\\
        &   \quad\quad p_\scale(\{ (w_{i},w_{i+1}) \mid s_i=1, 1 \le i < q \}, n) \nonumber\\
        =&
        \sum_{q=1}^{r+2} 
            \sum_{s \in \{0,1\}^{q-1}}
            \sum_{x_1 \in X(\delta^s_1+1)} 
            \sum_{x_2 \in X(\delta^s_2)} 
            \dots
            \sum_{x_{q-1} \in X(\delta^s_{q-1})} 
            \sum_{x_q \in X(\delta^s_q+1)} \nonumber\\
        &   \quad\quad \momn^{O(q^2)} \prod_{i=1}^q \frac{(n/x_i)^{\delta^s_i/(\scale-1)}}{n^{\delta^s_i/2}} \nonumber\\
        =&
            \momn^{O(r^2)}
            \sum_{q=1}^{r+2}
            \sum_{s \in \{0,1\}^{q-1}}
            \sum_{x_1 \in X(\delta^s_1+1)} 
                \frac{(n/x_{1})^{\delta^s_{1}/(\scale-1)}}{n^{\delta^s_{1}/2}}
            \sum_{x_2 \in X(\delta^s_2)} 
                \frac{(n/x_{2})^{\delta^s_{2}/(\scale-1)}}{n^{\delta^s_{2}/2}}
            \dots \nonumber\\
        &   \quad\quad \sum_{x_{q-1} \in X(\delta^s_{q-1})} 
                \frac{(n/x_{q-1})^{\delta^s_{q-1}/(\scale-1)}}{n^{\delta^s_{q-1}/2}}
            \sum_{x_q \in X(\delta^s_q+1)}
                \frac{(n/x_{q})^{\delta^s_{q}/(\scale-1)}}{n^{\delta^s_{q}/2}}.
    \end{align}
    The bound of (\ref{eq:fufufu}) depends on the degree sequence $\delta^s_1,\dots,\delta^s_q$.
    Remember that $\delta^s_1,\delta^s_q \in \{0,1\}$
    and $\delta^s_i \in \{0,1,2\}$ for $1 < i < q$.
    The following five bounds follow from Lemma~\ref{lem:sumbound}.
    \begin{equation*}
        \sum_{x \in X(0)} 
        \frac{(n/x)^{0/(\scale-1)}}{n^{0/2}}
        \le h
        \le \momn^2 h
    \end{equation*}
    \begin{equation*}
        \sum_{x \in X(1)} 
        \frac{(n/x)^{1/(\scale-1)}}{n^{1/2}}
        \le \momn h
        \le \momn^2 h
    \end{equation*}
    \begin{equation*}
        \sum_{x \in X(2)} 
        \frac{(n/x)^{2/(\scale-1)}}{n^{2/2}}
        \le \momn^2
        \le \momn^2 h
    \end{equation*}
    \begin{equation*}
        \sum_{x \in X(1)} 
        \frac{(n/x)^{0/(\scale-1)}}{n^{0/2}}
        \le h
        \le \momn^2 h
    \end{equation*}
    \begin{equation*}
        \sum_{x \in X(2)} 
        \frac{(n/x)^{1/(\scale-1)}}{n^{1/2}}
        \le \momn\sqrt{n}
        \le \momn^2 h^2 \sqrt{n}
    \end{equation*}
    These five bounds can be used to bound the inner $q$ sums
    of (\ref{eq:fufufu}).
    This yields
    \begin{equation}\label{eq:blabla}
        \sum_{Q \in \cal Q} p_\scale(E(Q) \setminus E(H))
        \le
        \momn^{O(r^2)}
        \sum_{q=1}^{r+2} 
        \sum_{s \in \{0,1\}^{q-1}}
        \momn^{2q}
        h^q
        \sqrt{n}
        \sqrt{n}
        \le h^{O(r)} \momn^{O(r^2)} n.
    \end{equation}

    At last, we combine
    (\ref{eq:expsize}) and
    (\ref{eq:blabla}) and get
    \begin{align*}\label{eq:expsize}
        & \E\bigl[\sum_{v \in V(\Gn)} \|\Gn[N^{\Gn}_r(v)]\| \bigm| H_n \embeds \Gn \bigr]  \P[\cal H \embeds \Gn] \\ 
        &\le 
        \momn^{O(r^2+h^2)}
        \sum_{H \in \cal H_n} p_\scale(E(H),n) \sum_{Q \in \cal Q}
        p_\scale(E(Q)\setminus E(H),n) \\
        &\le p_\scale(\cal H_n,n) h^{O(r)} \momn^{O(r^2+h^2)} n.
    \end{align*}

\end{proof}

\begin{lemma}\label{lem:condNhood}
    Let $\Gnn$ be an $\scale$-\dominated random graph model.
    Let $r,\y,n \in \N^+$ with $\y \ge 5$.
    Let $A_\t$ be the event that $\t \in \N^+$ is the minimal value such that
    $\Gn$ is \partitionable.
    Then 
    $$
        \E\bigl[\sum_{v \in V(\Gn)} \|\Gn[N^{\Gn}_r(v)]\| \bigm| A_\t\bigr] \P[A_\t] \le
        (r\y)^{O(r)} \momn^{O(\y^6r^2)}\t^{-\y^2/10} n.
    $$
\end{lemma}
\begin{proof}
    We start with a general observation about conditional expected values.
    Let $X$ be a non-negative random variable and $A \subseteq B$ be events.
    Then 
    \begin{equation}\label{eq:condExp}
        \E[X \mid A] \P[A] \le \E[X \mid B] \P[B].
    \end{equation}

    Assume $\t \ge 3$.
    Let $G$ be a graph with $V(G) = \{v_1,\dots,v_n\}$.
    If $\t \in \N^+$ is the minimal value such that $G$ is \partitionable
    then $G$ is not $(\t-1)$-$r$-$\y$-partitionable.
    Then by Lemma~\ref{lem:ifnotpartthenembed}, $\HHminus \embeds G$.
    Using (\ref{eq:condExp}), we see 
    \begin{multline*}
        \E\bigl[\sum_{v \in V(\Gn)} \|\Gn[N^{\Gn}_r(v)]\| \bigm| A_\t\bigr] \P[A_\t] \le \\
        \E\bigl[\sum_{v \in V(\Gn)} \|\Gn[N^{\Gn}_r(v)]\| \bigm| \HHminus \embeds \Gn\bigr] 
        \P[\HHminus \embeds \Gn].
    \end{multline*}
    Every subgraph in $\HHminus$ has by Definition~\ref{def:HH} size at most $200r\y^3$.
    Also for $\t \ge 3$ we have $(\t-1)^{-1} \le \t^{-1/2}$.
    Lemma~\ref{lem:condNhood1} and~\ref{lem:embedprob} imply
    \begin{align*}
        & \E\bigl[\sum_{v \in V(\Gn)} \|\Gn[N^{\Gn}_r(v)]\| \bigm| A_\t\bigr] \P[A_\t] \\
        \le& \E\bigl[\sum_{v \in V(\Gn)} \|\Gn[N^{\Gn}_r(v)]\| \bigm| \HHminus \embeds \Gn\bigr] \P[\HHminus \embeds \Gn] \\
        \stackrel{\ref{lem:condNhood1}}{\le}
        &(200r\y^3)^{O(r)} \momn^{O(\y^6r^2)} n p_\scale(\HHminus,n) \\
        \stackrel{\ref{lem:embedprob}}{\le}
        & (r\y)^{O(r)} \momn^{O(\y^6r^2)} n \momn^{O(\y^6r^2)}(\t-1)^{-\y^2/5} \\
        \le& (r\y)^{O(r)} \momn^{O(\y^6r^2)}\t^{-\y^2/10} n.
    \end{align*}

    Assume $\t \le 2$.
    By (\ref{eq:condExp}) and Lemma~\ref{lem:condNhood1} with 
    $\cal H_n = \{ \emptyset \}$
    \begin{multline*}
        \E\bigl[\sum_{v \in V(\Gn)} \|\Gn[N^{\Gn}_r(v)]\| \bigm| A_\t\bigr] \P[A_\t]
        \le \E\bigl[\sum_{v \in V(\Gn)} \|\Gn[N^{\Gn}_r(v)]\| \bigr] \\
        \le \momn^{O(r^2)} n
        \le (r\y)^{O(r)} \momn^{O(\y^6r^2)}\t^{-\y^2/10} n.
    \end{multline*}
\end{proof}


\section{Protrusion Decompositions of Neighborhoods}\label{sec:protrusiondecomp}

\newcommand{\tG}{G^*}
\newcommand{\tCA}{C_A^*}
\newcommand{\tCB}{C_B^*}
\newcommand{\tCY}{C_Y^*}
\newcommand{\tCZ}{C_Z^*}
\newcommand{\tE}{E^*}

\newcommand{\hG}{G^r}
\newcommand{\hCA}{C_A^r}
\newcommand{\hCB}{C_B^r}
\newcommand{\hCY}{C_Y^r}
\newcommand{\hCZ}{C_Z^r}
\newcommand{\hA}{A^r}
\newcommand{\hB}{B^r}
\newcommand{\hC}{C^r}
\newcommand{\hE}{E^r}


\newcommand{\pathlengthToCb}{17\y r}



\newcommand{\fummlerRadius}{20 \y r}

\newcommand{\fummlerSize}{130r\y^3}

\newcommand{\numFummlers}{390r\y^5}


\newcommand{\BEdgeCompSize}[1]{1600\y^7r(#1+1)^2}

\newcommand{\verticesInBadComponents}{O(\y^{17}r^3\t^{4\y})}

In this section, we show that local neighborhoods of \dominated graph classes
are likely to have the following nice structure:
They consist of a (small) core graph to which so called \emph{protrusions} are attached.
Protrusions are (possibly large) subgraphs with small treewidth and boundary.
The \emph{boundary} of a subgraph is the size of its neighborhood in the remaining graph.
Protrusions were introduced by Bodlaender et al.\ for very general kernelization
results in graph classes with bounded genus~\cite{bodlaender2016meta}.

Earlier, (Theorem~\ref{thm:partitionProbability})
we showed that $\scale$-\dominated random graph models
are (for certain values of $\scale$, $\t$, $r$, $\y$) likely to be \partitionable.
It is therefore sufficient to show that $r$-neighborhoods of \partitionable
graphs have such a nice protrusion structure.

However, in general it is not easy to find protrusions in a graph~\cite{kim2016linear}.
As we later need to be able to find them,
we define special protrusion decompositions, called \nhoodpartitions
in which (most of) the protrusions can be efficiently identified.
The main and only result of this section is the following theorem.

\begin{customthm}{\ref{thm:partitionIsNhoodpartition}}
    Let $\t,r,\y \in \N^+$ and let $G$ be an \partitionable graph.
    Let $\hG$ be an $r$-neighborhood in $G$.
    Then $\hG$ is \OOnhoodpartitionable.
\end{customthm}

It remains to define what a \nhoodpartition of a graph $\hG$ with
radius at most $r$ is.
The definition has to strike the right balance:
It needs to be permissive enough such that
neighborhoods of \dominated graph classes are likely to have
this structure and it needs to be restrictive enough to admit
efficient algorithms.
Informally speaking, a \nhoodpartition of a graph $\hG$ 
is a partition $(X,Y,Z)$ of the vertices of $\hG$
such that $X$ has small size and the connected components of $\hG[Y \cup Z]$ 
are protrusions.
In order to be able to efficiently identify the protrusions, we further
require that
the components of $\hG[Y]$ have bounded size and
the components of $\hG[Z]$ are trees.
This is formalized in the following definition.

\begin{definition}[\nhoodpartition]\label{def:nhoodpartition}
    Let $\t,r,\y \in \N^+$. Let $\hG$ be a graph with radius at most $r$.
    A tuple $(X, Y, Z)$ is called an \nhoodpartition of $\hG$ if
    \begin{enumerate}
        \item\label{prop:nhooditemPartition} the sets $X,Y,Z$ are pairwise disjoint and their union is $V(\hG)$.
        \item\label{prop:nhooditemSizeX}
            $|X| \le \t^\y$,
        \item\label{prop:nhooditemSizeY}
            every connected component of $\hG[Y]$ has size at most $r\y^7$
            and at most $\y$ neighbors in~$X$,
        \item\label{prop:nhooditemZTree}
            every connected component of $\hG[Z]$ is a tree with at most one edge to $\hG[X \cup Y]$.
        \item\label{prop:nhooditemYZBoundaries}
            For a subgraph $H$ of $\hG[Y \cup Z]$ we say $N^{\hG}(V(H)) \cap X$ is the boundary of $H$.
            The connected components of $\hG[Y]$ may have at most $\t^\y$
            distinct boundaries, i.e.,
            $|\{ N^{\hG}(V(H)) \cap X \mid H \text{ connected component of } \hG[Y \cup Z] \}| \le \t^\y$,
    \end{enumerate}
    A graph for which an \nhoodpartition exists is called \nhoodpartitionable.
\end{definition}

Property~\ref{prop:nhooditemSizeY} and~\ref{prop:nhooditemZTree} 
enforce that the components of $\hG[Y \cup Z]$ are protrusions.
Later, we will transform \nhoodpartitions into equivalent graphs
of bounded size by replacing the protrusions with small graphs.
Thus, Property~\ref{prop:nhooditemSizeX}
and~\ref{prop:nhooditemYZBoundaries} are there to ensure the
resulting kernelized graph will 
have size roughly $\t^\y$
(without Property~\ref{prop:nhooditemYZBoundaries} we could only
guarantee a size of roughly $\t^{\y^2}$).


{\it
    To simplify our proofs, we fix some notation which
    will be valid for this whole section.
    Let a graph $G$ and $\t,r,\y \in \N^+$ be fixed.
    We further assume $G$ to be \partitionable
    and we fix a \partition $(A,B,C)$ of $G$.
    Let further $\hG$ be an $r$-neighborhood in $G$ 
    and let
    $\hA = A \cap V(\hG)$,
    $\hB = B \cap V(\hG)$,
    $\hC = C \cap V(\hG)$.
}

The \OOnhoodpartition $(X,Y,Z)$ of $\hG$
will be created by building $X$ from $\hA \cup \hB$ and
some vertices from $\hC$. The remaining vertices from $\hC$ will be split into the
sets $Y$ and $Z$.
The remainder of this section will describe how this procedure happens in detail.

\subsection{Neighborhoods of \Partitionable Graphs}


We will start with the straight-forward result that
Properties~\ref{prop:ABCExcess} and~\ref{prop:ABCEdgesToA} of
a \partition (Definition~\ref{def:partition}) can be transferred to
neighborhoods.
\begin{lemma}\label{lem:partitionDefNhood}
    Every $\PartitionRadiusExtraEdges$-neighborhood in $G[\hB \cup
    \hC]$ has an \excess of at most $\y^2$, and every
    $\PartitionRadiusToA$-neighborhood in $G[\hC]$ has at most $\y$
    edges to~$\hA$.
\end{lemma}
\begin{proof}
    Since $\hC \subseteq C$, an $r$-neighborhood in $G[\hC]$ is a
    connected subgraph of an $r$-neighbor\-hood in $G[C]$. Since $G$ is
    \partitionable, a $\PartitionRadiusToA$-neighborhood in $G[C]$ has at most $\y$
    edges to $A$ and $\hA \subseteq A$. Therefore, a $\PartitionRadiusToA$-neighborhood
    of $G[\hC]$ has at most $\y$ edges to~$\hA$.

    Similarly, a $\PartitionRadiusExtraEdges$-neighborhood in $G[\hB \cup \hC]$ is a connected
    subgraph of a $\PartitionRadiusExtraEdges$-neighbo\-rhood in $G[B \cup C]$. If a connected
    graph has an \excess of at most $\y^2$, then so does every
    connected subgraph. Since $G$ is \partitionable, an
    $\PartitionRadiusExtraEdges$-neighborhood in $G[B \cup C]$ has an \excess of at most
    $\y^2$, which bounds the excess of every $\PartitionRadiusExtraEdges$-neighborhood in
    $G[\hB \cup \hC]$.
\end{proof}

\subsection{\Fummlers}

The vertices from $\hA$ and $\hB$ will all be put into the set $X$
of a \OOnhoodpartition.
The situation for the $\hC$ vertices is more complicated.
In this subsection we define so called \emph{\fummlers},
which we use in the next subsection to distribute the vertices $\hC$ 
to the sets $X$, $Y$, and $Z$ of a \nhoodpartition.

\begin{definition}[\Fummler]
    \label{def:fummler}
    Let $W \subseteq \hB \cup \hC$.
    We say $(u_1,u_2,v)$ is a \emph{$W$-\fummler} if $u_1, u_2 \in W$ and
    $v$ lies on a walk $p$ with the following properties:
    Every inner vertex of $p$ is contained in $\hC$ and has at least two neighbors in $p$;
    $u_1$ and $u_2$ are contained only as endpoints of $p$;
    and $p$ is contained in a $\fummlerRadius$-neighborhood in $G[\hB \cup \hC]$.
    We further say $V(p)$ is a \emph{walk set} of $(u_1,u_2,v)$.
\end{definition}

Ties are triples of vertices that are connected by a walk with certain properties.
In the following two lemmas, we bound the size of their walk sets, as well as
the number of \fummlers.  We need this later to prove the size constraints
of a \nhoodpartition.

\begin{lemma}\label{lem:fummlerSize}
    A walk set of a \fummler has at most size $\fummlerSize$.
\end{lemma}
\begin{proof}
    Let the walk set of a \fummler be the vertices on a walk $p$. By definition,
    $p$ is contained in a $\fummlerRadius$-neighborhood in $G[\hB \cup
    \hC]$. Let $T$ be a breadth-first-search spanning tree of such a
    neighborhood. According to Lemma~\ref{lem:partitionDefNhood},
    every $\PartitionRadiusExtraEdges$-neighborhood in $G[\hB \cup
    \hC]$ has an \excess of at most $\y^2$. A tree has an \excess of
    $-1$. Therefore, there are at most $\y^2+1$ edges in $p$ which are
    not contained in $T$. Also, every path in $T$ contains at most
    $2\cdot\fummlerRadius+1$ vertices. Thus, $p$ contains at most
    $(2\cdot\fummlerRadius+1)(\y^2+2) \le \fummlerSize$ vertices.
\end{proof}
Next we take the first step of counting the vertices of $\hC$, by
showing that for $W \subseteq \hB \cup \hC$, 
the number of $W$-\fummlers in $\hG$ is quadratic in $|W|$.
Note that this does not directly lead to
a bound for $|\hC|$ since it might be that $|W|>|\hC|$.
\begin{lemma}\label{lem:fewshortpaths}
    Let $W \subseteq \hB \cup \hC$.
    There are at most $\numFummlers|W|^2$ $W$-\fummlers in $\hG$.
\end{lemma}
\begin{proof}
    We fix $u_1,u_2 \in W$. Let $X_{u_1,u_2}=\{(u_1,u_2,v) \mid v \in
    \hB \cup \hC, (u_1,u_2,v) \text{ is a $W$-\fummler}\}$ be the set
    of all $W$-\fummlers for fixed endpoints $u_1,u_2$. There are
    exactly $|W|^2$ ways to choose $u_1,u_2$, thus, it is sufficient
    to show that $|X_{u_1,u_2}| \le \fummlerSize (\y^2+2)\le
    \numFummlers$.

    Assume for contradiction $|X_{u_1,u_2}|> \fummlerSize (\y^2+2)$.
    For every $x \in X_{u_1,u_2}$ let $V(x)$ be a walk set of $x$.
    The size of a walk set of a \fummler is at most $\fummlerSize$
    (Lemma~\ref{lem:fummlerSize}).
    Let $l=\y^2+3$.
    By a pigeonhole argument, one can choose
    $l$-many $W$-\fummlers $x_1,\dots,x_{l} \in X_{u_1,u_2}$ such that
    $V(x_i) \setminus (V(x_1) \cup \dots \cup V(x_{i-1})) \neq
    \emptyset$ for $1 \le i \le l$. We define for $1 \le i \le l$ a
    graph $G_i = G[V(x_1) \cup \dots \cup V(x_i)]$. We show by
    induction that $G_l$ has an \excess of at least $l-2$.

    By Definition~\ref{def:fummler}, the graphs $G[V(x_i)]$ are connected and
    all vertices except for $u_1,u_2$ have degree at least two in $G[V(x_i)]$.
    That means $G_1=G[V(x_1)]$ has an \excess of at least $-1$.
    It also means that every vertex in the non-empty set
    $V(G_i) \setminus V(G_{i-1})$ has degree at least two in $G_i$.
    Since $G_i$ is connected, there is at least one edge between
    $V(G_i) \setminus V(G_{i-1})$ and $V(G_{i-1})$ in $G_i$.
    If every vertex in
    $V(G_i) \setminus V(G_{i-1})$ has degree exactly two in $G_i$
    there are at least two edges between
    $V(G_i) \setminus V(G_{i-1})$ and $V(G_{i-1})$ in $G_i$.
    Thus, in the step from $G_{i-1}$ to $G_i$
    the number of added edges is at least one greater than the number of added
    vertices. The \excess increases by one.

    The walk set of each $W$-\fummler in $X_{u_1,u_2}$
    contains $u_1$ and is contained in a $\fummlerRadius$-neighborhood in $G[\hB \cup \hC]$.
    This means $G_{l}$ is contained in the
    $(2\cdot\fummlerRadius)$-neighborhood in $u_1$ in $G[\hB \cup \hC]$.
    The graph $G_l$ has an \excess of at least $l-2 = \y^2+1$ and
    according to Lemma~\ref{lem:partitionDefNhood}, this is a contradiction.
\end{proof}

\subsection{Partitioning $\boldsymbol \hC$ into $\boldsymbol \hCA$,
$\boldsymbol \hCB$, $\boldsymbol \hCY$, and $\boldsymbol \hCZ$}

We use the notion of \fummlers (Definition~\ref{def:fummler}) to partition
the set $\hC$. We distinguish vertices connected to $\hA$, vertices connected to
$\hB$ (but not to $\hA$), those which are connected to neither but lie
on a \fummler, and the rest.
We set
\begin{itemize}
    \item $\hCA = N(\hA) \cap \hC$,
    \item $\hCB = (N(\hB) \setminus N(\hA)) \cap \hC$,
    \item $\hCY =  \{v \mid v\in \hC \setminus (\hCA \cup \hCB)$
            and there exist $u_1,u_2 \in \hCA \cup \hCB$ such that
            $(u_1,u_2,v)$ is a $(\hCA \cup \hCB)$-\fummler$\}$,
    \item $\hCZ = \hC \setminus (\hCA \cup \hCB \cup \hCY)$.
\end{itemize}
%
We will show that the four previously defined sets have desirable structural
properties.
By Definition~\ref{def:partition}, we know that $|\hCA| \le a$ and $|\hCB| \le a^\y$.
We use the previously defined \fummlers
to show that $G[\hCZ]$ is a forest and
to bound the size of components of $G[\hCA \cup \hCB \cup \hCY]$.
We then use these properties to construct a \nhoodpartition.
We start with two auxiliary lemmas.
\begin{lemma}\label{lem:smalldistanceBC}
    In $G[\hC]$, every vertex has distance at most $2r$ to a vertex in $\hCA \cup \hCB$.
\end{lemma}
\begin{proof}
    We fix a vertex $v \in \hC$.
    The graph $\hG$ has radius at most $r$,
    thus, $v$ has in $\hG$ distance at most $2r$ to $\hCA \cup \hCB$.
    Vertices in $\hC \setminus (\hCA \cup \hCB)$ are in $\hG$ only adjacent to other vertices from $\hC$.
    This means, the distance from $v$ to the nearest vertex in $\hCA \cup \hCB$ is
    in $G[\hC]$ the same as in $\hG$.
\end{proof}

\begin{lemma}\label{lem:smallradius}
    A connected component of $G[\hC]$ with at most $l \in \N$ vertices from $\hCA \cup \hCB$
    is contained in a $5rl$-neighborhood in $G[\hC]$.
\end{lemma}
\begin{proof}
    Let $\tG$ be a connected component of $G[\hC]$.
    By Lemma~\ref{lem:smalldistanceBC}, every vertex from $(\hCY \cup \hCZ) \cap V(\tG)$ has
    distance at most $2r$ from $(\hCA \cup \hCB) \cap V(\tG)$ in $G[\hC]$. Therefore, $\tG$
    is contained in a $(4r+1)l$-neighborhood in $G[\hC]$. We have $(4r+1)l
    \le 5rl$.
\end{proof}

\subsection{Components of $\boldsymbol G \boldsymbol [ \boldsymbol \hCZ \boldsymbol ]$ are Trees}
In this section we show the somewhat surprising property that if you
take away all vertices that are connected to $A\cup B$ and those that
lie on a \fummler, you are left with a forest.

\begin{lemma}\label{lem:compsAreTrees}
    Each connected component of $G[\hCZ]$ is a tree
    and has at most one outgoing edge in $\hG$.
\end{lemma}
\begin{proof}
    We consider a connected component $H$ of $G[\hCZ]$. Assume for
    contradiction that either $H$ is not a tree, or has more than one
    outgoing edge in $\hG$. Then there has to exist a walk $p$ in
    $\hG$ whose inner vertices are in $V(H)$, whose endpoints are in
    $V(\hC) \setminus \hCZ$, and every inner vertex of $p$ has at
    least two different neighbors in $p$. We pick an arbitrary inner
    vertex $v \in V(H)$ from $p$.
    In this proof, we will successively construct walks $p'$, $p''$
    and $p^*$ with endpoints $(w'_1,w'_2)$, $(w''_1,w''_2)$ and
    $(w^*_1,w^*_2)$ which contain $v$. The final walk $p^*$ will be
    such that $(w^*_1,w^*_2,v)$ is a $(\hCA \cup \hCB)$-\fummler. This
    means by definition that $v \in \hCY$ and therefore $v \not\in
    \hCZ$ (a contradiction).

    \textbf{Constructing \boldmath $p'$:} Let $w_1,w_2$ be the endpoints of $p$.
    Since $\hCA \cup \hCB$ separates $\hCZ$ from $\hA \cup \hB$ in
    $\hG$, we know that $w_1,w_2 \in \hCY \cup \hCA \cup \hCB$. If all
    $w_i \in \hCA \cup \hCB$, we set $p'=p$. If any $w_i \in \hCY$
    then, by definition, $w_i$ lies on a $(\hCA \cup
    \hCB)$-\fummler walk $p_i$. By definition the walk $p_i$ contains
    no vertex from $V(H)$, since this would imply that said vertex is
    in $\hCY$. We modify $p$ into $p'$ as follows: At every endpoint
    $w_i \in \hCY$ we extend $p$ by traversing $p_i$ in an arbitrary
    direction until we reach an endpoint $w_i'\in\hCA \cup \hCB$ and
    then iteratively removing vertices with degree one that might have
    been introduced.
    Now $p'$ is a walk from $w_1'$ to $w_2'$, that goes over $v$ and where every
    inner vertex of $p'$ has at least two different neighbors in $p'$.
    Furthermore, $w'_1,w'_2 \in \hCA \cup \hCB$. However, $p'$ is
    still not necessarily a \fummler-walk, since it is not guaranteed
    to be contained in a $\fummlerRadius$-neighborhood in $G[\hB \cup
    \hC]$.

    \textbf{Constructing \boldmath $p''$:} We construct a sub-walk $p''$ of $p'$
    by starting at $v$ and traversing $p'$ in both directions until we
    either reach an endpoint in $\hCA \cup \hCB$ or a vertex with
    distance exactly $2r+1$ in $G[\hC]$ to $v$. The walk $p''$
    contains $v$ and every vertex on $p''$ has distance at most $2r+1$
    in $G[\hC]$ from $v$. The endpoints $w''_1,w''_2$ of $p''$ are
    either in $\hCA\cup\hCB$ or have distance exactly $2r+1$ in
    $G[\hC]$ from $v$. Every inner vertex of $p''$ has at least two
    different neighbors in $p''$.

    \textbf{Constructing \boldmath $p^*$:} At last, we extend $p''$ into $p^*$
    as follows: If $w''_i \in \hCA \cup \hCB$, we set $w^*_i=w''_i$.
    Otherwise, by Lemma~\ref{lem:smalldistanceBC}, there exists a
    vertex $w^*_i \in \hCA \cup \hCB$ with distance at most $2r$ in
    $G[\hC]$ from $w''_i$. Let $q_i$ be the shortest path from $w^*_i$
    to $w''_i$ in $G[\hC]$. The vertex $v$ has in $G[\hC]$ distance
    exactly $2r+1$ from $w''_i$, thus, $v$ is not contained in $q_i$.
    We traverse $p''$ from $v$ in both directions. While traversing in
    direction of $w''_i$, as soon as we reach a vertex from $q_i$ we
    continue traversing $q_i$ until we reach $w^*_i$.
    The walk $p^*$ contains $v$,
    and every inner vertex has at least two neighbors on $p'$. Also,
    every vertex has distance at most $4r+1$ in $G[\hC]$ from $v$. The
    endpoints $w^*_1,w^*_2$ are contained in $\hCA \cup
    \hCB$. This means that $(w^*_1,w^*_2,v)$ is a $(\hCA \cup
    \hB)$-\fummler.
\end{proof}

\subsection{Connected Components of \boldmath $G[\hC]$}
In this subsection we will speak only about connected components of
$G[\hC]$. While their number is unbounded we show that
inside a component the number of vertices that are not from $\hCZ$ will
be bounded. For every component, we first show that if it has few
vertices from $\hCB$, it has few vertices from $\hCA$
(Lemma~\ref{lem:CAbound}) and that a component with few edges to
$\hB$ has few vertices from $\hCA \cup \hCB \cup \hCY$
(Lemma~\ref{lem:smalliffewtoB}).

\begin{lemma}\label{lem:CAbound}
    A connected component of $G[\hC]$ with $l \in \N$ vertices from $\hCB$
    contains at most $(l+1)\y$ vertices from $\hCA$.
\end{lemma}
\begin{proof}
    Let $\tG$ be a connected component of $G[\hC]$ and let
    $\tCA = \hCA \cap V(\tG)$, 
    $\tCB = \hCB \cap V(\tG)$,  
    $\tCY = \hCY \cap V(\tG)$,
    $\tCZ = \hCZ \cap V(\tG)$.
    We show that $|\tCA| \le (|\tCB|+1)\y$ in two steps:
    At first we show that if
    $|\tCA| > (|\tCB|+1)\y$ then there exists a connected subgraph
    $H$ of $\tG$ which
    contains at least $\y+1$ vertices from $\tCA$ and at most one vertex from $\tCB$.
    Second, we show that such a subgraph $H$ cannot exist.

    Assume that $|\tCA| > (|\tCB|+1)\y$.
    If $\tCB = \emptyset$ we set $H = \tG$.
    Then $H$ contains at least $\y+1$ vertices from $\tCA$ and no vertex from $\tCB$.
    If $\tCB \neq \emptyset$ we proceed as follows.
    For every $v \in \tCB$
    we define $A(v)$ to be the set of all vertices from $\tCA$
    that are reachable from $v$
    in $G[\tCA \cup \tCY \cup \tCZ \cup \{v\}]$.
    Since $\tG$ is connected and
    $\tCB \neq \emptyset$, for all $w\in\tCA$ exists $v\in\tCB$ with $w\in A(v)$.
    This means
    $|\tCA| \le \sum_{v \in \tCB} |A(v)|$.
    Since $|\tCA| > (|\tCB|+1)\y$, there exists
    $v \in \tCB$ with $|A(v)| > \y$.
    Let $H$ be the connected component of $v$
    in $G[\tCA \cup \tCY \cup \tCZ \cup \{v\}]$.
    The graph $H$ is connected and contains exactly one vertex from $\tCB$.
    Since $|A(v)| > \y$, it also contains at least $\y+1$ vertices from $\tCA$.

    We now show that such a graph $H$ cannot exist.
    Let $H'$ be a connected subgraph of $H$ which contains exactly
    $\y+1$ vertices from $\tCA$ and at most one vertex from $\tCB$ (we can construct
    $H'$ by taking a spanning tree of $H$ and iteratively removing leaves until we have exactly
    $\y+1$ vertices from $\tCA$).
    The graph $H'$ contains $\y+1$ vertices from $\hCA$ and at most one vertex from $\hCB$.
    According to Lemma~\ref{lem:smallradius},
    $H'$ is contained in a $5r(\y+2)$-neighborhood in $G[\hC]$. Furthermore $H'$ has by
    construction at least $\y+1$ edges to $\hA$.
    Since $G$ is an \partition every $\PartitionRadiusToA$-neighborhood in
    $G[\hC]$ has, by Lemma~\ref{lem:partitionDefNhood}, at most $\y$ edges to $\hA$.
    This is a contradiction, so $H$ cannot exist.
\end{proof}
\begin{lemma}\label{lem:smalliffewtoB}
    A connected component of $G[\hC]$ with $l \in \N$ edges to $\hB$ contains at most
    $\BEdgeCompSize{l}$ vertices from $\hCA \cup \hCB \cup \hCY$.
\end{lemma}
\begin{proof}
    Let $\tG$ be a connected component of $G[\hC]$ and let
    $\tCA = \hCA \cap V(\tG)$, 
    $\tCB = \hCB \cap V(\tG)$,  
    $\tCY = \hCY \cap V(\tG)$,
    $\tCZ = \hCZ \cap V(\tG)$.
    Since $\tG$ has $l$ edges to $\hB$ we have $|\tCB| \le l$.
    According to Lemma~\ref{lem:CAbound}, $|\tCA| \le (l+1)\y$.
    Let $v \in \tCY$.
    By definition, there exists a $(\hCA \cup \hCB)$-\fummler $x=(u_1,u_2,v)$.
    Since $\tG$ is a connected component,
    $x$ is also a $(\tCA \cup \tCB)$-\fummler.
    By Lemma~\ref{lem:fewshortpaths}, $\hG$ contains at most
    $\numFummlers(\tCA \cup \tCB)^2 \le \numFummlers(l+(l+1)\y)^2$ many
    $(\tCA \cup \tCB)$-\fummlers, which also bounds the number
    of vertices in $\tCY$.
    We add up the bounds for the number of vertices from $\hCB$, $\hCA$, and $\hCY$
    and get $l + (l+1)\y + \numFummlers(l+(l+1)\y)^2 \le \BEdgeCompSize{l}$.
\end{proof}

For components that only have one edge to $\hB$ we can directly say
how many edges to $\hA$ it has.
\begin{lemma}\label{lem:edgesToA}
    A connected component of $G[\hC]$ with at most one edge to $\hB$ has at most $\y$ edges to $\hA$.
\end{lemma}
\begin{proof}
    Let $\tG$ be a connected component of $G[\hC]$ with at most one edge to $\hB$
    and therefore at most one vertex from $\hCB$.
    According to Lemma~\ref{lem:CAbound}, it contains
    at most $2\y$ vertices from $\hCA$.
    By Lemma~\ref{lem:smallradius}, $\tG$
    is contained in a $5r(2\y+1)$-neighborhood in $G[\hC]$.
    By Lemma~\ref{lem:partitionDefNhood}, every
    $\PartitionRadiusToA$-neighborhood can only have at most $\y$ edges to~$\hA$.
\end{proof}

\subsection{Connected Components of \boldmath $G[\hC]$ With More Than One Edge to~$\boldsymbol \hB$}
In this subsection we want to look at components that have more than one
edge to $\hB$. We start with a helping lemma, that states that for
every vertex there is a close vertex from $\hCB$ (or none at all).

\begin{lemma}\label{lem:newvertexfromb}
    Let $v \in \hC$.
    If a vertex $u \in \hCB$ with $u\neq v$ is reachable from $v$ in $G[\hC]$
    then there also is a vertex $w \in \hCB$ with $w\neq v$ that has in $G[\hC]$ distance at most $\pathlengthToCb$ from $v$.
\end{lemma}
\begin{proof}
    We can assume that the shortest path from $u \in \hCB$ to $v$ in $G[\hC]$
    has length at least $\pathlengthToCb$ (otherwise let $w = u$).
    We pick vertices $x_1,\dots, x_{\y+2}$ along this path,
    such that $x_i$ has distance $5ri$ from $v$ in $G[\hC]$.
    Therefore, $x_i$ has distance at least $5r$ from $x_j$ in $G[\hC]$ for $i\neq j$.
    For every $x_i$ there exists a vertex $s_i\in \hCA \cup \hCB$ with distance
    at most $2r$ in $G[\hC]$ from $x_i$ (Lemma~\ref{lem:smalldistanceBC}).
    Since the vertices $x_i$ are spaced sufficiently far apart,
    we have $s_i \neq s_j$, and $v \neq s_i$ for $i \neq j$.
    Each vertex $s_i$ has in $G[\hC]$ distance at most $5ri+2r\le 5r(\y+2)+2r \le \pathlengthToCb$ from $v$.
    If $s_i \in \hCB$  for some $i$ we set $w=s_i$ and
    there is a path in $G[\hC]$ from $v$ to
    $w$ of length at most $\pathlengthToCb$.
    Assume now $s_i \in \hCA$ for all $i\le \y+2$.
    The vertices $s_i$ are contained in the $\pathlengthToCb$-neighborhood
    of $v$ in $G[\hC]$
    and each vertex $s_i$ has one edge to $\hA$.
    In total, there are at least $\y+2$ edges to $\hA$.
    According to Lemma~\ref{lem:partitionDefNhood},
    every $\PartitionRadiusToA$-neighborhood in $G[\hC]$ has at
    most $\y$ edges to $A$. This is a contradiction.
\end{proof}

As stated earlier \fummlers are our tool of choice that we use to
count vertices. We will establish this in the following lemma that
shows that for every edge a component has to $\hB$ one introduces more
\fummlers. This will in turn bound the number of vertices in components
with more than one edge to $\hB$.
\begin{lemma}\label{lem:manyTriplesBadComp}
    Let $\tG$ be a connected component of $G[\hC]$ with $l \ge 2$ edges to $\hB$.
    There are at least $l$ many $\hB$-\fummlers of the form $(u_1,u_2,v)$
    with $v \in V(\tG)$.
\end{lemma}
\begin{proof}
    Let $u_1v$ be an edge between $\hB$ and $\tG$ with $u_1 \in \hB$ and $v
    \in V(\tG)$. Since $l\ge 2$, there has to be another edge $u_2'w$ between
    $\hB$ and $\tG$ with $u_2' \in \hB$ and $w \in V(\tG)$. If $w=v$, it follows
    $u_2' \neq u_1$ and $(u_1,u_2',v)$ is a $\hB$-\fummler.
    Otherwise, $w\neq v$ and
    since $\tG$ is a connected component, $w$ is reachable from $v$ in $\tG$.
    According to Lemma~\ref{lem:newvertexfromb},
    there also is a vertex $w' \in \hCB \cap V(\tG)$ with $w'\neq v$ which has in $\tG$
    distance at most $\pathlengthToCb$ from $v$.
    Since $w' \in \hCB$, $w'$ also has a neighbor $u_2 \in \hB$.
    There is a path from $u_1$ to $u_2$ which contains $v$,
    whose inner vertices are contained in $\tG$,
    and which has length at most $\pathlengthToCb+2$.
    This means $(u_1,u_2,v)$ is a $\hB$-\fummler.

    For each of the $l$ edges between $\tG$ and $\hB$ we can use the
    technique above to construct a $\hB$-\fummler. The first and third
    entry of the tuple correspond to an edge between $\tG$ and $\hB$
    and thus no two edges create the same \fummler.
\end{proof}

With the next lemma we show that a connected component in $G[\hC]$ with many edges to
$\hB$ has many paths with certain properties and then show that only
$\t^{O(\y)}$ many vertices from $\hCA \cup \hCB \cup \hCY$ are in a
connected component of $G[\hC]$ with more than one edge to $\hB$.

\begin{lemma}\label{lem:fewbadcomponents}
    The number of vertices in $\hCA \cup \hCB \cup \hCY$ which are in a
    connected component of $G[\hC]$ with more than one edge to $\hB$ is at most
    $\verticesInBadComponents$.
\end{lemma}
\begin{proof}
    Let $H_1,\dots,H_m$ be the connected components of $G[\hC]$ with more than one edge to $\hB$.
    Let $k_i$ be the number of vertices in $H_i$ which are from $\hCA \cup \hCB \cup \hCY$.
    Let $l_i$ be the number of edges to $\hB$ in $H_i$.
    Let $k=\sum_{i=1}^m k_i$ and $l=\sum_{i=1}^m l_i$.
    At first, we show that $k \le 6400\y^7rl^2$.
    Then, we show that $l \le \numFummlers \t^{2\y}$.
    Together, this yields $k = \verticesInBadComponents$.

    According to Lemma~\ref{lem:smalliffewtoB}, each connected component $H_i$ contains at most
    $\BEdgeCompSize{l_i}$ vertices from $\hCA \cup \hCB \cup \hCY$.
    We bound
    $k \le \sum_{i=1}^m \BEdgeCompSize{l_i} \le 6400\y^7r (\sum_{i=1}^m l_i)^2
    = 6400\y^7rl^2$.
    By Lemma~\ref{lem:manyTriplesBadComp},
    for $1 \le i \le m$
    there are at least $l_i$ many $\hB$-\fummlers $(u_1,u_2,v)$ with $v \in V(H_i)$, so
    in total, there are at least $l$ many $\hB$-\fummlers.
    With Lemma~\ref{lem:fewshortpaths} and $|\hB| \le a^\y$, we bound
    $l \le \numFummlers|\hB|^2 \le \numFummlers \t^{2\y}$.
\end{proof}

\subsection{Protrusion Decomposition}
Having analyzed the structure of $G[\hC]$ we can finally show that for
every \partitionable graph $G$, every $r$-neighborhood $\hG$ is
\OOnhoodpartitionable.
\begin{theorem}
    \label{thm:partitionIsNhoodpartition}
    Let $\t,r,\y \in \N^+$ and let $G$ be an \partitionable graph.
    Let $\hG$ be an $r$-neighborhood in $G$.
    Then $\hG$ is \OOnhoodpartitionable.
\end{theorem}
\begin{proof}
    Let $\hA$, $\hB$, $\hC$, $\hCA$, $\hCB$, $\hCY$, $\hCZ$ be as defined earlier.
    We need to define sets $(X,Y,Z)$ and show all the
    properties of Definition~\ref{def:nhoodpartition}.
    We define $X$ to be the union of $\hA$, $\hB$ and all vertices from $\hCA
    \cup \hCB \cup \hCY$ which are in a connected component of $G[\hC]$ with more than
    one edge to $\hB$.
    Since $(A,B,C)$ is an \partition, we know that
    $|\hA| \le \t$ and $|\hB| \le \t^\y$.
    Lemma~\ref{lem:fewbadcomponents} bounds the number of vertices from
    $\hCA \cup \hCB \cup \hCY$ which are in a connected component of $G[\hC]$ with more than
    one edge to $\hB$ by at most
    $\verticesInBadComponents$.
    This implies $|X| = O(\y^{17}r^3 \t)^{O(\y)}$
    (Property~\ref{prop:nhooditemSizeX})

    We define $Y$ to be the vertices from $\hCA \cup \hCB \cup \hCY$
    which are in a connected component of $G[\hC]$ with at most one edge to
    $\hB$. Each connected component of $G[Y]$ is contained in a connected component of $G[\hC]$
    with at most one edge to $\hB$. Thus, by Lemma~\ref{lem:smalliffewtoB},
    connected components of $G[Y]$ have size at most
    $\BEdgeCompSize{1} = O(\y^7r)$.
    Every connected component of $G[Y]$ has at most one edge to $\hB$
    and by Lemma~\ref{lem:edgesToA} at most $\y$ edges to $\hA$.
    By construction, every edge from it to $X$ goes either to $\hA$ or $\hB$.
    This means it has at most $\y+1$ neighbors in $X$
    (Property~\ref{prop:nhooditemSizeY}). Since $|\hA| \le \t$ and $|\hB| \le \t^\y$ there are at most
    $\t^{O(\y)}$ (choose at most $\y$ of $\t$ and at most one of $\t^\y$ vertices)
    possible sets of boundaries in $X$.
    This satisfies Property~\ref{prop:nhooditemYZBoundaries}.

    We define $Z = \hCZ$.
    According to Lemma~\ref{lem:compsAreTrees},
    every connected component of $G[Z]$ is a tree
    and has at most one edge to $X \cup Y$ (Property~\ref{prop:nhooditemZTree}).
    Finally, the sets $X,Y,Z$ are pairwise disjoint
    and their union is $V(\hG)$ (Property~\ref{prop:nhooditemPartition}).
\end{proof}

%

\section{Compressing Neighborhoods} \label{sec:kernel}

Earlier, 
(Theorem~\ref{thm:partitionProbability}, Theorem \ref{thm:partitionIsNhoodpartition})
we showed that neighborhoods of $\scale$-\dominated random graph models
are (for certain values of $\scale$, $\t$, $r$, $\y$) likely to be \nhoodpartitionable (Definition~\ref{def:nhoodpartition}).
This means these neighborhoods have the following nice structure:
They consist of a (small) core graph to which protrusions are attached.
Remember that protrusions are (possibly large) subgraphs with small treewidth and boundary
and that the \emph{boundary} of a subgraph is its neighborhood in the remaining graph.

In this section, we replace these protrusions
by subgraphs with bounded size that retain the same boundary.
This yields a small graph which is $q$-equivalent
to the original graph.
The same technique has been used for obtaining small kernels in larger graph
classes, \eg, in graphs that exclude a fixed minor~\cite{FLST2010}.
The main result of this section is the following theorem.

\begin{customthm}{\ref{thm:kernel}}
    There exists an algorithm that takes
    $q,r,\y \in \N^+$ and a connected labeled graph $G$ with radius at most $r$
    and at most $q$ labels as input,
    runs in time at most $f(q,r,\y)\|G\|$ for some function $f(q,r,\y)$,
    and computes a labeled graph $G^* \equiv_q G$.
    If $G$ is \nhoodpartitionable for some $\t \in \N^+$
    then $|G^*| \le f(q,r,\y)\t^{\y}$.
\end{customthm}

This kernelization procedure and its run time bound is independent in $\t$
but the size of the output kernel is not:
If $\t$ is small, then the output is small.
The result is obtained by replacing protrusions with the help of the
Feferman--Vaught theorem~\cite{FV59}.
However, in order to replace the protrusions, one first has to identify them.
The main complication in this section lies in partitioning a graph such
that the relevant protrusions can be easily identified.
It is crucial that we obtain the size bound $|G^*| \le f(q,r,\y)\t^{\y}$ 
in Theorem~\ref{thm:kernel}.
Weaker bounds are easier to obtain but would
not be sufficient for our purposes.

\subsection{Protrusion replacement and the Feferman--Vaught Theorem}

\newcommand{\join}{\textnormal{join}}
\newcommand{\Th}{\textit{Th}_{\textit{FOL}}}
\newcommand{\tp}{\textnormal{tp}^\textnormal{FO}_q}

In this subsection we obtain a suitable protrusion replacement procedure
(Lemma~\ref{lem:protrusionreplace}).
We use a variant of the Feferman--Vaught theorem~\cite{FV59} to replace a
protrusion by a $q$-equivalent boundaried graph of minimal size.  This
size depends only on $q$ and the size of the boundary.
The original Feferman--Vaught theorem states that the validity of FO-formulas
on the disjoint union or Cartesian product of two graphs is uniquely
determined by the value of FO-formulas on the individual graphs.
Makowsky adjusted the theorem for algorithmic
use~\cite{makowsky2004algorithmic} in the
context of MSO model-checking.  The following proposition contains
the Feferman--Vaught theorem in a very accessible form.
There is also a nice and short proof in~\cite{grohe2008logic}.
The notation is borrowed from~\cite{grohe2008logic}, too.
At first, we need to define so called $q$-types.

\begin{definition}[\cite{grohe2008logic}]
    Let $G$ be a labeled graph and $\bar v = (v_1,\dots,v_k) \in V(G)^k$,
    for some nonnegative integer $k$.
    The \emph{first-order $q$-type of $\bar v$ in $G$} is
    the set $\tp(G,\bar v)$ of all first-order formulas $\psi(x_1,\dots x_k)$ of rank at most $q$
    such that $G \models \psi(v_1,\dots,v_k)$.
\end{definition}
A $q$-type could be an infinite set, but one can reduce them to a
finite set by syntactically normalizing formulas, so that there are
only finitely many normalized formulas of fixed quantifier rank and
with a fixed set of free variables. These finitely many formulas can
be enumerated.
For a tuple $\bar u = (u_1,\dots,u_k)$, we write $\{\bar u\}$
for the set $\{u_1,\dots,u_k\}$.
The following is a variant of the Feferman--Vaught theorem~\cite{FV59}.
\begin{proposition}[{\cite[Lemma~2.3]{grohe2008logic}}]\label{thm:fv}
    Let $G,H$ be labeled graphs and $\bar u \in V(G)^k$, such that
    $V(G) \cap V(H) =\{\bar u\}$. Then for all $q \ge 0$, $\tp(G \cup H, \bar u)$
    is determined by $\tp(G, \bar u)$ and $\tp(H, \bar u)$.
\end{proposition}
We use this proposition in the following two lemmas to introduce a
$q$-type preserving protrusion replacement procedure.
%
\begin{lemma}\label{lem:fotype}
    Let $H$ be a connected labeled graph with treewidth at most $t$, 
    at most $q$ labels, and $\bar u \in V(H)^k$ for some $k$. 
    One can find in time $h(q,t,k)|H|$
    a connected labeled graph $H'$ with
    $\{\bar u\} \subseteq V(H') \subseteq V(H)$,
    such that $|H'| \le h(q,t,k)$ and $\tp(H,\bar u) = \tp(H',\bar u)$,
    for some function $h(q,t,k)$.
\end{lemma}
\begin{proof}
    The $q$-type $\tp(H,\bar u)$ can be represented by a set of normalized
    FO-formulas with quantifier rank at most $q$ and $k$ free variables.
    The number and length of these representing formulas
    can be bounded by a function of $q$ and $k$.
    Courcelle's theorem states that for a graph $H$ (with treewidth at most $t$)
    and a formula $\psi$
    (with quantifier rank at most $q$ and $k$ free variables)
    one can decide whether $H \models \psi(\bar u)$
    in time $g(q,t,k)|H|$,
    for some function $g(q,t,k)$.
    This lets us efficiently compute the $q$-type
    $\tp(H,\bar u)$ by checking all representing formulas.

    We now have to find a small graph $H'$ with the same $q$-type as $H$.
    We enumerate all connected graphs whose vertex set is a superset of $\{\bar u\}$ and
    which are labeled using the same labels as $H$
    in ascending order by vertex count.
    For each graph, we compute the $q$-type.
    We finish as soon as we find a graph $H'$ with $\tp(H,\bar u) = \tp(H', \bar u)$.
    Such a graph $H'$ exists.
    Each $q$-type of a $k$-tuple is represented by a subset
    of normalized formulas of rank at most $q$
    and at most $k$ free variables.
    This bounds the number of different $q$-types of $k$-tuples
    by a function of $q$ and $k$.
    Thus, the size of $H'$ can also be bounded by a function of $q$ and $k$.
    Since $|H'| \le |H|$, we can rename the vertices of $H'$
    such that $V(H') \subseteq V(H)$.
\end{proof}
\begin{lemma}\label{lem:protrusionreplace}
    Let $G,H$ be labeled graphs and $\bar u \in V(G)^k$ for some $k$, such that $V(G) \cap
    V(H) =\{\bar u\}$. Let $H$ be connected with treewidth at most $t$ and at most $q$ labels.
    One can find in time $h(q,t,k)|H| $ a connected labeled
    graph $H'$ such that $|H'| \le h(q,t,k)$, $V(G) \cap V(H') = \{\bar
    u\}$, and $G \cup H \equiv_q G \cup H'$, for some function
    $h(q,t,k)$.
\end{lemma}
\begin{proof}
    We use Lemma~\ref{lem:fotype} to construct a connected labeled graph $H'$
    such that $|H'| \le h(q,t,k)$, $\tp(H,\bar u) = \tp(H',\bar u)$,
    and $V(G) \cap V(H') = \{\bar u\}$.
    According to Proposition~\ref{thm:fv},
    $\tp(G \cup H, \bar u)$ is determined
    by $\tp(G, \bar u)$ and $\tp(H, \bar u)$.
    Therefore,
    $\tp(G \cup H, \bar u) = \tp(G \cup H', \bar u)$, which implies
    $G \cup H \equiv_q G \cup H'$.
\end{proof}

\subsection{Reduction Rules}

Let $G$ be an \nhoodpartitionable graph.
We want to construct a graph which is $q$-equivalent to $G$
and small if $\t$ is small.
We know there exists an \nhoodpartition $(X,Y,Z)$ of $G$,
but it is non-trivial to compute it.
In Lemma~\ref{lem:findZ} and \ref{lem:findheavybounaries},
we identify $Z$ and parts of $Y$.
In Lemma~\ref{lem:lastkernelstep} and
Theorem~\ref{thm:kernel}, we replace these parts
using the protrusion-replace
technique from Lemma~\ref{lem:protrusionreplace}.
\begin{lemma}\label{lem:findZ}
    There exists an algorithm that takes
    $r,\y \in \N^+$ and a graph $G$ with radius at most $r$
    as input,
    runs in time $O(\|G\|)$,
    and computes a set $Z \subseteq V(G)$ with the following property:
    If $G$ is \nhoodpartitionable for some $\t \in \N^+$
    then there exists an \nhoodpartition $(X,Y,Z)$ of~$G$.
\end{lemma}
\begin{proof}
    We construct $Z$ iteratively. At first $Z$ is empty, then in each step we
    add a vertex $v\in V(G)$ to $Z$ if it is a degree-one vertex of $G[V(G)
    \setminus Z]$. We repeat until there are no degree-one vertices in $G[V(G)
    \setminus Z]$. This can be done in $O(\|G\|)$ steps.
    The set $Z$ satisfies Property~\ref{prop:nhooditemZTree} of
    Definition~\ref{def:nhoodpartition}.
    Assume that $G$ is \nhoodpartitionable.
    We need to show that $X,Y$ exist such that $(X,Y,Z)$ is an \nhoodpartition
    of $G$.
    We consider an arbitrary
    \nhoodpartition $(X',Y',Z')$ of $G$.
    Every connected component of $G[Z']$ is a tree
    with at most one edge to $V(G) \setminus Z'$. 
    The set $Z$ was constructed such that $Z' \subseteq Z$.
    We set $X = X' \setminus Z$ and $Y = Y' \setminus Z$. The sets
    $X$, $Y$, $Z$ are disjoint and their union is $V(G)$
    (Property~\ref{prop:nhooditemPartition}).
    Since $X \subseteq X'$ and $Y \subseteq Y'$, the tuple $(X,Y,Z)$ satisfies
    Properties~\ref{prop:nhooditemSizeX},~and~\ref{prop:nhooditemSizeY}.

    Let $H$ be a connected component of $G[Y]$.
    By definition of $Z$, there exists a unique component $H'$ of $G[Y']$
    such that $H = H' \setminus Z$.
    Furthermore, $X \subseteq X'$.
    This means the number of distinct boundaries of $G[Y]$ in $X$
    is not larger than 
    the number of distinct boundaries of $G[Y']$ in $X'$.
    Since $(X',Y',Z')$
    satisfies Property~\ref{prop:nhooditemYZBoundaries},
    $(X,Y,Z)$ satisfies Property~\ref{prop:nhooditemYZBoundaries} as well.
\end{proof}

\begin{definition}[Heavy boundary]
    Let $(X,Y,Z)$ be an \nhoodpartition of a graph $G$.
    We call a set $S \subseteq X$ a \emph{heavy boundary}
    if there exist more than $r\y^7+\y$ connected components in $G[Y]$
    whose boundary in $X$ is exactly $S$.
\end{definition}

\begin{lemma}\label{lem:findheavybounaries}
    There exists an algorithm that takes
    $r,\y \in \N^+$ and a graph $G$ with radius at most $r$
    as input,
    runs in time $O(\|G\|)$,
    and computes sets $Z,P\subseteq V(G), \cal S \subseteq 2^{V(G)}$
    with the following properties:
    If $G$ is \nhoodpartitionable for some $\t \in \N^+$
    then there exists an \nhoodpartition $(X,Y,Z)$ of $G$.
    The connected components of $G[Y]$ with a heavy boundary are connected components of $G[P]$.
    The set $\cal S$ contains subsets of $X$ of size at most $\y$
    and $|\cal S| \le \min(2\t^\y,|G|)$.
    Every heavy boundary of $(X,Y,Z)$ is contained in~$\cal S$.
\end{lemma}
\begin{proof}
    Let $\t \in \N^+$ such that $G$ is \nhoodpartitionable.
    We use Lemma~\ref{lem:findZ} to construct in time $O(|G|)$
    a set $Z$ such that there exists an \nhoodpartition $(X,Y,Z)$ of $G$.

    Let $P$ be the set of all vertices with degree at most $r\y^7+\y$ in $G[V(G) \setminus Z]$.
    The set $P$ can be computed in $O(\|G\|)$.
    A vertex $v \in Y$ is contained in a connected component of $G[Y]$ of size at most $r\y^7$
    with at most $\y$ neighbors in $X$ (Property~\ref{prop:nhooditemSizeY}).
    Therefore $v$ has degree at most $r\y^7+\y$ in $G[V(G) \setminus Z]$,
    which implies $Y \subseteq P$.
    Let $H$ be a connected component of $G[Y]$ with the heavy boundary $S$.
    There are more than $r\y^7+\y$ connected components in $G[Y]$ with boundary $S$.
    Therefore, every vertex in $S$ has degree more than $r\y^7+\y$ in $G[V(G) \setminus Z]$
    and thus $S \cap P = \emptyset$.
    This and $V(H) \subseteq P$ imply that $H$ is a connected component of $G[P]$, i.e.,
    the connected components of $G[Y]$ with a heavy boundary are connected components of $G[P]$.

    Let $\cal S$ be the set of all subsets of $V(G) \setminus (P \cup Z)$
    with size at most $\y$ which are the boundary
    of some connected component of $G[P]$ in $V(G) \setminus (P \cup Z)$.
    The set $\cal S$ can be computed in time $O(\|G\|)$.
    Note that $|\cal S| \le |G|$.
    Let $S \in \cal S$.
    Since $S \cap (P \cup Z) = \emptyset$ and $Y \subseteq P$, we have $S \subseteq X$.
    For every heavy boundary of $(X,Y,Z)$, there exists a connected component of $G[P]$
    with this boundary in $V(G) \setminus (P \cup Z)$.
    Boundaries of $(X,Y,Z)$ have by definition
    (Property~\ref{prop:nhooditemSizeY}) size at most $\y$.
    Thus, every heavy boundary of $(X,Y,Z)$ is contained in $\cal S$.

    Since $Y \subseteq P$ and $P \cap Z = \emptyset$, connected
    components of $G[P]$ are either connected components of $G[Y]$ or
    contain a vertex from $X$. Since $|X| \le \t^\y$, there are at
    most $\t^\y$ connected components of $G[P]$ that are not
    connected component of $G[Y]$. Components of $G[Y]$ have by
    definition (Property~\ref{prop:nhooditemYZBoundaries}) at most
    $\t^\y$ distinct boundaries in 
    $V(G) \setminus (P \cup Z)$. The remaining at most $\t^\y$ many
    connected components of $G[P]$ that contain a vertex from $X$ have
    at most $\t^\y$ boundaries in $G[V(G) \setminus (P \cup Z)]$.
    Together, this gives $|\cal S| \le 2\t^\y$.
\end{proof}

\begin{lemma}\label{lem:lastkernelstep}
    There exists an algorithm that takes
    $q,r,\y \in \N^+$ and a connected labeled graph $G$ with radius at most $r$ and at most $q$ labels
    as input,
    runs in time at most $f(q,r,\y)\|G\|$ for some function $f(q,r,\y)$,
    and computes a connected labeled graph $G^* \equiv_q G$ and a set $Z^* \subseteq V(G^*)$.
    In $G^*$, every connected component of $G^*[Z^*]$ is a tree and has at most one neighbor
    in $V(G^*) \setminus Z^*$.
    If $G$ is \nhoodpartitionable for some $\t \in \N^+$
    then $|V(G^*) \setminus Z^*| \le f(q,r,\y)\min(\t^\y,|G|)$.
\end{lemma}
\begin{proof}
    Let $\t \in \N^+$ such that $G$ is \nhoodpartitionable.
    We use Lemma~\ref{lem:findheavybounaries} to
    compute sets $Z, P\subseteq V(G)$, $\cal S \subseteq 2^{V(G)}$.
    There exists an \nhoodpartition $(X,Y,Z)$ of $G$.
    The connected components of $G[Y]$ with a heavy boundary are connected components of $G[P]$.
    The set $\cal S$ contains subsets of $X$ of size at most $\y$
    and $|\cal S| \le \min(2\t^\y,|G|)$.
    Every heavy boundary of $(X,Y,Z)$ is contained in $\cal S$.
    For every $W \subseteq V(G) \setminus Z$, we define
    $Z(W) \subseteq Z$ to be the vertices that are reachable from $W$
    in $G[W \cup Z]$.


    For every $S \in \cal S$ we do the following:
    We compute the set $P_S$ of all vertices which are contained in a
    connected component of $G[P]$ with size at most $r\y^7$ and
    boundary $S$ in $V(G)\setminus (P \cup Z)$. We compute $H_S = G[S \cup
    P_S \cup Z(P_S)]$ and $G_S = G[V(G) \setminus (P_S \cup Z(P_S))]$.
    Notice that $G = H_S\cup G_S$ and $V(H_S) \cap V(G_S) = S$.
    Furthermore, $|S| \le \y$ and if we remove $S$ from $H_S$, the
    remaining graph consists of connected components of size at most
    $r\y^7$ to which trees are attached. This means $H_S$ has
    treewidth at most $t=\y+r\y^7$.
    Also, $H_S$ has at most $q$ labels.
    Lemma~\ref{lem:protrusionreplace} lets us construct
    in time at most $h(q,t,\y)|H_S| $ a graph $H_S'$ such that
    $|H_S'| \le h(q,t,\y)$,
    $V(G_S) \cap V(H'_S) = S$,
    and $G \equiv_q G_S \cup H_S'$. We now replace $H_S$ with $H'_S$.

    This replacement procedure gives us graphs
    $\hat G = G[V(G) \setminus \bigcup_{S\in \cal S} (P_S \cup Z(P_S))]$
    and
    $\hat H = \bigcup_{S \in \cal S} H'_S$
    with $G \equiv_q \hat G \cup \hat H$.
    We set $G^* = \hat G \cup \hat H$.
    Notice that $G^*$ is connected,
    since the described construction preserves connectivity.


    We now bound the run time of this procedure.
    The underlying graph of $G$ can be extracted in time $q\|G\|$.
    Constructing $Z$, $P$, and $\cal S$ by Lemma~\ref{lem:findheavybounaries}
    takes time $O(\|G\|)$.
    The graphs $\hat G$, and $\{H_S \mid S \in \cal S\}$ can be constructed in time
    $O(\sum_{S \in \cal S} \|H_S\|)$.
    Constructing $\hat H$ takes time $h(q,t,\y) \sum_{S \in \cal S} \|H_S\|$.

    For every vertex $v \in P \cup Z$ there exists
    at most one $S \in \cal S$ such that $v \in H_S$. Also $|\cal S|
    \le |G|$. 
    Furthermore, since a graph $H_S$ has treewidth at most $t$,
    $\| H_S \| \le t |H_S|$.
    Therefore,
    $$
    \sum_{S \in \cal S} \|H_S\|
    \le
    t\sum_{S \in \cal S} |H_S|
    = t\sum_{S \in \cal S} |S| + |H_S \cap (P \cup Z)|
    \le t\y|\cal S| + t|P \cup Z|
    \le t\y|G| + t|G|.
    $$
    In total, the whole algorithm runs in time $f(q,r,\y)\|G\|$,
    for some function $f(q,r,\y)$.


    We proceed to show that $|V(G^*) \setminus Z|$ is small.
    Note that
    \begin{equation}\label{eqn:fooo1}
    |V(G^*) \setminus Z| \le |X| + |V(\hat G) \cap Y| + |\hat H|.
    \end{equation}
    We know that $|X| \le \t^\y$.
    Furthermore, with $|\cal S| \le \min(2\t^\y,|G|)$
    \begin{equation}\label{eqn:fooo0}
    |\hat H| \le
    \sum_{S \in \cal S} |H'_S| \le \min(2\t^\y,|G|) h(q,t,\y).
    \end{equation}
    We further bound the number of vertices from $Y$ in $\hat G$.
    Let $\cal X$ be the set of all boundaries in $X$ of connected components of $G[Y]$.
    Let $S \in \cal X$ be a boundary,
    we define $Y_S$ to be the vertices of all connected components of $G[Y]$
    which have $S$ as their boundary.
    We distinguish between $S$ being a heavy or non-heavy boundary:
    Assume $S$ is heavy.
    Then $S \in \cal S$.
    The connected components of $G[Y]$ with a heavy boundary are connected components of $G[P]$
    and have size at most $r\y^7$.
    This means $Y_S \subseteq P_S$.
    The graph $\hat G$ was defined such that $P_S \cap V(\hat G) = \emptyset$,
    and we have $|Y_S \cap V(\hat G)| = 0$.
    Assume now $S$ is non-heavy.
    Thus, $Y_S$ consists of at most $r\y^7+\y$ connected components of $G[Y]$
    of size at most $r\y^7$.
    This means $|Y_S| \le (r\y^7+\y)r\y^7$.

    In total, for every $S \in \cal X$,
    $|V(\hat G) \cap Y_S| \le (r\y^7+\y)r\y^7$.
    Note that $Y = \bigcup_{S \in \cal X} Y_S$ and $|\cal X| \le \t^\y$ (Property~\ref{prop:nhooditemYZBoundaries}).
    We can therefore bound
    \begin{equation}\label{eqn:fooo2}
    |V(\hat G) \cap Y|
    \le \sum_{S \in \cal X} |V(\hat G) \cap Y_S|
    \le \t^\y (r\y^7+\y)r\y^7.
    \end{equation}\label{eqn:foo1}
    Combining (\ref{eqn:fooo1}), (\ref{eqn:fooo0}), (\ref{eqn:fooo2}), and $|X| \le \t^\y$ yields
    $$
    |V(G^*) \setminus Z| \le
    \t^\y + \t^\y(r\y^7+\y)r\y^7 + 2\t^\y h(q,t,\y),
    $$
    which can be bounded by $f(q,r,\y)\t^\y$ for some function
    $f(q,r,\y)$. Furthermore, combining  $|V(G^*) \setminus Z|\le
    |\hat G| + |\hat H|$, (\ref{eqn:fooo0}), and $|\hat
    G| \le |G|$ yields $|V(G^*) \setminus Z| \le f(q,r,\y)|G|$ for
    some function $f(q,r,\y)$.

    Let $Z^* = Z \cap V(G^*)$.
    At last, we need to show that in $G^*$, the connected components of $G^*[Z^*]$ are trees
    with at most one neighbor in $V(G^*) \setminus Z^*$.
    This follows from that fact that in $G$ the connected components of $G[Z]$ are trees with at most one
    neighbor in $X \cup Y$,
    and the graph $G^*$ introduces no new edges to vertices from $Z$.
\end{proof}

\begin{theorem}\label{thm:kernel}
    There exists an algorithm that takes
    $q,r,\y \in \N^+$ and a connected labeled graph $G$ with radius at most $r$
    and at most $q$ labels as input,
    runs in time at most $f(q,r,\y)\|G\|$ for some function $f(q,r,\y)$,
    and computes a labeled graph $G^* \equiv_q G$.
    If $G$ is \nhoodpartitionable for some $\t \in \N^+$
    then $|G^*| \le f(q,r,\y)\t^{\y}$.
\end{theorem}
\begin{proof}
    Let $\t \in \N^+$ such that $G$ is \nhoodpartitionable.
    We use the algorithm of Lemma~\ref{lem:lastkernelstep} to
    compute a connected graph $G' \equiv_q G$ and a set $Z \subseteq V(G')$.
    Let $\bar Z = V(G') \setminus Z$.
    We have $|\bar Z| \le g(q,r,\y)\min(\t^\y,|G|)$,
    for some function $g(q,r,\y)$.
    Also, in $G'$, every connected component of $G'[Z]$
    is a tree and has at most one neighbor in $\bar Z$.
    Since $G' \equiv_q G$, $G'$ also has at most $q$ labels.
    If $\bar Z$ is empty, then $G'$ is a tree with at most $q$ labels
    and we use Lemma~\ref{lem:protrusionreplace}
    to construct in time $h(q,1,1)|G'|$ a graph $G^*$
    with $|G^*| \le h(q,1,1)$ and
    $G^* \equiv_q G' \equiv_q G$.
    We therefore assume $\bar Z \neq \emptyset$.

    For every $v \in \bar Z$ we do the following:
    We define $Z_v$ to be the set of
    vertices which are contained in a connected component of $G'[Z]$ which has
    $v$ as its only neighbor.
    We also define the graph $H_v = G'[\{v\} \cup Z_v]$,
    which is a tree with at most $q$ labels and intersects $G'[\bar Z]$ only in $v$.
    We use Lemma~\ref{lem:protrusionreplace}
    to construct in time $h(q,1,1)|H_v|$ a graph $H_v'$ with
    $|H_v'| \le h(q,1,1)$ and
    $G \equiv_q G'[V(G') \setminus Z_v] \cup H_v'$. We replace the
    subgraph $H_v$ of $G'$ with $H'_v$.

    This gives us a graph $G^* = G'[V(G') \setminus \bigcup_{v \in
    \bar Z} Z_v] \cup \bigcup_{v \in \bar Z} H_v'$ with $G^* \equiv_q
    G' \equiv_q G$. Since $\bar Z \neq \emptyset$ and $G'$ is
    connected $\bigcup_{v \in \bar Z} Z_v = Z$. We bound with $|\bar
    Z| \le g(q,r,\y)\t^\y$ and $v\in V(H_v')$ for all $v\in\bar Z$
    $$
    |G^*| = \sum_{v \in \bar Z} |H_v'|
    \le g(q,r,\y)\t^\y h(q,1,1)
    \le f(q,r,\y)\t^\y,
    $$
    for some function $f(q,r,\y) \ge g(q,r,\y)h(q,1,1)$.
    Notice that the graphs $\{H_v \mid v \in \bar Z \}$
    are disjoint and their union is $G'$.
    The time needed to construct $\hat H$ therefore is at most
    $$
    \sum_{v \in \bar V} h(q,1,1)\|H_v\|
    \le h(q,1,1) \|G'\|
    \le h(q,1,1) g(q,r,\y)\|G\|
    \le f(q,r,\y)\|G\|.
    $$
    As in Lemma~\ref{lem:lastkernelstep}, $f$ can be chosen such that
    the algorithm runs in time $f(q,r,\y)\|G\|$.
\end{proof}


\section{Model-Checking}\label{sec:modelchecking}

In this section, we finally obtain the main result of this paper,
namely that for certain values of $\scale$ 
one can perform model-checking on \dominated random graph models
in efficient expected time.

An important tool in this section is Gaifman's locality theorem~\cite{gaifman1982local}.
It states that first-order formulas can express only local properties of graphs.
It is a well established tool for the design of model-checking algorithms (e.g.~\cite{grohe2001generalized,grohe2008logic,frick2001deciding}).
We use it to reduce the model-checking problem on a graph
to the model-checking problem on neighborhoods of said graph (Lemma~\ref{lem:globalIfLocal}).
This technique is described well by Grohe~\cite[section 5]{grohe2008logic}.

To illustrate our approach, consider the following thought experiment:
Let $X$ be a non-negative random variable with $\P[X = \t] = \Theta(\t^{-10})$ for all $\t \in \N$.
Assume an algorithm that gets an integer $\t \in \N$ as input and runs in time $t(\t)$.
Its expected run time on input $X$ is $\sum_{\t \in \N} \Theta(\t^{-10}) t(X)$.
If $t(\t) = \t^{10}$ then the expected run time is infinite.
If $t(\t) = \t^{8}$ then the expected run time is $\Theta(1)$.
Thus, small polynomial differences in the run time can have a huge impact
on the expected run time.
We notice that
the run time on an input has to grow slower than the inverse of the probability
that the input occurs.

Let us fix a formula $\varphi$ and let $r$ and $\y$ be constants
depending on $\varphi$.
In this section we provide a model-checking algorithm whose run time on a graph $G$ depends
on the minimal value $\t \in \N$ such that $G$ is \partitionable.
This means, we need to solve the model-checking problem on \partitionable graphs
faster than the inverse of the probability that $\t$ is minimal.

\Cref{sec:partitionprob} states that
a graph from \dominated graph classes is for some $\t$ not \partitionable
with probability approximately $\t^{-\y^2}$
(we ignore the terms in $r$, $\y$ and $\momn$ for now).
Thus, the probability that a value $\t$ is minimal is approximately $\t^{-\y^2}$.

Let $G$ be a graph and $a$ be the minimal value such that $G$ is \partitionable.
In \Cref{sec:protrusiondecomp} we showed that all its $r$-neighborhoods are 
\OOnhoodpartitionable.
The kernelization result from \Cref{sec:kernel}
states that 
such $r$-neighborhoods
can be converted in linear time
into $|\varphi|$-equivalent graphs of size approximately $\t^\y$ (we again ignore the factors independent of $\t$ for now).
This means, using the naive model-checking algorithm,
one can decide for an $r$-neighborhood $G^r$ of $G$ whether $G^r \models \varphi$
in time approximately $\|G\| \t^{\y |\varphi|}$.
Thus, one can perform model-checking on all $r$-neighborhoods of $G$
in time approximately $\t^{\y |\varphi|} \sum_v \| N^G_r(v) \|$.
Using Gaifman's locality theorem, this (more or less) yields the answer to the model
checking problem in the \emph{whole} graph.

Let $G$ be a graph from a \dominated random graph model.
In summary, we have for every $\t \in \N$:
\begin{itemize}
    \item $\t \in \N$ is the minimal value such that a graph is \partitionable
    with probability approximately $\t^{-\y^2}$.
    \item
    If $\t \in \N$ is the minimal value such that $G$ is \partitionable 
    then we can decide whether $G \models \varphi$
    in time approximately $\t^{\y |\varphi|} \sum_v \| N^G_r(v) \|$.
\end{itemize}

In this example one may choose $\y = |\varphi|^2$ such that the run time grows
slower than the inverse of the probability.
We changed some numbers in these examples to simplify our arguments.
Thus, in reality, $\y$ needs to be chosen slightly differently.

This section is structured as follows:
In \Cref{sec:mc1}, we introduce the concept of Gaifman locality.
Then, in \Cref{sec:mc2}, we use Gaifman locality and the kernelization
result from \Cref{sec:kernel}
to solve the model-checking in \partitionable graphs.
At last, in \Cref{sec:mc3}, we prove our main result
by showing that the run time of this algorithm grows
slower than the inverse of the probability that $\t$ is minimal.

\subsection{Locality}
\label{sec:mc1}

In this section, we present a well-known technique which reduces
the model-checking problem to local regions.
Lemma~\ref{lem:globalIfLocal} gives
a slightly different version of what can be found in the literature~\cite{grohe2008logic}.
Without this modification we would only be able to prove
expected polynomial time of our model-checking algorithm
instead of expected linear time.

A formula $\omega(x)$ is called \emph{$r$-local} if $G\models\omega(v)$ if and
only if $ G[N^G_r(v)]\models\omega(v)$ for all labeled graphs $G$ and all $v\in V(G)$. Let
$\text{dist}_{> r}(x_1,x_2)$ be the first-order formula denoting that
the distance between $x_1$ and $x_2$ is greater than $r$. Let $\omega$
be an $r$-local formula. A \emph{basic local sentence} is a
sentence of the form
$$
\exists x_1 \dots \exists x_s \big(\bigwedge_{i\neq j} \text{dist}_{> 2r}(x_i,x_j) \wedge \bigwedge_i \omega(x_i) \big).
$$

\begin{proposition}[Gaifman's locality theorem~\cite{gaifman1982local,grohe2008logic}]
    \label{prop:gaifman}
    Every first-order sentence is equivalent to a boolean combination
    of basic local sentences. Furthermore, there is an algorithm that
    computes a boolean combination of basic local sentences equivalent
    to a given first-order sentence.
\end{proposition}

The following lemma uses Gaifman locality
(Proposition~\ref{prop:gaifman}) to reduce model-checking in graphs
to model-checking in neighborhoods of graphs.
The proof is similar to~\cite[Lemma~4.9]{grohe2008logic}.

\begin{lemma}\label{lem:globalIfLocal}
    Let $g$ be a function such that for 
    every graph~$G$ with at most $r$ labels,
    every $r$-neighborhood $H$ of $G$, 
    every $v \in V(H)$, and
    every first-order formula $\varphi(x)$ with $|\varphi| \le r$
    one can decide whether $H \models \varphi(v)$ in time $g(v,G,r)$. 

    There exists a
    function $\rho$ such that
    for every first-order sentence $\varphi$ 
    and every labeled graph $G$ with at most $|\varphi|$ labels
    one can decide whether $G \models \varphi$ in
    time at most $O(\|G\|) + \rho(|\varphi|) \sum_{v \in V(G)}
    g(v,G,\rho(|\varphi|))$.
\end{lemma}
\begin{proof}
    We can reduce the first-order sentence $\varphi$ to a boolean
    combination of basic local sentences $\Psi$ with
    Proposition~\ref{prop:gaifman}. We will independently evaluate each
    basic local sentence $\psi\in\Psi$ in the graph and use the result
    to determine whether $\varphi$ is satisfied. Let
    $$
    \psi = \exists x_1 \dots \exists x_s \big(\bigwedge_{i\neq j}
    \text{dist}_{> 2r}(x_i,x_j) \wedge \bigwedge_i \omega(x_i) \big)
    $$
    be a basic local sentence, where $\omega$ is $r$-local. 
    Let $G$ be a graph with at most $|\varphi|$ labels and $v \in V(G)$. We have
    $G\models\omega(v)$ if and only if
    $G[N_{r}(v)]\models\omega(v)$ for $v \in V$. We compute for all
    $v \in V$ whether $G\models\omega(v)$. By our assumption, this
    can be done in time $\sum_{v \in V(G)} g(v,G,r+|\omega|+|\varphi|)$. Let
    now $W$ be the set of all $v \in V$ such that $G \models
    \omega(v)$. A set of vertices is called an $r$-scattered set if the
    $r$-neighborhoods of all pairs of vertices in this set are
    disjoint. Notice that $G \models \psi$ if and only if there exists
    an $r$-scattered set of cardinality $s$ which is a subset of $W$.
    Therefore, all left to do is to find out whether there is an
    $r$-scattered set $S\subseteq W$ with $|S|\geq s$.

    In time $O(\|G\|)$ we do the following: Construct a graph $H$ that
    consists of all nodes that have distance at most $r$ from~$W$, and construct
    the connected components of $H$.
    For each component $H'$ of $H$ pick a vertex $v \in V(H')$ and perform a
    breadth-first-search in $H'$, starting at $v$.
    This way, we either find out that the radius of $H'$ larger than $12rs$
    or that the diameter of $H'$ is at most $12rs$.

    Let us consider two cases.  First, we verified that there is a
    component of $H$ whose diameter is at least~$12rs$.  Then this
    component must contain a shortest path $p$ of length~$12rs$.
    We constructed $H$ such that the
    $r$-neighborhoods of every vertex $u$ on $p$ contains a vertex from $W$.
    Since there are at least $s$ nodes on $p$ whose $r$-neighborhoods
    are disjoint and each of the neighborhoods contains a
    vertex from~$W$, we know that $W$ contains an $r$-scattered set of 
    size at least~$s$.

    The second case we have to consider is that we verified that all components
    of $H$ have a radius at most $12rs$.  
    Note that $W\subseteq V(H)$. For $u,v\in W$ from
    different components of $H$, the distance between $u$ and $v$ in $G$
    is at least~$2r$. Hence, maximal cardinality $r$-scattered
    subsets of $W$ of the components of $H$ form
    together a maximal cardinality
    $r$-scattered subset of $W$ in~$G$.
    A component $H'$ of $H$ contains an 
    $r$-scattered subset of $W$ of size $l$ iff $H'\models \psi_l$ with
    $$
    \psi_l = \exists x_1 \dots \exists x_l \big(\bigwedge_{i\neq j}
    \text{dist}_{> 2r}(x_i,x_j) \wedge \bigwedge_i \omega(x_i) \big),
    $$
    which we can evaluate in time
    $g(v,G,12rs+|\psi_l|+|\varphi|)$ for some $v \in V(H')$.
    We need to check whether $H' \models \psi_l$ for every component 
    $H'$ of $H$ and $l \in \{1,\dots,s\}$.
    In that way we can compute the maximal size of an $r$-scattered
    subset of $W$ in~$H$ and therefore in~$G$.  

    The complete procedure has to be repeated for each $\psi \in \Psi$. 
    Note that $r$, $s$, $|\psi_l|$, and $|\psi|$ depend only on $\varphi$.
    This means we can choose $\rho$ such that all this can be done in time
    $O(\|G\|) + \rho(|\varphi|)\sum_{v \in V(G)} g(v,G,\rho(|\varphi|))$.
\end{proof}

\subsection{Model-Checking in \Nhoodpartitionable \\ Graphs}
\label{sec:mc2}

We use the kernelization result of \Cref{thm:kernel}
to construct a model-checking algorithm for neighborhoods.

\begin{lemma}\label{lem:asdfasdf}
    There is a function $f(r,\y)$ such that for 
    every $r,\y \in \N^+$,
    every graph~$G$ with at most $r$ labels,
    every $r$-neighborhood $H$ of $G$, 
    every $v \in V(H)$, and
    every first-order formula $\varphi(x)$ with $|\varphi| \le r$
    one can decide whether $H \models \varphi(v)$ in time 
    $f(r,\y)\t^{O(\y r) }\|G[N^G_{2r}(v)]\|$,
    where $\t \in \N^+$ be the minimal value such that $G$ is $\t$-$r$-$\y$-partitionable.
\end{lemma}
\begin{proof}
    We construct a graph $H'$ by adding another label to $H$ that identifies $v$ 
    and construct a sentence $\varphi'$ with $|\varphi'| = O(|\varphi|)$ 
    such that $H' \models \varphi'$ if and only if $H \models \varphi(v)$.
    Let $\t \in \N^+$ be the minimal value such that $G$ is $\t$-$r$-$\y$-partitionable.
    According to Theorem~\ref{thm:partitionIsNhoodpartition},
    $H'$ is \OOnhoodpartitionable.
    We use Theorem~\ref{thm:kernel} to construct in time
    $f'(r,\y)\|H'\|$ a graph $H^*$ with $H^*
    \equiv_{|\varphi'|} H'$ and $|H^*| \le f'(r,\y)\t^{O(\y)}$, for
    some function $f'$. On this smaller structure we can perform the
    naive model-checking algorithm in time
    $O(|H^*|^{r}) = O\bigl(f'(r,\y)^{|\varphi|}\t^{O(\y r)}\bigr)$.
    Furthermore, the radius of $H'$ is at most $r$, thus $\|H'\| \le \|G[N^G_{2r}(v)]\|$.
    We choose $f(r,\y)$ accordingly.
\end{proof}

\begin{lemma}\label{lem:MCpartion}
    Let $\y \in \N^+$.
    There exist functions $\rho$ and $f$ such that
    for every first-order sentence $\varphi$ 
    and every labeled graph $G$ with at most $|\varphi|$ labels
    one can decide whether $G \models \varphi$ in
    time $f(\rho(|\varphi|),\y) \t^{\y\rho(|\varphi|)}\sum_{v \in
    V(G)}\|G[N^G_{\rho(|\varphi|)}(v)]\|$,
    where $\t \in \N^+$ is the minimal value such that $G$ is $\t$-$\rho(r)$-$\y$-partitionable.
\end{lemma}
\begin{proof}
    By Lemma~\ref{lem:globalIfLocal} and \ref{lem:asdfasdf}, 
    there exist functions $\rho'$ and $f'$ such that one can
    decide whether $G \models \varphi$ in time
    $$
    O(\|G\|) + \rho'(|\varphi|) \sum_{v \in V(G)} 
    f'(\rho'(|\varphi|),\y)\t^{O(\y \rho'(|\varphi|)) }\|G[N^G_{2\rho'(|\varphi|)}(v)]\|,
    $$
    where $\t \in \N^+$ is the minimal value such that $G$ is $\t$-$\rho'(r)$-$\y$-partitionable.
    We choose $f$ and $\rho$ sufficiently large.
\end{proof}

\subsection{Model-Checking in \Dominated Random Graph \\ Models}
\label{sec:mc3}

In this section we show that the algorithm from Lemma~\ref{lem:MCpartion}
has efficient expected run time on \dominated random graph models.
Our analysis is based upon two results we established earlier:
First, the run time of the algorithm in Lemma~\ref{lem:MCpartion}
depends on
the minimal value $\t$ such that the input graph is \partitionable.
If it is \partitionable for a small $\t$ the run time is fast.
Secondly, Theorem~\ref{thm:partitionProbability} bounds
for our random graphs the probability
that $\t \in \N$ is the minimal value such that a graph is an \partitionable.
For bigger $\t$ it is more and more unlikely that $\t$ is minimal.
In order to have an efficient expected run time on our random graphs,
the run time of the algorithm needs to grow asymptotically slower in $\t$
than the inverse of the probability that $\t$ is minimal.
In Theorem~\ref{thm:MCpowerlaw1} we show that this is the case.

The run time of the algorithm from Lemma~\ref{lem:MCpartion}
depends not only on $\t$ but also on the sum of the sizes of all neighborhoods
in a graph, which might be quadratic in the worst case.
In order to get almost linear expected run time,
we bound the expectation of this value in Lemma~\ref{lem:condNhood}.
We can now prove our main result.

\begin{theorem}\label{thm:MCpowerlaw1}
    There exists a function $f$ such that one can solve \FOMCGL 
    on every $\scale$-\dominated random graph model 
    in expected time $\momn^{f(|\varphi|)} n$.
\end{theorem}
\begin{proof}
    Let $\Gnn$ be an $\scale$-\dominated random graph model and
    $\varphi$ be a first-order formula.
    We fix a $|\varphi|$-labeling function $L$ and $n \in \N$.
    We consider labeled graphs with vertices $V(\Gn)$
    whose underlying graph is distributed according to $\Gn$,
    and analyze the expected run time
    of the model-checking algorithm from Lemma~\ref{lem:MCpartion}
    on these graphs.

    Let $\rho$ be the function from
    Lemma~\ref{lem:MCpartion} and let $r=\rho(|\varphi|)$ and
    $\y = \rho(|\varphi|)^2 + 100$.
    For every graph $G$ there exists a
    value $\t \in \N^+$ such that $G$ is
    $\t$-$r$-$\y$-partitionable (i.e., by setting $\t =
    |V(G)|$, $A = V(G)$). Let $A_\t$ be the
    event that $\t \in \N^+$ is the minimal value such that $\Gn$ is
    $\t$-$r$-$\y$-partitionable and let $R$ be the
    expected run time of the model-checking algorithm from
    Lemma~\ref{lem:MCpartion}.
    The expected run time of the algorithm is exactly 
    $\sum_{\t=1}^\infty \E[R \mid A_\t] \P[A_\t]$.
    We use Lemma~\ref{lem:MCpartion} and~\ref{lem:condNhood} to bound

    \begin{multline*}
         \sum_{\t=1}^\infty \E[R \mid A_\t] \P[A_\t] 
        \le \sum_{\t=1}^\infty \E\bigl[f'(r,\y) \t^{r\y}\sum_{v \in V(\Gn)} \|\Gn[N^{\Gn}_{r}(v)]\| \mid A_\t \bigr] \P[A_\t] \\
        =    \sum_{\t=1}^\infty f'(r,\y) \t^{r\y} \E\bigl[\sum_{v \in V(\Gn)} \|\Gn[N^{\Gn}_{r}(v)]\| \mid A_\t \bigr] \P[A_\t] \\
        \le  \sum_{\t=1}^\infty f'(r,\y) \t^{r\y} (200r\y^3)^{O(r)} \momn^{O(\y^6r^2)}\t^{-\y^2/10} n \\
        =    f'(r,\y) (200r\y^3)^{O(r)} \momn^{O(\y^6r^2)} n \sum_{\t=1}^\infty \t^{-\y^2/10 + r\y}.
    \end{multline*}
    Note that for $\y = \rho(|\varphi|)^2+100$ and $r=\rho(|\varphi|)$
    we have
    $\sum_{\t=1}^\infty \t^{-\y^2/10+r\y}
    \le \sum_{\t=1}^\infty \t^{-2} = O(1)$.
    This yields a run time of $\momn^{f(|\varphi|)}n$
    for some function $f$.
\end{proof}

\begin{theorem}\label{thm:MCpowerlaw2}
        Let $\Gnn$ be a random graph model.
        There exists a function $f$ such that one can solve \FOMCGL in expected time
        \begin{itemize}
        \item $f(|\varphi|)n$                                                              \tabto{4.5cm}  if $\Gnn$ is $\scale$-\dominated for some $\scale > 3$,
        \item $\log(n)^{f(|\varphi|)}n$                                                    \tabto{4.5cm}  if $\Gnn$ is $\scale$-\dominated for $\scale = 3$,
        \item $f(|\varphi|,\varepsilon) n^{1 + \varepsilon}$ for all $\varepsilon > 0$     \tabto{4.5cm}  if $\Gnn$ is $\scale$-\dominated for every $2 < \scale < 3$.
        \end{itemize}
\end{theorem}
\begin{proof}
    Assume that $\Gnn$ is $\scale$-\dominated for some $\scale > 3$.
    By Theorem~\ref{thm:MCpowerlaw1}, there exists a function $f'$ such that
    \FOMCGL can be solved on $\Gnn$ in expected time $ O(1)^{f'(|\varphi|)} n$.
    By choosing $f(|\varphi|)=c^{f'(|\varphi|)}$ for a suitable $c$
    we get the desired expected run time.

    Assume that $\Gnn$ is $\scale$-\dominated with $\scale = 3$.
    By Theorem~\ref{thm:MCpowerlaw1}, there exists a function $f'$ such that
    \FOMCGL can be solved on $\Gnn$ in expected time 
    $\log(n)^{O(1)f'(|\varphi|)} n$.  By choosing
    $f(|\varphi|)=cf'(|\varphi|)$ for a suitable $c$ we get the desired expected
    run time.

    Assume that $\Gnn$ is $\scale$-\dominated for every $\scale < 3$.
    According to Theorem~\ref{thm:MCpowerlaw1}, there exists a function $f'$
    such that one can solve \FOMCGL in expected time
    $O(n^{\varepsilon' f'(|\varphi|)}n)$ for all $\varepsilon' > 0$.
    This means, there exists functions $c(\varepsilon')$ and $n_0(\varepsilon')$ such that for all
    $\varepsilon' > 0$ and $n \ge n_0(\varepsilon)$
    the expected time is at most
    $c(\varepsilon') n^{1+\varepsilon' f'(|\varphi|)}$.
    Thus, we can choose $c'(\varepsilon')$ such that for all
    $\varepsilon' > 0$ and $n \in \N$
    the expected time is at most
    $c'(\varepsilon') n^{1+\varepsilon' f'(|\varphi|)}$.
    Let $\varepsilon > 0$.
    With $\varepsilon'= \varepsilon/f'(|\varphi|)$,
    the algorithm runs for all $n \in \N$ in expected time
    $c'(\varepsilon/f'(|\varphi|)) n^{1+\varepsilon}$.
    We set $f(x,\varepsilon) = c'(\varepsilon/f'(x))$.
    The algorithm runs for all $n \in \N$ in expected time $f(|\varphi|,\varepsilon) n^{1+\varepsilon}$.
\end{proof}


\section{Asymptotic Structural Properties}\label{sec:structure}

In \Cref{sec:partitionprob} and \ref{sec:protrusiondecomp}
we analyzed the structure of $\scale$-\dominated random graphs.
We obtained decompositions depending on parameters $\t$, $r$ and $\y$.
These parameters are needed for algorithmic purposes.
In this section we substitute the parameters $\t$ and $\y$,
which leads to structural results in a more accessible form.

We observe that $\scale$-\dominated random graphs have mostly an extremely sparse structure,
with the exception of a part whose size is bounded by the second order average degree of the degree distribution.
This denser part can be separated well from the remaining graph.
We show that local regions admit a protrusion decomposition consisting of a 
core part, bounded in size by the second order average degree,
to which trees and graphs of constant size are attached.
At first, we define a function $\hatdegree$ similarly to $\momn$ without $O$-notation.
\begin{definition}
    We define
    $$
    \hatdegree = 
    \begin{cases}
        2 &\quad \scale > 3 \\
        \log(n) &\quad \scale = 3 \\
        n^{(3 - \scale)} &\quad \scale < 3.
    \end{cases}
    $$
\end{definition}

We use $\hatdegree$ to obtain a good bound on the minimal value $\t$
such that $\scale$-\dominated graphs are \partitionable.
We fix $\y = 5$ to have one free variable less.

\begin{lemma}\label{lem:aaspartitionable}
    Let $\Gnn$ be an $\scale$-\dominated random graph model.
    There exist constants $c,r_0$ such that
    for every $r \ge r_0$, $\Gnn$ is
    \aas $\hatdegree^{cr^2}$-$r$-$5$-partitionable.
\end{lemma}
\begin{proof}
    Assume $\Gnn$ is $\scale$-\dominated.
    By Theorem~\ref{thm:partitionProbability}, the probability that
    $\Gnn$ is not \partitionable is bouned by at most
    $\momn^{O(\y^{6}r^2)}\t^{-\y^2/10}$.
    Let $\y=5$, $\t=\hatdegree^{cr^2}$. We bound the probability
    of not being $\hatdegree^{cr^2}$-$r$-$5$-partitionable by at most
    $\momn^{O(r^2)}\hatdegree^{-cr^2}$.
    We set $c$ large enough
    such that the probability converges to zero with~$n$.
\end{proof}

Substituting the definition of a \partition into Lemma~\ref{lem:aaspartitionable}
yields the following self-sufficient theorem.

\begin{theorem}\label{thm:aasabc}
    Let $\Gnn$ be an $\scale$-\dominated random graph model.
    There exist constants $c,r_0$ such that for every $r \ge r_0$
    \aas one can
    partition $V(\Gn)$ into three (possibly empty) sets $A$, $B$, $C$ with the following properties.
    \begin{itemize}
        \item
            $|A|,|B| \le \hatdegree^{cr^2}$.
        \item
            Every $r$-neighborhood in $\Gn[B \cup C]$ has at most $25$ more edges than vertices.
        \item
            Every $r$-neighborhood in $\Gn[C]$ has
            at most $5$ edges to $A$.
    \end{itemize}
\end{theorem}
\begin{proof}
    Direct consequence of Lemma~\ref{lem:aaspartitionable}.
\end{proof}

Therefore, one can remove a few vertices to make the graph extremely sparse,
as observed by the following corollary. This corollary might not have
algorithmic consequences by itself, but sheds a lot of light on the
structure of such graphs.

\begin{corollary}\label{aastreewidth}
    Let $\Gnn$ be an $\scale$-\dominated random graph model.
    There exist constants $c,r_0$ such that 
    for every $r \ge r_0$ 
    \aas one can remove $\hatdegree^{cr^2}$ vertices from $\Gn$
    such that every $r$-neighborhood has treewidth at most $26$.
\end{corollary}

In \Cref{sec:protrusiondecomp} we analyze
the local structure of $\scale$-\dominated graphs.
We observe that local regions
consist of a core part, bounded in size by the second order average degree,
to which trees and graphs of constant size are attached.
We obtain a self-contained theorem.\looseness-1

\begin{theorem}\label{thm:aasprotrusion}
    Let $\Gnn$ be an $\scale$-\dominated random graph model.
    There exist constants $c,r_0$ such that for every $r \ge r_0$
    \aas for every $r$-neighborhood $H$ of $\Gn$ one can
    partition $V(H)$ into three (possibly empty) sets $X$, $Y$, $Z$ with the following properties.
    \begin{itemize}
        \item $|X| \le \hatdegree^{cr^2}$.
        \item Every connected component of $H[Y]$ has size at most $cr$
            and at most $c$ neighbors in $X$.
        \item Every connected component of $H[Z]$ is a tree with at most one edge to ${H[X \cup Y]}$.
    \end{itemize}
\end{theorem}
\begin{proof}
    By Lemma~\ref{lem:aaspartitionable},
    $\Gn$ is \aas $\hatdegree^{c'r^2}$-$r$-$5$-partitionable for some constant~$c'$.
    Thus, by Theorem~\ref{thm:partitionIsNhoodpartition},
    every $r$-neighborhood of $\Gn$ is $O\bigl(r^3\hatdegree^{c'r^2}\bigr)$-$r$-$O(5)$-partitionable.
    We refer to Definition~\ref{def:nhoodpartition}
    and choose $c$ large enough such that this statement holds.
\end{proof}
Using Theorem~\ref{thm:aasprotrusion}, we can make statements about
the structural sparsity of a random graph model.
Note that locally bounded treewidth implies nowhere density~\cite{nevsetvril2012sparsity}.
The first corollary is based on the fact that $X$ has a.a.s.\ constant
size if $\alpha>3$.
\begin{corollary}\label{col:localtreewidth}
    Let $\Gnn$ be an $\scale$-\dominated random graph model with $\scale > 3$.
    Then $\Gnn$ has \aas locally bounded treewidth.
\end{corollary}
\begin{corollary}\label{col:cliquesizebound}
    Let $\Gnn$ be an $\scale$-\dominated random graph model.
    There exist constants $c,r_0$ such that for every $r \ge r_0$
    \aas the size of the largest $r$-subdivided clique in $\Gn$
    is at most $\hatdegree^{cr^2}$.
\end{corollary}

\section{Implications for Various Graph Models}\label{sec:domination}

A wide range of unclustered random graph models are $\scale$-\dominated.
In this section, we 
show that 
certain \ErReGrs,
preferential attachment graphs, 
configuration graphs
and
Chung--Lu graphs 
are $\scale$-\dominated
and discuss what implications this has for the tractability
of the model-checking problem on these graph models.
We also discuss the connections to clustered random graph models,
which currently do not fit into our framework.
For convenience, we restate the definition of $\scale$-\domination.
\begin{customdef}{\ref{def:wellBehavedPowerLaw}}
    Let $\scale > 2$.
    We say a random graph model $\Gnn$ is $\scale$-\dominated
    if for every $n \in \N$ there exists
    an ordering $v_1,\dots,v_n$ of $V(\Gn)$
    such that for all $E \subseteq {\{v_1,\dots,v_n \} \choose 2}$
    $$
    \P\bigl[E \subseteq E(\Gn) \bigr] \le \\
    \prod_{v_iv_j \in E} 
    \frac{(n/i)^{1/(\scale-1)}(n/j)^{1/(\scale-1)}}{n}
    \cdot
    \begin{cases}
        2^{O(|E|^2)} &\text{if }\scale > 3 \\
        \log(n)^{O(|E|^2)} &\text{if }\scale = 3 \\
        O(n^\varepsilon)^{|E|^2} \text{ for every $\varepsilon > 0$ }
	&\text{if }\scale < 3.
    \end{cases}
    $$
\end{customdef}

\subsection{Preferential Attachment Model}

The maybe best-known model proposed 
to mimic the features observed in complex networks
are preferential attachment graphs introduced by Barabási and
Albert~\cite{barabasi1999emergence,price1976general}.
They have been studied in great detail (see for example~\cite{hofstad1}).
These random graphs are created by a process that
iteratively adds new vertices
and randomly connects them to already existing ones,
where the attachment probability is proportional to the current degree of a vertex.
The model depends on a constant $m$ which is the number
of edges that are inserted per vertex.
The random graph with $n$ vertices and parameter $m$ is denoted by $G^n_m$.

The preferential attachment process 
exhibits small world behavior~\cite{dommers2010diameters}
and has been widely recognized as a reasonable
explanation of the heavy tailed degree distribution of complex networks~\cite{Bollobas:2001}.

Recent efficient model-checking algorithms on random graph models
only worked on random graph models that asymptotically almost surely (\aas) are
nowhere dense~\cite{grohe2001generalized,StrucSpars}.
It is known that preferential attachment graphs are not \aas nowhere
dense~\cite{StrucSpars} and even \aas somewhere dense~\cite{cliqueminors},
thus previous techniques do not work.

Nevertheless, we are able to solve the model-checking problem
efficiently on these graphs.
Usually, the parameter $m$ of the model is considered to be constant.
We obtain efficient algorithms even if we allow $m$ to be a function of
the size of the network.
For a function $m(n) : \N \to \N$ we define 
$(G^n_{m(n)})_{n \in \N}$
be the corresponding preferential attachment model.
The following lemma follows directly from~\cite{countmotif}.

\begin{lemma}[\cite{countmotif}, Lemma 10]\label{lem:PAdomination}
    Let $m : \N \to \N$. The preferential attachment model $(G^n_{m(n)})_{n \in \N}$ is
    \begin{itemize}
    \item $3$-\dominated                                                    \tabto{8cm} if $m(n)=\log(n)^{O(1)}$,
    \item $\alpha$-\dominated for every $2 < \alpha < 3$                    \tabto{8cm} if $m(n)=O(n^{\varepsilon})$ for every $\varepsilon > 0$.
    \end{itemize}
\end{lemma}

According to Lemma~\ref{lem:PAdomination} and Theorem~\ref{thm:MCpowerlaw2}
one can therefore solve the model-checking problem efficiently 
on preferential attachment graphs.

\begin{corollary}\label{col:prefattmc}
        Let $m \colon \N \to \N$.
        There exists a function $f$ such that one can solve \FOMCGL on 
        the preferential attachment model $(G^n_{m(n)})_{n \in \N}$
        in expected time
        \begin{itemize}
        \item $\log(n)^{f(|\varphi|)}n$                                                     \tabto{5cm} if $m(n)=\log(n)^{O(1)}$,
        \item $f(|\varphi|,\varepsilon) n^{1 + \varepsilon}$ for every $ \varepsilon > 0$   \tabto{5cm} if $m(n)=O(n^{\varepsilon})$ for every $\varepsilon > 0$.
        \end{itemize}
\end{corollary}

\subsection{Chung--Lu Model}

The Chung--Lu model
has been proposed to generate random graphs that fit a certain
degree sequence and has been studied extensively~\cite{chung2002average,chung2002connected,chung2006complex}.
We completely characterize the tractability of the model-checking problem
on Chung--Lu graphs based on the power-law exponent $\scale$ (Corollary~\ref{col:chunglumc}). 
Previous tractability results were obtained for a non-standard
variant of the model and did not cover the case $\scale = 3$.

Let $W = (w_1,\dots,w_n)$ be a sequence of positive weights
with $\max_{i=1}^n w_i^2 \le \sum_{k=1}^n w_k$.
The Chung--Lu random graph to $W$ is a random graph $\Gn$ with
vertices $v_1,\dots,v_n$ such that
each edge $v_iv_j$ with $1 \le i,j \le n$ occurs in $\Gn$ independently 
with probability $w_i w_j / \sum_{k=1}^n w_k$.

Often, the weights are chosen according to a power-law distribution.
Let $\scale > 2$.
We say $\Gnn$ is the \emph{Chung--Lu random graph model with exponent $\scale$}
if for every $n \in \N$, $\Gn$ is the Chung--Lu random graph to $W_n = \{w_1,\dots,w_n\}$
with $w_i = c\cdot (n/i)^{1/(\scale-1)}$ where $c$ is a constant depending on $\scale$~\cite{chung2002average}.
This model nicely matches our concept of $\scale$-\domination.


\begin{lemma}\label{lem:chungludominated}
    Let $\scale > 2$.
    The Chung--Lu random graph model with exponent $\scale$ is $\scale$-\dominated.
\end{lemma}
\begin{proof}
    One can easily verify that $\sum_{k=1}^n w_k = \Theta(n)$ for all $\alpha > 2$.
    Thus, the probability of an edge $v_iv_j$ in a Chung--Lu graph 
    with exponent $\scale > 2$ of size $n$ is
    $$
        w_i w_j / \sum_{k=1}^n w_k = 
        \frac{(n/i)^{1/(\scale-1)}(n/j)^{1/(\scale-1)}}{\Theta(n)}.
    $$
    All edges are independent of each other,
    therefore the probability that an edge set $E$ is contained
    is the product of the probabilities of the individual edges.
    This yields
    $$
    \P\bigl[E \subseteq E(\Gn)\bigr] \le
    2^{O(|E|)}
    \prod_{v_iv_j \in E} 
    \frac{(n/i)^{1/(\scale-1)}(n/j)^{1/(\scale-1)}}{n}.
    $$

\kern-10pt
\end{proof}

We can combine Lemma~\ref{lem:chungludominated},
Theorem~\ref{thm:MCpowerlaw2} and \cite{averagehardness}
to characterize the tractability of the labeled model-checking
problem on Chung--Lu graphs.

\begin{corollary}\label{col:chunglumc}
        Let $\cal G$ be the Chung--Lu random graph model with exponent $\scale$.
        There exists a function $f$ such that one can solve \FOMCGL
        on $\cal G$ in expected time
        \begin{itemize}
        \item $f(|\varphi|)n$           \tabto{5cm} if $\scale > 3$,
        \item $\log(n)^{f(|\varphi|)}n$ \tabto{5cm} if $\scale = 3$.
        \end{itemize}
        Furthermore, if $ 2.5 \le \scale < 3$, $\scale \in \Q$ then one cannot solve 
        \FOMCGL on $\cal G$ in expected {\rm FPT} time unless $\rm AW[*] \subseteq FPT/poly$.
\end{corollary}

Previously, 
the model-checking problem has been known to be tractable on Chung--Lu graphs with exponent $\alpha > 3$,
and hard on Chung--Lu graphs with exponent $2.5 \le \scale < 3$.
The important case $\scale = 3$ was open.
Furthermore, the previous tractability result assumes
the maximum expected degree of a Chung--Lu graph with exponent $\alpha$
to be at most $O(n^{1/\alpha})$, while in the
canonical definition of Chung--Lu graphs (stated above) it is $\Theta(n^{1/(\alpha-1)})$.
Our results hold for the canonical definition.
The missing case $\scale < 2.5$ is still open.
We believe it can be proven to be hard with similar techniques
as for $2.5 \le \scale < 3$.


The second order average degree $\bar d$ of a Chung--Lu graph with weights $w_1,\dots,w_n$
is defined as $\sum_{i=1}^n w_i^2 / \sum_{k=1}^n w_k$.
After substituting the maximum degree $m = \Omega(n^{1/(\scale-1)})$
in~\cite{chung2002average} one can see for the Chung--Lu graph with exponent $\scale$ that
$$
\bar d 
    = \begin{cases}
        \Omega(1) &\quad \scale > 3 \\
        \Omega(\log(n)) &\quad \scale = 3 \\
        \Omega(n^{(3-\scale)/(\scale-1)}) &\quad \scale < 3.
    \end{cases}
$$
We can further bound the run time of the model-checking problem
in terms of $\bar d$.
\begin{lemma}\label{lem:chungludegree}
    There exist a function $f$ such that one can solve \FOMCGL
    on Chung--Lu graphs with exponent $\scale$
    in expected time
    $(c_\scale \bar d)^{f(|\varphi|)} n$,
    where $\bar d$ is the second order average degree and $c_\scale$ is a constant
    depending on $\scale$.
\end{lemma}
\begin{proof}
    According to Theorem \ref{thm:MCpowerlaw1} one can solve \FOMCGL
    on the Chung--Lu graph with exponent $\scale$ in expected time
    $$
        \begin{cases}
            \mu_\scale^{f'(|\varphi|)}n &\quad \scale > 3 \\
            \log(n)^{\mu_\scale f'(|\varphi|)}n &\quad \scale = 3 \\
            \mu_\scale n^{(3-\scale)f'(|\varphi|)}n &\quad \scale < 3
        \end{cases}
    $$

    where $f'$ is some function and $\mu_\scale$ is a constant depending on $\scale$.
    On the other hand, we have
    $$
    \bar d 
        = \begin{cases}
            \lambda_\scale &\quad \scale > 3 \\
            \lambda_\scale \log(n) &\quad \scale = 3 \\
            \lambda_\scale n^{(3-\scale)/(\scale-1)} &\quad \scale < 3
        \end{cases}
    $$
    for another constant $\lambda_\scale$ depending on $\scale$.
    Thus, one can solve \FOMCGL in expected time
    $$
    \bigl((\mu_\scale/\lambda_\scale) \bar d \bigr)^{\max(2,\mu_3) f'(|\varphi|)}n.
    $$
    The result follows by setting $c_\scale = \mu_\scale/\lambda_\scale$
    and $f(|\varphi|) = \max(2,\mu_3) f'(|\varphi|)$.
\end{proof}

\subsection{Configuration Model}

The configuration model has been proposed to
generate random multigraphs whose degrees are fixed~\cite{molloy1995critical,MR98,BENDER1978296}.
We solve the model-checking problem on configuration graphs with
a power-law exponent $3$ (Corollary~\ref{col:configurationmc}).
Previously, this was only known for those configuration graphs
with an exponent strictly larger than $3$~\cite{StrucSpars}.

Let $W = (w_1,\dots,w_n)$ be a degree sequence of a multigraph (i.e., a sequence of positive integers whose sum is even).
The configuration model constructs a random multigraph with $n$ vertices
whose degree sequence is exactly $W$ as follows~\cite{molloy1995critical}:
Let $v_1,\dots,v_n$ be the vertices of the graph.
We form a set $L$ of $w_i$ many distinct copies of $v_i$ for $1 \le i \le n$.
We call the copies of a node $v_i$ in $L$ the \emph{stubs} of $v_i$.
We then construct a random perfect matching on $L$.
This describes a multigraph on $v_1,\dots,v_n$ where the number of edges between two 
vertices equals the number of edges between their stubs.
The degree sequence of this multigraph is exactly~$W$.
Since we only consider simple graphs in this work,
we turn to the so called \emph{erased}~\cite{hofstad1} model.
Here self-loops are removed and multi-edges are replaced with single edges.
As self-loops can be expressed by labels, this is no real limitation for the model-checking problem.
Let $\Gn$ be the probability distribution over simple graphs with $n$ vertices
defined by this process.
We say $\Gn$ is the \emph{random configuration graph} corresponding to $W$.

This defines a random graph with a fixed number of vertices.
In order to define a random graph model we need to define configuration graphs
of arbitrary size.
Let $(w_i(n))_{i \in \N}$ be a sequence of functions such that all $n \in \N$,
$(w_1(n),\dots,w_n(n))$ is a degree sequence of a multigraph.
For $n \in \N$ let $\Gn$ be the random configuration graph corresponding to the degree
sequence $(w_1(n),\dots,w_n(n))$.
We then say $\Gnn$ is the \emph{random configuration graph model} corresponding
to $(w_i(n))_{i \in \N}$.
For technical reasons,
our definition differs slightly from the original one by Molloy and Reed~\cite{molloy1995critical}.

\begin{lemma}
    Let $\Gnn$ be a random configuration graph model with
    corresponding sequence $(w_i(n))_{i \in \N}$.
    Assume there exists a function $p(n)$ with $p(n) = O(n^\varepsilon)$ for all $\varepsilon > 0$
    such that for all $i,n \in \N$, $w_i(n)  \le p(n) \sqrt{n/i}$
    and $\sum_{k=1}^n w_k(n) \ge n/p(n)$.
    Then $\Gnn$ is $3$-\dominated.
\end{lemma}
\begin{proof}
    We consider the configuration model with weight sequence $(w_1(n),\dots,w_n(n))$
    and vertices $v_1,\dots,v_n$.
    Let $E \subseteq {\{v_1,\dots,v_n\} \choose 2}$.
    By Definition~\ref{def:wellBehavedPowerLaw},
    it suffices to show that for every $\varepsilon > 0$
    $$
    \P\bigl[E \subseteq E(\Gn) \bigr] \le \\
    O(n^\varepsilon)^{|E|^2}
    \prod_{v_iv_j \in E} 
    \frac{1}{\sqrt{ij}}.
    $$
    We can assume $|E| \le n^{1/4}$,
    since for $|E| > n^{1/4}$ and every $\varepsilon > 0$ trivially holds
    $$
    \P[ E \subseteq E(\Gn)] \le 1 = O(n^\varepsilon)^{|E|^2} \prod_{v_iv_j \in \E} \frac{1}{n}.
    $$

    As described in~\cite[Lemma 7.6]{hofstad1}, the
    perfect matching of the stubs in the configuration model
    can also be generated by a so-called \emph{adaptive pairing scheme},
    where unmatched stubs are taken one-by-one and matched uniformly to
    the remaining unmatched stubs.
    Assume at most $l$ stubs have been matched already in such a scheme.
    We fix $i,j \in \N$ with $i,j \le n$ and $i \neq j$.
    The probability that a fixed stub of $v_i$ is matched with
    some stub of $v_j$ is at most $w_j(n) / \bigl((\sum_{k=1}^n w_k(n)) - 1 - l\bigr)$.
    By applying the union bound to a pairing scheme which matches the $w_i(n)$ many stubs of $v_i$
    we obtain
    $$
    \P[ v_iv_j \in E(\Gn)] \le \frac{w_i(n)w_j(n)}{(\sum_{k=1}^n w_k(n)) - 1 - l}.
    $$
    Let $d$ be the maximum of $w_1(n),\dots,w_n(n)$.
    We consider an adaptive pairing scheme which iteratively matches the stubs of the vertices in $E$
    and obtain
    $$
    \P[ E \subseteq E(\Gn)] \le \prod_{v_iv_j \in \E} \frac{w_i(n)w_j(n)}{(\sum_{k=1}^n w_k(n)) - 1 - 2|E|d}.
    $$
    Since $d \le p(n)\sqrt{n}$, $|E| \le n^{1/4}$ and 
    $\sum_{k=1}^n w_k(n) \ge n/p(n)$
    we can further bound
    $$
    \frac{w_i(n)w_j(n)}{(\sum_{k=1}^n w_k(n)) - 1 - 2|E|d} 
    = O(p(n)) \frac{w_i(n)w_j(n)}{n}
    = O(p(n)^3) \frac{1}{\sqrt{ij}}.
    $$
    The final result follows from the fact that $p(n) = O(n^\varepsilon)$ for all $\varepsilon > 0$.
\end{proof}

Now the previous lemma together with Theorem~\ref{thm:MCpowerlaw2}
yields an efficient model-checking algorithm for configuration graphs.

\begin{corollary}\label{col:configurationmc}
    Let $\Gnn$ be a random configuration graph model with
    corresponding sequence $(w_i(n))_{i \in \N}$.
    Assume there exists a function $p(n)$ with $p(n) = O(n^\varepsilon)$ for all $\varepsilon > 0$
    such that for all $i,n \in \N$, $w_i(n)  \le p(n) \sqrt{n/i}$
    and $\sum_{k=1}^n w_k(n) \ge n/p(n)$.

    Then there exists a function $f$ such that one can decide \FOMCGL on $\Gnn$ in expected time
    $f(|\varphi|,\varepsilon)n^{1+\varepsilon}$ for every $\varepsilon > 0$.
\end{corollary}

\subsection{\ErRe Model}
One of the earliest and most intensively studied random graphs is the
\ErRe model~\cite{bollobas_2001,erdos}. We say
$G(n,p(n))$ is a random graph with $n$
vertices where each pair of vertices is connected independently uniformly at
random with probability~$p(n)$. 
Many properties of \ErReGrs are well
studied, including but not limited to, threshold phenomena, the sizes of
components, diameter, and length of paths~\cite{bollobas_2001}. 
We classify sparse \ErRe graphs with respect to $\scale$-\domination.

\begin{lemma}\label{lem:ErReDomination}
    \ErReGrs $G(n,p(n))$ are
    \begin{itemize}
    \item $\alpha$-\dominated for every $2 < \alpha$     \tabto{7.5cm} if $p(n)=O(1/n)$,
    \item $3$-\dominated                                 \tabto{7.5cm} if $p(n)=\log(n)^{O(1)}/n$,
    \item $\alpha$-\dominated for every $2 < \alpha < 3$ \tabto{7.5cm} if $p(n)=O(n^{\varepsilon}/n)$ for every $\varepsilon > 0$.
    \end{itemize}
\end{lemma}
\begin{proof}
    The probability of a set of edges $E$ to exist in $G(n,p(n))$ is
   $$
       \P[E \subseteq E(G(n,p(n))] = p(n)^{|E|} \le
       \bigl(np(n)\bigr)^{|E|} \prod_{v_iv_j \in E} 
       \frac{(n/i)^{1/(\scale-1)}(n/j)^{1/(\scale-1)}}{n},
   $$
   since $(n/i)^{1/(\scale-1)} \ge 1$ for all $1 \le i\le n$. 
   The rest follows from Definition~\ref{def:wellBehavedPowerLaw}.
\end{proof}

Using the previous Lemma~\ref{lem:ErReDomination} and Theorem~\ref{thm:MCpowerlaw2}
we obtain a fine grained picture over the tractability
of the model-checking problem on sparse \ErRe graphs.
\begin{corollary}\label{col:erremc}
        There exists a function $f$ such that one can solve \FOMCGL on $G(n,p(n))$ 
        in expected time
        \begin{itemize}
        \item $f(|\varphi|)n$                                                            \tabto{5cm} if $p(n)=O(1/n)$,
        \item $\log(n)^{f(|\varphi|)}n$                                                  \tabto{5cm} if $p(n)=\log(n)^{O(1)}/n$,
        \item $f(|\varphi|,\varepsilon) n^{1 + \varepsilon}$ for every $\varepsilon > 0$ \tabto{5cm} if $p(n)=O(n^{\varepsilon}/n)$ for every $\varepsilon > 0$.
        \end{itemize}
\end{corollary}
The third case has been shown previously by Grohe~\cite{grohe2001generalized}.
Furthermore, under reasonable assumptions (AW[$*$] $\not\subseteq$ FPT/poly) we know that \FOMCGL cannot be decided in expected FPT time on denser \ErRe graphs
with $p(n) = n^\delta/n$ for some $0 < \delta < 1$, $\delta \in \Q$~\cite{averagehardness}.

\subsection{Clustered Models}\label{sec:clusteredappendix}

$\scale$-\dominated random graphs tend to capture unclustered random graphs.
One can show that for the algorithmically tractable values of $\scale$ close to or larger than three
the expected number of triangles is subpolynomial (via union bound over all embeddings as in \cref{lem:boundLL}).
Random models with non-vanishing clustering coefficient, such as the
Kleinberg model~\cite{kleinberg2000small,kleinberg2000navigation}, the
hyperbolic random graph model~\cite{krioukov2010hyperbolic,candellero2016clustering},
or the random intersection graph model
\cite{karonski1999random,rybarczyk2011diameter}
generally have a high expected number of triangles.
This means these models are not $\scale$-\dominated for interesting values of $\scale$ close to three
(they may be for smaller $\scale$).
We shall prove a stronger statement for the random intersection graph model
which is defined as follows.

\begin{definition}[Random Intersection Graph Model, \cite{farrell2015hyperbolicity}]
Fix a positive constant~$\delta$.
Let $B$ be a random bipartite graph on parts of sizes $n$ and $\lfloor n^\delta \rfloor$ with each edge present independently with probability
$n^{-(1+\delta)/2}$.
Let $V$ (the vertices) denote the part of size $n$ and $A$ (the attributes) the part of size $\lfloor n^\delta \rfloor$.
The associated random intersection graph $G(n,\delta)$ is defined on the vertices $V$:
two vertices are connected in $G$ if they share (are in $B$ both adjacent to) at least one attribute in~$A$.
\end{definition}

It has been shown that
$(G(n,\delta))_{n \in \N}$ has \aas bounded expansion~\cite{farrell2015hyperbolicity}
if and only if $\delta > 1$.
Furthermore, if $\delta > 1$,
then one can solve \FOMCGL in expected time $f(|\varphi|)n$~\cite{farrell2015hyperbolicity}.
We now argue that 
intersection graphs nevertheless do not fit into our framework of $\scale$-\domination.

\begin{lemma}
    $(G(n,\delta))_{n \in \N}$ is not $\scale$-\dominated for all values of $\delta$ and~$\scale$.
\end{lemma}
\begin{proof}
    Assume $\delta$ and $\scale$ such that $(G(n,\delta))_{n \in \N}$ is $\scale$-\dominated.
    For a fixed $n$ let the vertices of $G(n,\delta)$ be $v_1,\dots,v_n$, ordered
    as in Definition~\ref{def:wellBehavedPowerLaw}.

    If a fixed set of $k$ vertices shares a common attribute then these vertices form a clique.
    The probability that this happens is at least $n^{-k(1+\delta)/2} \ge n^{-ck}$ for some constant $c$.
    Let $E = {\{v_{n-k},\dots,v_n\} \choose 2}$ be the set of all edges between the last
    $k$ vertices in the ordering.
    By the previous argument, 
    $$
    \P\bigl[ E \subseteq E(G(n,\delta)) \bigr] \ge n^{-ck}.
    $$
    
    By Definition~\ref{def:wellBehavedPowerLaw} there exists a  term $p(n)$ with $p(n) = O(n^\varepsilon)^{|E|^2}$ for every $\varepsilon > 0$
    such that
    $$
        \P\bigl[ E \subseteq E(G(n,\delta)) \bigr] \le
        p(n) 
        \prod_{v_iv_j \in E} 
        \frac{(n/i)^{1/(\scale-1)}(n/j)^{1/(\scale-1)}}{n}.
    $$
    By the definition of $p(n)$, there exists a monotone function $f$ such that $p(n) \le (f(1/\varepsilon)n^\varepsilon)^{|E|^2}$ for every $\varepsilon > 0$.
    By setting $\varepsilon = 1/|E|^2$, we obtain $p(n) \le f(|E|^2)^{|E|^2}n$.
    We consider only edges between the last $k$ vertices, and for $n \ge 2k$ holds
    $ (n/(n-k))^{1/(\scale-1)}(n/(n-k))^{1/(\scale-1)} \le 4 $.
    By assuming $n \ge 2k$ we obtain
    $$
        \P\bigl[ E \subseteq E(G(n,\delta)) \bigr] \le
        f(|E|^2)^{|E|^2}n
        \prod_{v_iv_j \in E} 
        \frac{4}{n}
        \le
        4^{k^2} f(k^4)^{k^4} n^{-{k\choose 2}+1}.
    $$
    Together, this yields 
    $$
    n^{-ck} \le \P\bigl[ E \subseteq E(G(n,\delta)) \bigr] \le 4^{k^2} f(k^4)^{k^4} n^{-{k\choose 2}+1}.
    $$
    We choose $k$ large enough such that $ck < {k\choose 2}-1$.
    Then the previous bound yields a contradiction for sufficiently large $n$.
\end{proof}

\section{Conclusion}

We define $\scale$-\dominated random graphs which generalize many unclustered random graphs models.
We provide a structural decomposition of neighborhoods of these graphs
and use it to obtain a meta-algorithm for deciding first-order properties
in the the preferential attachment-,
\ErRe-, Chung--Lu- and configuration random graph model.

There are various factors to consider when evaluating the practical implications of this result.
The degree distribution of most real world networks is similar to a power-law
distribution with exponent between two and three~\cite{clauset2009power},
but our algorithm is only fast for exponents at least three.
This leaves many real world networks where our algorithm is slow.
However, it has been shown that the model-checking problem (with labels)
becomes hard on these graphs if we assume independently distributed edges~\cite{averagehardness}.

So far, we do not know whether the model-checking problem
is hard or tractable on clustered random graphs.
If a random graph model is $3$-\dominated
then one can show that the expected number of triangles is polylogarithmic (via union bound of all possible embeddings of a triangle).
Therefore, random models with clustering, such as the
Kleinberg model~\cite{kleinberg2000small}, the
hyperbolic random graph model~\cite{krioukov2010hyperbolic,candellero2016clustering},
or the random intersection graph model
\cite{karonski1999random},
which have a high number of triangles 
currently do not fit into our framework (see \Cref{sec:clusteredappendix} for a proof that random intersection graphs are not $\scale$-\dominated for any $\scale$).
This is unfortunate, since clustering is a key aspect of real networks~\cite{watts1998collective}.
In the future, we hope to extend our results to clustered
random graph models.
We observe that some clustered random graph models
can be expressed as \emph{first-order transductions} of $\scale$-\dominated random graph models.
For example the random intersection graph model
is a transduction of a sparse \ErRe graph.
We believe this connection can be used to transfer tractability results
to clustered random graphs.
If we can efficiently compute for a clustered random graph model $\cal G$ a
pre-image of a transduction
that is distributed like an $\scale$-\dominated random graph
then we can efficiently solve \FOMCGL on~$\cal G$.
The same idea is currently being considered for solving the model
checking problem for transductions of sparse graph classes (e.g.\ structurally bounded expansion classes)~\cite{DBLP:conf/lics/GajarskyHOLR16}.

In our algorithm,
we use Gaifman's locality theorem to reduce our problem to $r$-neighbor\-hoods of the input graph.
In this construction the value of $r$ can be exponential in the length of the formula~\cite{gaifman1982local}.
On the other hand, the small world property states that the radius of real networks is rather small.
This means, even for short formulas our neighborhood-based approach may practically be working on the whole graph instead of neighborhoods.
It would be interesting to analyze for which values of $r$
practical protrusion decompositions according to \Cref{thm:aasprotrusion} exist in the real world.

At last,
a big problem with all parameterized model-checking algorithms 
is their large run time dependence on the length of the formula.
Grohe and Frick showed that already on trees
every first-order model-checking algorithm takes worst-case time at least $f(|\varphi|)n$
where $f$ is a non-elementary tower function~\cite{frick2004complexity}.
So far, it is unclear whether this also holds in the average-case setting.
The results presented in this paper have a non-elementary
dependence on the length of the formula.
We are curious whether one can find average-case model-checking 
algorithms with elementary expected FPT run time.
In summary, many more obstacles need to be to overcome to obtain a truly practical general purpose meta-algorithm for complex networks.



\bibliography{references,conferences}

\begin{thebibliography}{10}

\bibitem{albert1999internet}
R{\'e}ka Albert, Hawoong Jeong, and Albert-L{\'a}szl{\'o} Barab{\'a}si.
\newblock Internet: Diameter of the world-wide web.
\newblock {\em Nature}, 401(6749):130, 1999.

\bibitem{arora2009computational}
Sanjeev Arora and Boaz Barak.
\newblock {\em Computational complexity: {A} modern approach}.
\newblock Cambridge University Press, 2009.

\bibitem{barabasi1999emergence}
Albert-L\'aszl\'o Barab\'asi and R\'eka Albert.
\newblock Emergence of scaling in random networks.
\newblock {\em Science}, 286(5439):509--512, 1999.

\bibitem{BENDER1978296}
Edward~A Bender and E.Rodney Canfield.
\newblock The asymptotic number of labeled graphs with given degree sequences.
\newblock {\em Journal of Combinatorial Theory, Series A}, 24(3):296 -- 307,
  1978.
\newblock \href
  {http://dx.doi.org/https://doi.org/10.1016/0097-3165(78)90059-6}
  {\path{doi:https://doi.org/10.1016/0097-3165(78)90059-6}}.

\bibitem{blasius2016hyperbolic}
Thomas Bl{\"a}sius, Tobias Friedrich, and Anton Krohmer.
\newblock Hyperbolic random graphs: Separators and treewidth.
\newblock In {\em 24th Annual European Symposium on Algorithms (ESA 2016)}.
  Schloss Dagstuhl-Leibniz-Zentrum fuer Informatik, 2016.

\bibitem{bodlaender2016meta}
Hans~L Bodlaender, Fedor~V Fomin, Daniel Lokshtanov, Eelko Penninkx, Saket
  Saurabh, and Dimitrios~M Thilikos.
\newblock ({M}eta) kernelization.
\newblock {\em Journal of the ACM (JACM)}, 63(5):44, 2016.

\bibitem{bogdanov2006average}
Andrej Bogdanov and Luca Trevisan.
\newblock {Average-Case Complexity}.
\newblock {\em Foundations and Trends in Theoretical Computer Science},
  2(1):1--106, 2006.

\bibitem{Bollobas:2001}
B{\'e}la Bollob\'{a}s, Oliver Riordan, Joel Spencer, and G\'{a}bor Tusn\'{a}dy.
\newblock The degree sequence of a scale-free random graph process.
\newblock {\em Random Structures \& Algorithms}, 18(3):279--290, May 2001.

\bibitem{bollobas_2001}
Béla Bollobás.
\newblock {\em Random Graphs}.
\newblock Cambridge University Press, 2nd edition, 2001.

\bibitem{broido2018scale}
Anna~D. Broido and Aaron Clauset.
\newblock Scale-free networks are rare.
\newblock {\em Nature communications}, 10(1):1017, 2019.

\bibitem{candellero2016clustering}
Elisabetta Candellero and Nikolaos Fountoulakis.
\newblock Clustering and the hyperbolic geometry of complex networks.
\newblock {\em Internet Mathematics}, 12(1-2):2--53, 2016.

\bibitem{chung2002average}
Fan Chung and Linyuan Lu.
\newblock The average distances in random graphs with given expected degrees.
\newblock {\em Proc. of the National Academy of Sciences}, 99(25):15879--15882,
  2002.

\bibitem{chung2002connected}
Fan Chung and Linyuan Lu.
\newblock Connected components in random graphs with given expected degree
  sequences.
\newblock {\em Annals of Combinatorics}, 6(2):125--145, 2002.

\bibitem{chung2006complex}
Fan Chung and Linyuan Lu.
\newblock {\em Complex graphs and networks}, volume 107.
\newblock American Math. Soc., 2006.

\bibitem{clauset2009power}
Aaron Clauset, Cosma~Rohilla Shalizi, and Mark E.~J. Newman.
\newblock {Power-Law Distributions in Empirical Data}.
\newblock {\em SIAM Review}, 51(4):661--703, 2009.

\bibitem{Cou90}
Bruno Courcelle.
\newblock The monadic second-order logic of graphs {I}. {R}ecognizable sets of
  finite graphs.
\newblock {\em Information and Computation}, 85(1):12--75, 1990.

\bibitem{CourcelleMR2000}
Bruno Courcelle, Johann~A. Makowsky, and Udi Rotics.
\newblock Linear time solvable optimization problems on graphs of bounded
  clique-width.
\newblock {\em Theory Comput. Syst.}, 33(2):125--150, 2000.
\newblock \href {http://dx.doi.org/10.1007/s002249910009}
  {\path{doi:10.1007/s002249910009}}.

\bibitem{cygan2015parameterized}
Marek Cygan, Fedor~V. Fomin, Lukasz Kowalik, Daniel Lokshtanov, D{\'{a}}niel
  Marx, Marcin Pilipczuk, Michal Pilipczuk, and Saket Saurabh.
\newblock {\em Parameterized Algorithms}.
\newblock Springer, 2015.
\newblock \href {http://dx.doi.org/10.1007/978-3-319-21275-3}
  {\path{doi:10.1007/978-3-319-21275-3}}.

\bibitem{dawar2007locally}
Anuj Dawar, Martin Grohe, and Stephan Kreutzer.
\newblock {Locally Excluding a Minor}.
\newblock In {\em Proceedings of the 22nd Symposium on Logic in Computer
  Science}, pages 270--279, 2007.

\bibitem{Demaine:2005:SPA:1101821.1101823}
Erik~D. Demaine, Fedor~V. Fomin, Mohammadtaghi Hajiaghayi, and Dimitrios~M.
  Thilikos.
\newblock Subexponential parameterized algorithms on bounded-genus graphs and
  {$H$}-minor-free graphs.
\newblock {\em J. ACM}, 52(6):866--893, November 2005.
\newblock \href {http://dx.doi.org/10.1145/1101821.1101823}
  {\path{doi:10.1145/1101821.1101823}}.

\bibitem{DH08}
Erik~D. Demaine and M.~Hajiaghayi.
\newblock The bidimensionality theory and its algorithmic applications.
\newblock {\em Comput. J.}, 51(3):292--302, 2008.

\bibitem{StrucSpars}
Erik~D. Demaine, Felix Reidl, Peter Rossmanith, Fernando~S{\'{a}}nchez
  Villaamil, Somnath Sikdar, and Blair~D. Sullivan.
\newblock Structural sparsity of complex networks: {B}ounded expansion in
  random models and real-world graphs.
\newblock {\em J. Comput. Syst. Sci.}, 105:199--241, 2019.
\newblock \href {http://dx.doi.org/10.1016/j.jcss.2019.05.004}
  {\path{doi:10.1016/j.jcss.2019.05.004}}.

\bibitem{diestel}
R.~Diestel.
\newblock {\em Graph Theory}.
\newblock Springer, Heidelberg, 2010.

\bibitem{dommers2010diameters}
Sander Dommers, Remco van~der Hofstad, and Gerard Hooghiemstra.
\newblock Diameters in preferential attachment models.
\newblock {\em Journal of Statistical Physics}, 139(1):72--107, 2010.

\bibitem{downey1996parameterized}
Rod~G. Downey, Michael~R. Fellows, and Udayan Taylor.
\newblock {The Parameterized Complexity of Relational Database Queries and an
  Improved Characterization of W[1]}.
\newblock {\em DMTCS}, 96:194--213, 1996.

\bibitem{cliqueminors}
Jan Dreier, Philipp Kuinke, and Peter Rossmanith.
\newblock Maximum shallow clique minors in preferential attachment graphs have
  polylogarithmic size.
\newblock In {\em Approximation, Randomization, and Combinatorial Optimization.
  Algorithms and Techniques (APPROX/RANDOM)}, volume 176 of {\em LIPIcs}.
  Schloss Dagstuhl - Leibniz-Zentrum f{\"{u}}r Informatik, 2020.

\bibitem{averagehardness}
Jan Dreier and Peter Rossmanith.
\newblock Hardness of {FO} model-checking on random graphs.
\newblock In {\em 14th International Symposium on Parameterized and Exact
  Computation, {IPEC} 2019, September 11-13, 2019, Munich, Germany}, volume 148
  of {\em LIPIcs}, pages 11:1--11:15. Schloss Dagstuhl - Leibniz-Zentrum
  f{\"{u}}r Informatik, 2019.
\newblock \href {http://dx.doi.org/10.4230/LIPIcs.IPEC.2019.11}
  {\path{doi:10.4230/LIPIcs.IPEC.2019.11}}.

\bibitem{countmotif}
Jan Dreier and Peter Rossmanith.
\newblock Motif counting in preferential attachment graphs.
\newblock In {\em 39th {IARCS} Annual Conference on Foundations of Software
  Technology and Theoretical Computer Science, {FSTTCS} 2019, December 11-13,
  2019, Bombay, India}, volume 150 of {\em LIPIcs}, pages 13:1--13:14. Schloss
  Dagstuhl - Leibniz-Zentrum f{\"{u}}r Informatik, 2019.
\newblock \href {http://dx.doi.org/10.4230/LIPIcs.FSTTCS.2019.13}
  {\path{doi:10.4230/LIPIcs.FSTTCS.2019.13}}.

\bibitem{dvorak2010deciding}
Zdenek Dvo\v{r}ak, Daniel Kr{\'a}l, and Robin Thomas.
\newblock {Deciding First-Order Properties for Sparse Graphs}.
\newblock In {\em Proceedings of the 51st Conference on Foundations of Computer
  Science}, pages 133--142, 2010.

\bibitem{erdos}
P.~Erd\H{o}s and A.~R\'{e}nyi.
\newblock On random graphs.
\newblock {\em Publicationes Mathematicae}, 6:290--297, 1959.

\bibitem{fagin1976probabilities}
Ronald Fagin.
\newblock Probabilities on finite models 1.
\newblock {\em The Journal of Symbolic Logic}, 41(1):50--58, 1976.

\bibitem{farrell2015hyperbolicity}
Matthew Farrell, Timothy~D Goodrich, Nathan Lemons, Felix Reidl,
  Fernando~S{\'a}nchez Villaamil, and Blair~D Sullivan.
\newblock Hyperbolicity, degeneracy, and expansion of random intersection
  graphs.
\newblock In {\em International Workshop on Algorithms and Models for the
  Web-Graph}, pages 29--41. Springer, 2015.

\bibitem{flum2002query}
J{\"o}rg Flum, Markus Frick, and Martin Grohe.
\newblock {Query Evaluation via Tree-Decompositions}.
\newblock {\em Journal of the ACM (JACM)}, 49(6):716--752, 2002.

\bibitem{flum2001fixed}
J{\"o}rg Flum and Martin Grohe.
\newblock {Fixed-Parameter Tractability, Definability, and Model-Checking}.
\newblock {\em SIAM Journal on Computing}, 31(1):113--145, 2001.

\bibitem{FLST2010}
Fedor~V Fomin, Daniel Lokshtanov, Saket Saurabh, and Dimitrios~M Thilikos.
\newblock Bidimensionality and kernels.
\newblock In {\em Proc. of the Twenty-First Annual ACM-SIAM Symposium on
  Discrete Algorithms}, pages 503--510, 2010.

\bibitem{frick2001deciding}
Markus Frick and Martin Grohe.
\newblock Deciding first-order properties of locally tree-decomposable
  structures.
\newblock {\em Journal of the ACM (JACM)}, 48(6):1184--1206, 2001.

\bibitem{frick2004complexity}
Markus Frick and Martin Grohe.
\newblock The complexity of first-order and monadic second-order logic
  revisited.
\newblock {\em Annals of pure and applied logic}, 130(1-3):3--31, 2004.

\bibitem{gaifman1982local}
Haim Gaifman.
\newblock On local and non-local properties.
\newblock In {\em Studies in Logic and the Foundations of Mathematics}, volume
  107, pages 105--135. Elsevier, 1982.

\bibitem{DBLP:conf/lics/GajarskyHOLR16}
Jakub Gajarsk{\'{y}}, Petr Hlin{\v{e}}n{\`y}, Jan Obdrz{\'{a}}lek, Daniel
  Lokshtanov, and M.~S. Ramanujan.
\newblock A new perspective on {FO} model checking of dense graph classes.
\newblock In {\em Proceedings of the 31st Annual {ACM/IEEE} Symposium on Logic
  in Computer Science, {LICS} '16, New York, NY, USA, July 5-8, 2016}, pages
  176--184, 2016.
\newblock \href {http://dx.doi.org/10.1145/2933575.2935314}
  {\path{doi:10.1145/2933575.2935314}}.

\bibitem{gao2012treewidth}
Yong Gao.
\newblock Treewidth of {E}rd{\H{o}}s--{R}{\'e}nyi random graphs, random
  intersection graphs, and scale-free random graphs.
\newblock {\em Discrete Applied Mathematics}, 160(4-5):566--578, 2012.

\bibitem{glebskii1969range}
Yu~V Glebskii, DI~Kogan, MI~Liogon'kii, and VA~Talanov.
\newblock Range and degree of realizability of formulas in the restricted
  predicate calculus.
\newblock {\em Cybernetics and Systems Analysis}, 5(2):142--154, 1969.

\bibitem{goldenberg2010survey}
Anna Goldenberg, Alice~X. Zheng, Stephen~E. Fienberg, Edoardo~M. Airoldi,
  et~al.
\newblock A survey of statistical network models.
\newblock {\em Foundations and Trends in Machine Learning}, 2(2):129--233,
  2010.

\bibitem{grohe2001generalized}
Martin Grohe.
\newblock Generalized model-checking problems for first-order logic.
\newblock In {\em Annual Symposium on Theoretical Aspects of Computer Science},
  pages 12--26. Springer, 2001.

\bibitem{grohe2008logic}
Martin Grohe.
\newblock Logic, graphs, and algorithms.
\newblock {\em Logic and Automata}, 2:357--422, 2008.

\bibitem{grohe2017deciding}
Martin Grohe, Stephan Kreutzer, and Sebastian Siebertz.
\newblock Deciding first-order properties of nowhere dense graphs.
\newblock {\em Journal of the ACM}, 64(3):17, 2017.

\bibitem{karonski1999random}
Micha{\l} Karo{\'n}ski, Edward~R. Scheinerman, and Karen~B. Singer-Cohen.
\newblock On random intersection graphs: The subgraph problem.
\newblock {\em Combinatorics, Probability and Computing}, 8(1-2):131--159,
  1999.

\bibitem{FV59}
Carol Karp.
\newblock The first order properties of products of algebraic systems.
  fundamenta mathematicae.
\newblock {\em Journal of Symbolic Logic}, 32(2):276–276, 1967.
\newblock \href {http://dx.doi.org/10.2307/2271704}
  {\path{doi:10.2307/2271704}}.

\bibitem{kim2016linear}
Eun~Jung Kim, Alexander Langer, Christophe Paul, Felix Reidl, Peter Rossmanith,
  Ignasi Sau, and Somnath Sikdar.
\newblock Linear kernels and single-exponential algorithms via protrusion
  decompositions.
\newblock {\em ACM Transactions on Algorithms (TALG)}, 12(2):21, 2016.

\bibitem{kleinberg2000small}
Jon Kleinberg.
\newblock {The Small-World Phenomenon: An Algorithmic Perspective}.
\newblock In {\em Proceedings of the 32nd Symposium on Theory of Computing},
  pages 163--170, 2000.

\bibitem{kleinberg2000navigation}
Jon~M. Kleinberg.
\newblock Navigation in a small world.
\newblock {\em Nature}, 406(6798):845--845, 2000.

\bibitem{kreutzer2008algorithmic}
Stephan Kreutzer.
\newblock Algorithmic meta-theorems.
\newblock In {\em International Workshop on Parameterized and Exact
  Computation}, pages 10--12. Springer, 2008.

\bibitem{krioukov2010hyperbolic}
Dmitri Krioukov, Fragkiskos Papadopoulos, Maksim Kitsak, Amin Vahdat, and
  Mari{\'a}n Bogun{\'a}.
\newblock Hyperbolic geometry of complex networks.
\newblock {\em Physical Review E}, 82(3):036106, 2010.

\bibitem{levin1986average}
Leonid~A. Levin.
\newblock Average case complete problems.
\newblock {\em SIAM Journal on Computing}, 15(1):285--286, 1986.

\bibitem{makowsky2004algorithmic}
Johann~A. Makowsky.
\newblock Algorithmic uses of the feferman--vaught theorem.
\newblock {\em Annals of Pure and Applied Logic}, 126(1-3):159--213, 2004.

\bibitem{milgram1967small}
Stanley Milgram.
\newblock The small world problem.
\newblock {\em Psychology Today}, 2(1):60--67, 1967.

\bibitem{milo2002network}
Ron Milo, Shai Shen-Orr, Shalev Itzkovitz, Nadav Kashtan, Dmitri Chklovskii,
  and Uri Alon.
\newblock Network motifs: simple building blocks of complex networks.
\newblock {\em Science}, 298(5594):824--827, 2002.

\bibitem{mislove2007measurement}
Alan Mislove, Massimiliano Marcon, Krishna~P Gummadi, Peter Druschel, and Bobby
  Bhattacharjee.
\newblock Measurement and analysis of online social networks.
\newblock In {\em Proc. of the 7th ACM SIGCOMM Conference on Internet
  Measurement}, pages 29--42. ACM, 2007.

\bibitem{MR98}
M.~Molloy and B.~A. Reed.
\newblock The size of the giant component of a random graph with a given degree
  sequence.
\newblock {\em Combin., Probab. Comput.}, 7(3):295--305, 1998.

\bibitem{molloy1995critical}
Michael Molloy and Bruce Reed.
\newblock A critical point for random graphs with a given degree sequence.
\newblock {\em Random Structures \& Algorithms}, 6(2-3):161--180, 1995.

\bibitem{RS03}
Paul D.~Seymour N.~Robertson.
\newblock Graph minors {XVI}. {E}xcluding a non-planar graph.
\newblock {\em Journal of Combinatorial Theory, Series B}, 89:43--76, 2003.

\bibitem{nevsetvril2012sparsity}
Jaroslav Ne{\v{s}}et{\v{r}}il and Patrice Ossona~de Mendez.
\newblock {\em Sparsity}.
\newblock Springer, 2012.

\bibitem{NOdM08}
Jaroslav Ne\v{s}et\v{r}il and Patrice Ossona~de Mendez.
\newblock Grad and classes with bounded expansion {I}. {D}ecompositions.
\newblock {\em European Journal of Combinatorics}, 29(3):760--776, 2008.

\bibitem{price1976general}
Derek de~Solla Price.
\newblock A general theory of bibliometric and other cumulative advantage
  processes.
\newblock {\em Journal of the American society for Information science},
  27(5):292--306, 1976.

\bibitem{prvzulj2007biological}
Nata{\v{s}}a Pr{\v{z}}ulj.
\newblock Biological network comparison using graphlet degree distribution.
\newblock {\em Bioinformatics}, 23(2):e177--e183, 2007.

\bibitem{rybarczyk2011diameter}
Katarzyna Rybarczyk.
\newblock Diameter, connectivity, and phase transition of the uniform random
  intersection graph.
\newblock {\em Discrete Mathematics}, 311(17):1998--2019, 2011.

\bibitem{schaeffer2007graph}
Satu~Elisa Schaeffer.
\newblock Graph clustering.
\newblock {\em Computer Science Review}, 1(1):27--64, 2007.

\bibitem{DBLP:conf/pods/SchweikardtSV18}
Nicole Schweikardt, Luc Segoufin, and Alexandre Vigny.
\newblock Enumeration for {FO} queries over nowhere dense graphs.
\newblock In {\em Proceedings of the 37th {ACM} {SIGMOD-SIGACT-SIGAI} Symposium
  on Principles of Database Systems, Houston, TX, USA, June 10-15, 2018}, pages
  151--163. {ACM}, 2018.
\newblock \href {http://dx.doi.org/10.1145/3196959.3196971}
  {\path{doi:10.1145/3196959.3196971}}.

\bibitem{seese1996linear}
Detlef Seese.
\newblock {Linear time computable problems and first-order descriptions}.
\newblock {\em Math. Struct. in Comp. Science}, 6:505--526, 1996.

\bibitem{spencer2013strange}
Joel Spencer.
\newblock {\em The strange logic of random graphs}, volume~22.
\newblock Springer Science \& Business Media, 2013.

\bibitem{stockmeyer1976polynomial}
Larry~J. Stockmeyer.
\newblock The polynomial-time hierarchy.
\newblock {\em Theoretical Computer Science}, 3(1):1--22, 1976.

\bibitem{hofstad1}
Remco van~der Hofstad.
\newblock {\em Random graphs and complex networks}, volume~1.
\newblock Cambridge University Press, 2016.

\bibitem{watts1998collective}
Duncan~J. Watts and Steven~H. Strogatz.
\newblock Collective dynamics of ‘small-world’networks.
\newblock {\em nature}, 393(6684):440, 1998.

\end{thebibliography}



\end{document}